\documentclass[11pt]{article}

\usepackage{xcolor}
\usepackage{amsfonts,amsmath,amssymb,mathtools,amsthm}
\usepackage{graphicx}
\usepackage{stmaryrd}
\usepackage{thmtools}
\DeclareGraphicsExtensions{.eps}

\renewcommand{\theequation}{\thesection.\arabic{equation}}

\newcommand\encadremath[1]{\vbox{\hrule\hbox{\vrule\kern8pt
\vbox{\kern8pt \hbox{$\displaystyle #1$}\kern8pt}
\kern8pt\vrule}\hrule}}
\def\enca#1{\vbox{\hrule\hbox{
\vrule\kern8pt\vbox{\kern8pt \hbox{$\displaystyle #1$}
\kern8pt} \kern8pt\vrule}\hrule}}

\newcommand\framefig[1]{
\begin{figure}[bth]
\hrule\hbox{\vrule\kern8pt
\vbox{\kern8pt \vbox{
\begin{center}
{#1}
\end{center}
}\kern8pt}
\kern8pt\vrule}\hrule
\end{figure}
}

\newcommand\figureframex[3]{
\begin{figure}[bth]
\hrule\hbox{\vrule\kern8pt
\vbox{\kern8pt \vbox{
\begin{center}
{\mbox{\epsfxsize=#1.truecm\epsfbox{#2}}}
\end{center}
\caption{#3}
}\kern8pt}
\kern8pt\vrule}\hrule
\end{figure}
}
\newcommand\figureframey[3]{
\begin{figure}[bth]
\hrule\hbox{\vrule\kern8pt
\vbox{\kern8pt \vbox{
\begin{center}
{\mbox{\epsfysize=#1.truecm\epsfbox{#2}}}
\end{center}
\caption{#3}
}\kern8pt}
\kern8pt\vrule}\hrule
\end{figure}
}

\renewcommand{\thesection}{\arabic{section}}
\renewcommand{\theequation}{\arabic{section}-\arabic{equation}}
\makeatletter
\@addtoreset{equation}{section}
\makeatother
\newtheorem{theorem}{Theorem}[section]

\newtheorem{proposition}{Proposition}[section]
\newtheorem{lemma}{Lemma}[section]
\newtheorem{corollary}{Corollary}[section]

\theoremstyle{definition}
\newtheorem{remark}{Remark}[section]
\newtheorem{definition}{Definition}[section]

\def\br{\begin{remark}\rm\small}
\def\er{\end{remark}}
\def\bt{\begin{theorem}}
\def\et{\end{theorem}}
\def\bd{\begin{definition}}
\def\ed{\end{definition}}
\def\bp{\begin{proposition}}
\def\ep{\end{proposition}}
\def\bl{\begin{lemma}}
\def\el{\end{lemma}}
\def\bc{\begin{corollary}}
\def\ec{\end{corollary}}
\def\beaq{\begin{eqnarray}}
\def\eeaq{\end{eqnarray}}

\theoremstyle{definition}

\newcommand{\be}{\begin{equation}}
\newcommand{\ee}{\end{equation}}
\newcommand{\beq}{\begin{equation}}
\newcommand{\eeq}{\end{equation}}
\newcommand{\bea}{\begin{eqnarray}}
\newcommand{\eea}{\end{eqnarray}}
\newcommand{\beqq}{\begin{equation*}}
\newcommand{\eeqq}{\end{equation*}}
\newcommand{\beaa}{\begin{eqnarray*}}
\newcommand{\eeaa}{\end{eqnarray*}}

\newcommand{\diag}{{\operatorname{diag}}}

\newcommand{\td}{\tilde}

\newcommand\blfootnote[1]{%
  \begingroup
  \renewcommand\thefootnote{}\footnote{#1}%
  \addtocounter{footnote}{-1}%
  \endgroup
}

\newcommand{\Res}{\mathop{\,\rm Res\,}}
\usepackage[pdftex]{hyperref}
\hypersetup{colorlinks,urlcolor=magenta,citecolor=red,linkcolor=blue,filecolor=black}

\title{\bf{Isomonodromic and isospectral deformations of meromorphic connections: the $\mathfrak{sl}_2(\mathbb{C})$ case}}

\date{\vspace{-5ex}}
 
\author{$_{1}$Olivier Marchal\footnote{Universit\'{e} Jean Monnet Saint-\'{E}tienne, CNRS, Institut Camille Jordan UMR 5208, Institut Universitaire de France, F-42023, Saint-\'{E}tienne, France}\,\,,
$_{2}$Mohamad Alameddine\footnote{Universit\'{e} Jean Monnet Saint-\'{E}tienne, CNRS, Institut Camille Jordan UMR 5208, F-42023, Saint-\'{E}tienne, France.}
}

\begin{document}

\maketitle

\vspace{1.0cm}

\begin{abstract}
We consider non-twisted meromorphic connections in $\mathfrak{sl}_2(\mathbb{C})$ and the associated symplectic Hamiltonian structure. In particular, we provide explicit expressions of the Lax pair in the geometric gauge supplementing results of \cite{MarchalOrantinAlameddine2022} where explicit formulas have been obtained in the oper gauge. Expressing the geometric Lax matrices requires the introduction of specific Darboux coordinates for which we provide the explicit Hamiltonian evolutions. These expressions allow to build bridges between the isomonodromic deformations and the isospectral ones. More specifically, we propose an explicit change of Darboux coordinates to obtain isospectral coordinates for which Hamiltonians match the spectral invariants. This result solves the issue left opened in \cite{BertolaHarnadHurtubise2022} in the case of $\mathfrak{sl}_2(\mathbb{C})$.

\blfootnote{\textit{Email addresses:}$_{1}$\textsf{olivier.marchal@univ-st-etienne.fr}, $_{2}$\textsf{mohamad.alameddine@univ-st-etienne.fr}}
\end{abstract}



\tableofcontents

\newpage

\section{Introduction}
Isomonodromic deformations of meromorphic connections have been studied for a long time \cite{Picard,Garnier,Painleve,schlesinger1912klasse,JimboMiwa,JimboMiwaUeno} in relation with the Painlev\'{e} property. The main feature of isomonodromic deformations is that they contain an underlying Hamiltonian symplectic structure associated to the Darboux coordinates describing the Lax matrices. This property has been extensively described and proved for a very large class of meromorphic connections \cite{Darboux_coord93,HarnadHurtubise1997,HURTUBISE20081394,BertolaHarnadHurtubise2022,SpectralCurves_BenZviFrenkel_2001} including those studied in the present article. In particular, the case of Fuchsian singularities (i.e. simple poles) dates back to the works of R. Fuchs \cite{Fuchs} and B. Gambier \cite{Gambier} giving rise to the Lax formulation of the Painlev\'{e} $6$ equation. This point of view was pursued by R. Garnier and K. Okamoto \cite{Okamoto1986Iso,Okamoto1986} for arbitrary Fuchsian systems in the rank-$2$ setting i.e. to Garnier systems in their scalar version. The case of connections with Fuchsian singularities and their underlying Hamiltonian systems are now documented and thoroughly understood (see Chap. $3$ of \cite{FromGaussToPainleve} for a review on the subject). 
On the contrary, the case of irregular singularities (i.e. meromorphic connections with poles of arbitrary orders) is much more difficult to handle and many issues remain opened and is still generating active research, dating back from the works of the Japenese school of Jimbo--Miwa--Ueno \cite{JimboMiwaUeno,JimboMiwa}, to the most recent generalizations of isomonodromy equations as extensions of the encoded symplectic structure \cite{Boalch2001,BertolaKorotkin2021,Chiba,PainleveMonodromyManifolds,SpectralCurves_BenZviFrenkel_2001}. There are several complications in the case of irregular singularities. Indeed, the underlying geometry of the moduli space is much more complex to describe since one needs to consider generalized monodromy data (who are proved to form Poisson local systems \cite{boalch2012geometry,Douccot2022TopologyOI,Douccot2022LocalWM}) in addition to the location of poles. Nevertheless, there have been recent developments in the understanding of the natural geometric parameters \cite{Boalch2001,Boalch2012}. Nowadays, the existence of an underlying Hamiltonian symplectic structure is well understood in the case of irregular singularities, however, getting some general formulas to express the Hamiltonian remains a difficult question.

\medskip

There are several strategies to tackle irregular singularities, a first option is to start with the well-known case of simple poles and make confluences of these simple poles to generate poles of higher orders. This strategy has been used in many papers \cite{Babujian:1997xg,mazzocco2007hamiltonian,Chernyakov:2003gy} with very interesting results developed in \cite{MartaPaper2022}. A second strategy is simply to write down the compatibility equations of the Lax system (Cf. \eqref{LaxCompat}), which corresponds to a zero-curvature equation for connections and try to solve it to get the complete expression for the Lax matrices and the Hamiltonian system. This strategy was partly solved in \cite{MarchalOrantinAlameddine2022} for meromorphic connections in $\mathfrak{sl}_2(\mathbb{C})$ on the Riemann sphere and the authors discussed the relation with the confluence approach developed in \cite{MartaPaper2022}. In particular, in this article, the authors provided explicit formulas for the Lax pairs and the Hamiltonian in the oper gauge (i.e. the gauge where the Lax matrix is companion-like) for non-twisted meromorphic connections with poles of arbitrary orders. The main limitation of these results is that they express the Lax pair in the oper gauge and obtain the Hamiltonian structure relatively to coordinates adapted to this gauge. However, the modern approach of isomonodromic deformations using meromorphic connections with fixed pole structure does not make the oper gauge the most natural geometric gauge because of the presence of apparent singularities in the Lax matrices. Solving this issue by defining appropriate Darboux coordinates and obtaining explicit expressions in the initial geometric gauge was one of the first motivations for this article that we achieved in Theorems \ref{GeoLaxMatrices}, \ref{TheoHamNew} and \ref{GeoLaxMatricesQR}.

The third historical strategy to tackle isomonodromic deformations by considering first isospectral deformations (i.e. deformations that preserve the spectrum of the Lax matrix) and then choose some specific Darboux coordinates so that an additional ``isospectral condition''
\beq \label{Condition} \delta_{\mathbf{t}} \td{L}(\lambda)=\partial_\lambda \td{A}(\lambda)\eeq
is realized \cite{Chiba,BertolaHarnadHurtubise2022,Yamakawa2017TauFA,Yamakawa2019FundamentalTwoForms}. The main feature of the additional isospectral condition is that when it is realized the Hamiltonian relatively to the Darboux coordinates identify with the spectral invariants of the Lax matrix (i.e. the coefficients of the expansion of the eigenvalues of the Lax matrix at each pole). This isospectral strategy has been used successfully for many cases such as Fuchsian singularities and all Painlev\'{e} cases and the general theory has been set up in \cite{BertolaHarnadHurtubise2022} by the Montr\'{e}al school. Let us mention here that the isospectral condition puts strong constraints on the explicit dependence of the Darboux coordinates relatively to the deformation parameters. For Fuchsian singularities, it turns out that the additional isospectral condition is immediate to solve and that the explicit dependence of the Darboux coordinates relatively to the deformation parameters is trivial. However, this observation does not hold for non-Fuchsian singularities and if one can find the explicit dependence of the Darboux coordinates relatively to deformation parameters on a case by case basis for meromorphic connections with singular poles of low orders, the authors of \cite{BertolaHarnadHurtubise2022} left the issue of finding Darboux coordinates solving the isospectral condition in a general setting opened. The second motivation of the present article was thus to give an explicit description of the Darboux coordinates solving the isospectral condition in the case of meromorphic connections in $\mathfrak{sl}_2(\mathbb{C})$ leading to Theorem \ref{MainTheoIsospectral}. As expected, the explicit dependence of the Darboux coordinates relatively to deformation parameters is more complicated for non-Fuchsian singularities. The example of the second element of the Painlev\'{e} $2$ hierarchy (studied also in Appendix A of \cite{BertolaHarnadHurtubise2022}) is done in details in Appendix \ref{AppendixP2}.

\medskip

The present work as well as \cite{MarchalOrantinAlameddine2022,BertolaHarnadHurtubise2022} are in close relations with the works of D. Yamakawa \cite{Yamakawa2017TauFA,Yamakawa2019FundamentalTwoForms}. Indeed, Yamakawa defined the isospectral Hamiltonians using the standard intrinsic residue-trace formulas at each pole. He then imposes some conditions on the exterior derivatives relatively to irregular times to implicitly select Darboux coordinates for which the isospectral Hamiltonians match the isomonodromic ones. This strategy has also been used in \cite{Darboux_coord93,Chiba} and is very efficient to obtain global properties of the symplectic Ehresmann connection such as flatness or completeness. However, it does not provide explicit Darboux coordinates so that one could not obtain trivialization of the symplectic bundle apart from simple examples. Nevertheless, our final conclusion is that solving directly the isomonodromic compatibility equations or solving the isospectral condition are equivalent problems that share the same level of difficulty for $\mathfrak{sl_2}(\mathbb{C})$. This can be seen directly on the fundamental two-form $\Omega$ defined by Yamakawa underlying the symplectic structure. Indeed, this fundamental two-form is defined as
\beqq \Omega=\sum_j dr_j\wedge ds_j -\sum_{k} dH_k \wedge dt_k\eeqq
where $\mathbf{t}:=(t_k)_{k\geq 1}$ are the times, $\mathbf{H}:=(H_k(\mathbf{r},\mathbf{s},\mathbf{t}))_{k\geq 1}$ are the Hamiltonians associated to any set of Darboux coordinates $(\mathbf{r},\mathbf{s})$ that is chosen to parametrize the Lax matrix. Hamiltonians $\mathbf{H}(\mathbf{r},\mathbf{s},\mathbf{t})$ are determined by solving the isomonodromy equations (compatibility equations of the Lax pairs). The strategy used in the isospectral approach is to fix the Hamiltonians to the standard isospectral invariants, that satisfy only the isospectral equations and then to adapt the choice of Darboux coordinates so that the isospectral equations identify with the isomonodromic ones. On the contrary, the strategy of \cite{MarchalOrantinAlameddine2022} is to set the Darboux coordinates to the apparent singularities and their dual on the spectral curve and then to solve the isomonodromy equations to obtain directly the Hamiltonians.
Both strategies have interests and technical difficulties which we shall investigate by building an explicit bi-rational change of coordinates that relates both sides. Up to this change of coordinates, the non-autonomous Hamiltonians and the fundamental symplectic two form (Theorems 6.3 and 7.1 of \cite{MarchalOrantinAlameddine2022}) computed identify with those of Yamakawa \cite{Yamakawa2017TauFA,Yamakawa2019FundamentalTwoForms}.

Note that in this geometric perspective, there are consequences on the symplectic Ehresmann connections on the fibration $\hat{F}_{\mathcal{R},r}\to B$ from the space of Lax matrices ($\hat{F}_{\mathcal{R},r}$ of Definition \ref{DefFR}) to the base $B$ of times (i.e. irregular times and location of poles). Indeed, recall that the intrinsic perspective on non-autonomous Hamiltonians is that given any two symplectic Ehresmann connections on $\hat{F}_{\mathcal{R},r}$, their difference is governed by an intrinsic horizontal one-form $\omega$ on $\hat{F}_{\mathcal{R},r}$. Choosing a system of coordinates $\mathbf{t}$ on $B$, we have $\omega= \sum_k H_k dt_k$. Therefore any symplectic and time-independent change of Darboux coordinates shall leave both the fundamental two form $\Omega$ but also $\omega$ invariant corresponding to the same symplectic Ehresmann connection. On the contrary, if the change of Darboux coordinates is time dependent, then $\Omega$ may be invariant but $\omega$ is generally not, implying that the two symplectic Ehresmann connections differ. In our case, the one-to-one map from Darboux coordinates given by apparent singularities and the Darboux coordinates corresponding to the isospectral setting are related by a non-trivial time-dependent map indicating that both symplectic Ehresmann connections differ. This point shall be discussed in more details as a conclusion.

\section{Meromorphic connections, gauges and Darboux coordinates}\label{SectionMero}

\subsection{Meromorphic connections and irregular times}

One of the main results of \cite{MarchalOrantinAlameddine2022}, which is also obvious from the underlying Poisson structure, is to show that the Hamiltonian structure is the same for meromorphic connections in $\mathfrak{gl}_2(\mathbb{C})$ and in $\mathfrak{sl}_2(\mathbb{C})$. Thus, we shall restrict, without loss of generality, to meromorphic connections in $\mathfrak{sl}_2(\mathbb{C})$ in the rest of the article.

The space of $\mathfrak{sl}_2(\mathbb{C})$ meromorphic connections has been studied from many different perspectives. In the present article, we shall mainly follow the point of view of the Montr\'{e}al group \cite{Darboux_coord93,HarnadHurtubise1997} together with some insights from the work of P. Boalch \cite{Boalch2001}. Let us first define the space we shall study.

\begin{definition}[Space of rational connections]\label{DefFR}
Let $n \in \mathbb{N}$ be a non-negative integer and $\{X_i\}_{i=1}^n$ be $n$ distinct points in the complex plane (we shall implicitly identify a point in the complex plane with its corresponding complex number). Let us denote $\mathcal{R} := \{\infty,X_1,\dots,X_n\}$ and $\mathcal{R}_0 := \{X_1,\dots,X_n\}$. For any $\mathbf{r}:=(r_\infty,r_1,\dots,r_n)\in \left(\mathbb{N}\setminus\{0\}\right)^{n+1}$, let us define
\beq
F_{\mathcal{R},\mathbf{r}} := \left\{\hat{L}(\lambda) = \sum_{k=1}^{r_\infty-1} \hat{L}^{[\infty,k]} \lambda^{k-1} + \sum_{s=1}^n \sum_{k=0}^{r_s-1} \frac{\hat{L}^{[X_s,k]}}{(\lambda-X_s)^{k+1}}
\,\, \text{ with }\,\, \{\hat{L}^{[p,k]}\} \in \left(\mathfrak{sl}_2(\mathbb{C})\right)^{r-1}\right\}/\text{SL}_2(\mathbb{C}) 
\eeq
where $r = r_\infty + \underset{s=1}{\overset{n}{\sum}} r_s$ and $\text{SL}_2(\mathbb{C})$ acts simultaneously by conjugation on all coefficients $\{\hat{L}^{[p,k]}\}_{p,k}$.

\sloppy{Let us denote $\hat{F}_{\mathcal{R},\mathbf{r}}$  the subset of $F_{\mathcal{R},\mathbf{r}}$ corresponding to generic non-resonant rational connections. It corresponds to the subset  of $F_{\mathcal{R},\mathbf{r}}$ composed of elements with coefficients $\{\hat{L}^{[\infty,k]}\}_{1\leq k\leq r_\infty-1}\cup\{\hat{L}^{[X_s,k]}\}_{1\leq s\leq n, 1\leq k\leq r_s-1}$ having distinct eigenvalues or with eigenvalues not differing by integers in the case of Fuchsian singularities (i.e. when $r_s=1$, the difference of two eigenvalues (monodromies) of $\hat{L}^{[X_s,0]}$ should not be an integer).\footnote{In fact, all formulas obtained in this article remain valid under the weaker assumption that the leading order at each singular pole has distinct eigenvalues (or that the eigenvalues do not differ by some integers for Fuchsian singularities) provided that one should still consider the deformation space relatively to all eigenvalues. 
}}
\end{definition}

\begin{remark}We warn the reader that the terminology twisted/untwisted, ramified/unramified, generic/non-generic and resonant/non-resonant varies across the literature and we refer to Definition \ref{DefFR} for the precise setup at stake in this article.
\end{remark}

\begin{remark}\normalsize{In} the present article, we shall assume that $\infty$ is always a pole following the standard convention. Of course, one may always use a change of coordinates in order to remove such assumption.   
\end{remark}

\begin{remark}In general, the ``non-resonant" condition for irregular singularities is presented in the assumption that leading order matrices at each pole have distinct eigenvalues. This condition ensures that each eigenspace is of dimension one and hence that the leading order is diagonalizable. As mentioned in \cite{MarchalAlameddineP1Hierarchy2023} and \cite{MarchalOrantinAlameddine2022} for $\mathfrak{gl}_2$ meromorphic connections, most of the algebraic formulations extend to the limit where some eigenvalues are the same up to some truncations of lower-triangular Toeplitz matrices. The situation is more complicated at the level of deformations because there are two options: we can either just consider the full deformation space and just look at a point where two times would be equal, or completely erase one direction by imposing it to belong to the submanifold on which two eigenvalues are equal. But in both cases, the situation is manageable on the isomonodromic side and formulas remain essentially of the same form. On the isospectral side, the situation is similar and Theorem \ref{MontrealResults} shows that in the $\mathfrak{sl}_2$ setting the isospectral invariants remain the same in the resonant case. Thus, we believe that the distinct eigenvalues condition for leading orders is too strong but that only the assumption that the leading order is diagonalizable is sufficient for the computation of Lax matrices and Hamiltonians in both setups and we let this issue for future works.
\end{remark}

$F_{\mathcal{R},\mathbf{r}}$ can be given a Poisson structure  inherited from the Poisson structure of a corresponding loop algebra \cite{Harnad:1993hw,HarnadHurtubise,woodhouse2007duality}. It is a Poisson space of dimension 
\beq
\dim F_{\mathcal{R},\mathbf{r}} = 3r-6.
\eeq	

The space $F_{\mathcal{R},\mathbf{r}}$ has been intensively studied from the point of view of isospectral and isomonodromic deformations. Following P. Boalch's works \cite{Boalch2001,Boalch2012,Boalch2022} and \cite{Douccot2022TopologyOI,Douccot2022LocalWM} one can use the Poisson structure on $F_{\mathcal{R},\mathbf{r}}$ in order to describe it as a bundle whose fibers are symplectic leaves obtained by fixing the irregular type and monodromies of $\hat{L}(\lambda)$. Let us briefly review this perspective and use it to define local coordinates on $F_{\mathcal{R},\mathbf{r}}$ trivializing the fibration.

In this article, \textbf{we shall restrict to $\hat{F}_{\mathcal{R},\mathbf{r}}$ to simplify the presentation} but the present setup may surely be extended to $F_{\mathcal{R},\mathbf{r}}$.

\medskip

For any pole $p \in \mathcal{R}$, let us define a local coordinate
\beq
\forall\, p \in \mathcal{R} \, , \; 
z_p(\lambda):= \left\{ \begin{array}{l}
(\lambda-p) \qquad \text{if} \qquad p \in \{X_1,\dots,X_n\} \cr
\lambda^{-1} \qquad \text{if} \qquad p = \infty \cr
\end{array}
\right.
\eeq

Given $\hat{L}(\lambda)$ in an orbit of $\hat{F}_{\mathcal{R},\mathbf{r}}$ and a pole $p \in \mathcal{R}$, there exists a gauge matrix $G_p \in \text{GL}_2[[z_p(\lambda)]]$ holomorphic at $\lambda = p$ such that the gauge transformation $\Psi_p=G_p \hat{\Psi}$ provides
\footnotesize{\beq \label{Birkhoff} \Psi_p(\lambda)=\Psi_p^{(\text{reg})}(\lambda)\, \diag\left(\exp\left(- \sum_{k=1}^{r_p-1}\frac{t_{p,k}}{k z_p(\lambda)^{\,k}}  +t_{p,0}\ln z_p(\lambda)\right) , \exp\left( \sum_{k=1}^{r_p-1}\frac{t_{p,k}}{k z_p(\lambda)^{\,k}}  -t_{p,0}\ln z_p(\lambda)\right) \right) 
\eeq} 
\normalsize{where} $\Psi_p^{(\text{reg})}(\lambda)$ is regular at $\lambda=p$. It corresponds to a Lax matrix $L_p=G_p\hat{L}(\lambda) G_p^{-1} + (\partial_\lambda G_p)G_p^{-1}$ satisfying
\beq \label{DiagoCondition}
G_p \,\hat{L}(\lambda) d\lambda\, G_p^{-1} + (\partial_\lambda G_p)G_p^{-1}d\lambda = dD_p(z_p) - \Lambda_p \frac{dz_p}{z_p} 
 \,\,\text{where}\,\, D_p(z_p) = \sum_{k=1}^{r_p-1} \frac{Q_{p,k}}{z_p^{\, k}}
\eeq 
with 
\beq 
D_{p,k} =  \diag \left(- \frac{t_{p,k}}{k}, \frac{t_{p,k}}{k}\right) \,\,\text{ and }\, 
\Lambda_p = \diag(-t_{p,0},t_{p,0})\,,\, \forall \, p\in \mathcal{R} 
\eeq
for some complex numbers  $\left(t_{p,k}\right)_{p\in \mathcal{R},0\leq k\leq r_p-1}$. $D_p(z_p)$ is called the irregular type of $\hat{L}$ at $p$ and $ \Lambda_p$ its residue (also called formal monodromy). Equation \eqref{Birkhoff} is known as the Birkhoff factorization or formal normal solution or Turritin-Levelt fundamental form. We shall denote $\mathbf{t}=\left(t_{p,k}\right)_{p\in \mathcal{R},1\leq k\leq r_p-1}$ the irregular times while the residues $\mathbf{t}_0:=\left(t_{p,0}\right)_{p\in \mathcal{R}}$ will be referred to as monodromies by abuse of language.

\begin{remark}\normalsize{In} the literature, the set of irregular times $\mathbf{t}=\left(t_{p,k}\right)_{p\in \mathcal{R},1\leq k\leq r_p-1}$ is referred to as ``spectral times'' or ``KP times''. This terminology originates from the study of isospectral systems and do not include the monodromy parameters $\left(t_{p,0}\right)_{p\in \mathcal{R}}$. 
\end{remark}

\subsection{Representative normalized at infinity \label{SectionNormalizationInfinity}}
Fixing the irregular times $\mathbf{t}$ and monodromies $\mathbf{t}_0$ of $\hat{L}(\lambda)$ does not fix it uniquely. In each orbit in $\hat{F}_{\mathcal{R},\mathbf{r}}$, there exists a unique element such that $\hat{L}^{[\infty,r_\infty-1]}$ is diagonal and such that $\hat{L}^{[\infty,r_\infty-2]}$ takes the form 
\beq
\Res_{\lambda\to \infty} \hat{L}(\lambda) \lambda^{-(r_\infty-2)}=-\begin{pmatrix} \beta_{r_\infty-2} & \omega\\ \delta_{r_\infty-2} & -\beta_{r_\infty-2} \end{pmatrix} .
\eeq
where $\omega$ is a given non-zero constant.\footnote{In \cite{MarchalOrantinAlameddine2022}, $\omega$ was set to $1$ for simplicity since the Hamiltonian system does not depend on the choice of $\omega$. However, as we will see it in Section \ref{SectionConnectionIso}, imposing the isospectral condition of \cite{BertolaHarnadHurtubise2022} requires that $\omega$ must depend on the irregular times when $r_\infty=1$ so that in this context, one cannot simply set it to $1$.}
 
In particular, one may identify $\hat{F}_{\mathcal{R},\mathbf{r}}$ with the space of such representatives:
\begin{itemize} \item If $r_\infty\geq 3$:
\small{\beq\label{NormalizationInfty}
\begin{array}{ll}
\hat{F}_{\mathcal{R},\mathbf{r}} := &\left\{  \tilde{L}(\lambda) = {\displaystyle \sum_{k=1}^{r_\infty-1}} \tilde{L}^{[\infty,k]} \lambda^{k-1} + {\displaystyle \sum_{s=1}^n \sum_{k=0}^{r_s-1}} \frac{\tilde{L}^{[X_s,k]}}{(\lambda-X_s)^{k+1}}
\,\text{ such that }\,  \right.\cr
& \,\,\, \{\tilde{L}^{[\infty,k]}\}_{1\leq k\leq r_\infty-1}\cup\{\tilde{L}^{[X_s,k]}\}_{1\leq s\leq n, 0\leq k\leq r_s-1} \in \left(\mathfrak{sl}_2(\mathbb{C})\right)^{r-1} \, \text{have distinct eigenvalues,} \cr
& \left.  \, \tilde{L}^{[\infty,r_\infty-1]}=\text{diag}(-t_{\infty,r_\infty-1},t_{\infty,r_\infty-1})  \text{ and }\right.\cr
&\left. \, \tilde{L}^{[\infty,r_\infty-2]} =\begin{pmatrix} \beta_{r_\infty-2} & \omega \\ \delta_{r_\infty-2} & -\beta_{r_\infty-2} \end{pmatrix}, \, (\beta_{r_\infty-2},\delta_{r_\infty-2}) \in \mathbb{C}^2
\right\}
\end{array}
\eeq}
\item If $r_\infty=2$:
\small{\beq\begin{array}{ll}
\hat{F}_{\mathcal{R},\mathbf{r}} := &\left\{  \tilde{L}(\lambda) = {\displaystyle \tilde{L}^{[\infty,1]}} + {\displaystyle \sum_{s=1}^n \sum_{k=0}^{r_s-1}} \frac{\tilde{L}^{[X_s,k]}}{(\lambda-X_s)^{k+1}}
\,\text{ such that }\,  \right.\cr
& \,\,\, \{\tilde{L}^{[\infty,1]}\}\cup\{\tilde{L}^{[X_s,k]}\}_{1\leq s\leq n, 0\leq k\leq r_s-1} \in \left(\mathfrak{sl}_2(\mathbb{C})\right)^{r-1} \, \text{have distinct eigenvalues,} \cr
& \left.  \, \tilde{L}^{[\infty,1]}=\text{diag}(-t_{\infty,1},t_{\infty,1})  \text{ and } {\displaystyle\sum_{s=1}^n}\tilde{L}^{[X_s,0]} =\begin{pmatrix} \beta_0 & \omega \\ \delta_0 & -\beta_0 \end{pmatrix}, \, (\beta_0,\delta_0) \in \mathbb{C}^2
\right\}
\end{array}
\eeq}
\item If $r_\infty=1$:
\small{\beq\begin{array}{ll}
\hat{F}_{\mathcal{R},\mathbf{r}} := &\left\{  \tilde{L}(\lambda) =  {\displaystyle \sum_{s=1}^n \sum_{k=0}^{r_s-1}} \frac{\tilde{L}^{[X_s,k]}}{(\lambda-X_s)^{k+1}}
\,\text{ such that }\,  \right.\cr
& \,\,\, \{\tilde{L}^{[X_s,k]}\}_{1\leq s\leq n, 0\leq k\leq r_s-1} \in \left(\mathfrak{sl}_2(\mathbb{C})\right)^{r-1} \, \text{have distinct eigenvalues,} \cr
& \left.  \, {\displaystyle\sum_{s=1}^n\tilde{L}^{[X_s,0]}}=\text{diag}(-t_{\infty,0},t_{\infty,0})  \text{ and }\right.\cr
&\left. \, {\displaystyle\sum_{s=1}^n}\tilde{L}^{[X_s,1]}+{\displaystyle\sum_{s=1}^n} X_s\tilde{L}^{[X_s,0]} =\begin{pmatrix} \beta_{-1} & \omega \\ \delta_{-1} & -\beta_{-1} \end{pmatrix}, \, (\beta_{-1},\delta_{-1}) \in \mathbb{C}^2
\right\}
\end{array}
\eeq}
\end{itemize}

\normalsize{In} the following, we shall use the notation $\tilde{L}(\lambda)$ whenever we consider such a representative and we shall call it a representative ``normalized at infinity'' to stress that the $\text{SL}_2(\mathbb{C})$ global action has been used to select a specific form at infinity. Note that the choice of a representative is necessary to have uniquely defined Lax matrices. Of course, one may always perform a gauge transformation to select another type of normalization.

\begin{remark}\normalsize{The} Hamiltonian system is independent of the choice of representative. In the isospectral approach, the part of the $\text{SL}_2(\mathbb{C})$ action corresponding to the action of diagonal matrices (that has been used here to fix the subleading order at infinity with a $\omega$ in the upper-right entry) may be translated at the level of coordinates. In \cite{BertolaHarnadHurtubise2022}, this corresponds to the reduction of the $(\mathbf{x},\mathbf{y})$ variables to the $(\mathbf{u},\mathbf{v})$ variables.
\end{remark}

\begin{remark}\label{RemarkCoeff}\normalsize{The} choice of normalization at infinity, implies that coefficient $\beta_{r_\infty-2}$ identifies with $-t_{\infty,r_\infty-2}$ for $r_\infty\geq 2$. Indeed, the diagonalization of the singular part given by \eqref{DiagoCondition} implies that
\bea\label{SubLeadingCoeff} \det (\td{L}(\lambda)+G_\infty^{-1} \partial_\lambda G_\infty)&\overset{\lambda\to \infty}{=}&-\left(\sum_{k=0}^{r_\infty-1} t_{\infty,k}\lambda^{k-1}+ O(\lambda^{-2})\right)^2\cr
&=&-t_{\infty,r_\infty-1}^2\lambda^{2r_\infty-4}- 2t_{\infty,r_\infty-1}t_{\infty,r_\infty-2}\lambda^{2r_\infty-5}+O(\lambda^{2r_\infty-6})\cr
&&\eea
The l.h.s. is of the form:
\begin{itemize}\item If $r_\infty\geq 3$:
\footnotesize{\bea \label{Idrinftygeq3}&& \det (\td{L}(\lambda)+G_\infty^{-1} \partial_\lambda G_\infty)\overset{\lambda\to \infty}{=}\cr
&&\begin{vmatrix} -t_{\infty,r_\infty-1}\lambda^{r_\infty-2}+\beta_{r_\infty-2}\lambda^{r_\infty-3}+O(\lambda^{r_\infty-4})&\omega\lambda^{r_\infty-3}+O(\lambda^{r_\infty-4})\\
\delta_{r_\infty-2} \lambda^{r_\infty-3}+O(\lambda^{r_\infty-4})& t_{\infty,r_\infty-1}\lambda^{r_\infty-2}-\beta_{r_\infty-2}\lambda^{r_\infty-3}+O(\lambda^{r_\infty-4})\end{vmatrix}\cr
&&=-t_{\infty,r_\infty-1}^2\lambda^{2r_\infty-4}+2\beta_{r_\infty-2} t_{\infty,r_\infty-1}\lambda^{2r_\infty-5}+O(\lambda^{2r_\infty-6})
\eea}
\normalsize{Identifying} with the r.h.s. of \eqref{SubLeadingCoeff} we get $\beta_{r_\infty-2}=-t_{\infty,r_\infty-2}$.
\item If $r_\infty=2$, we observe that the matrix $G_\infty(\lambda)$ is of the form $G_{\infty}(\lambda)=G_0+G_1\lambda^{-1}+O(\lambda^{-2})$ with $G_0$ diagonal (in order to preserve the fact that the  leading order is already diagonal). Consequently $G_{\infty}^{-1}\partial_\lambda G_\infty=O(\lambda^{-2})$. The diagonalization \eqref{DiagoCondition} implies:
\beq\label{SubLeadingCoeff2} \det (\td{L}(\lambda)+G_\infty^{-1} \partial_\lambda G_\infty)\overset{\lambda\to \infty}{=}-t_{\infty,1}^2- 2t_{\infty,1}t_{\infty,0}\lambda^{-1}+O(\lambda^{-2})\eeq
Since $G_{\infty}^{-1}\partial_\lambda G_\infty=O(\lambda^{-2})$, the l.h.s. is of the form
\bea \label{Idrinftyequal2} \det (\td{L}(\lambda)+G_\infty^{-1} \partial_\lambda G_\infty)&\overset{\lambda\to \infty}{=}&\begin{vmatrix} -t_{\infty,1}+\beta_0\lambda^{-1}+O(\lambda^{-2})&\omega\lambda^{-1}+O(\lambda^{-2})\\
\delta_0 \lambda^{-1}+O(\lambda^{-2})& t_{\infty,1}-\beta_0\lambda^{-1}+O(\lambda^{-2})\end{vmatrix}\cr
&=&-t_{\infty,1}^2+2\beta_0 t_{\infty,1}\lambda^{-1}+O(\lambda^{-2})
\eea
Identifying with the r.h.s. of \eqref{SubLeadingCoeff2} we get $\beta_0=-t_{\infty,0}$.
\item If $r_\infty=1$ we simply have
\beq \beta_{-1}=\sum_{s=1}^n \left[\td{L}^{[X_s,1]}+X_s\td{L}^{[X_s,0]}\right]_{1,1}\eeq
\end{itemize}
\end{remark}

\subsection{General isomonodromic deformations}
The irregular times $\mathbf{t}$ and monodromies $\mathbf{t}_0$ provide a natural set of parameters parametrizing $\hat{F}_{\mathcal{R},\mathbf{r}}$. As the terminology suggests, a general isomonodromic deformation corresponds to a differential operator of the form.
\begin{definition}\label{DefGeneralDeformationsDefinition} We define the following general deformation operators.
\beq \label{GeneralDeformationsDefinition}\mathcal{L}_{\boldsymbol{\alpha}}=\sum_{k=1}^{r_\infty-1} \alpha_{\infty,k} \partial_{t_{\infty,k}}+ \sum_{s=1}^n\sum_{k=1}^{r_s-1} \alpha_{X_s,k} \partial_{t_{X_s,k}}+ \sum_{s=1}^n \alpha_{X_s} \partial_{X_s}\eeq
where we define the vector $\boldsymbol{\alpha}\in \mathbb{C}^{g+2-n}$ by
\beq \boldsymbol{\alpha}= \sum_{k=1}^{r_\infty-1} \alpha_{\infty,k}\mathbf{e}_{\infty,k}+\sum_{s=1}^n\sum_{k=1}^{r_s-1} \alpha_{X_s,k} \mathbf{e}_{X_s,k}+\sum_{s=1}^n \alpha_{X_s}\mathbf{e}_{X_s}.\eeq
\end{definition}

Note that the vector $\boldsymbol{\alpha}$ may depend on both the irregular times and the monodromies. Associated to a general isomonodromic deformation is a matrix $\td{A}_{\boldsymbol{\alpha}}(\lambda,\mathbf{t},\mathbf{t_0})$ such that
\beq \mathcal{L}_{\boldsymbol{\alpha}}[\td{\Psi}(\lambda,\mathbf{t},\mathbf{t_0})]=\td{A}_{\boldsymbol{\alpha}}(\lambda,\mathbf{t},\mathbf{t_0})\td{\Psi}(\lambda,\mathbf{t},\mathbf{t_0})\eeq
The important point is that $\td{A}_{\boldsymbol{\alpha}}(\lambda,\mathbf{t},\mathbf{t_0})$ is meromorphic in $\lambda$ with a pole structure dominated by those ot $\td{L}$.
The compatibility of 
\bea \partial_\lambda \td{\Psi}(\lambda,\mathbf{t},\mathbf{t_0})&=&\td{L}(\lambda,\mathbf{t},\mathbf{t_0}) \td{\Psi}(\lambda,\mathbf{t},\mathbf{t_0})\cr
\mathcal{L}_{\boldsymbol{\alpha}}[\td{\Psi}(\lambda,\mathbf{t},\mathbf{t_0})]&=&\td{A}_{\boldsymbol{\alpha}}(\lambda,\mathbf{t},\mathbf{t_0})\td{\Psi}(\lambda,\mathbf{t},\mathbf{t_0})
\eea
is equivalent to the isomonodromic compatibility equations
\beq\label{LaxCompat} \mathcal{L}_{\boldsymbol{\alpha}}[\td{L}(\lambda,\mathbf{t},\mathbf{t_0})]-\partial_\lambda \td{A}_{\boldsymbol{\alpha}}(\lambda,\mathbf{t},\mathbf{t_0})+\left[\td{L}(\lambda,\mathbf{t},\mathbf{t_0}),\td{A}_{\boldsymbol{\alpha}}(\lambda,\mathbf{t},\mathbf{t_0})\right]=0\eeq
which corresponds to a zero curvature equation for connections.

\medskip

In \cite{MarchalOrantinAlameddine2022}, the authors showed that one may define $g$ (the genus of the associated spectral curve) isomonodromic times $\left(\tau_i\right)_{1 \leq i \leq g}$ from the irregular times along with some additional trivial times $T_1$ and $T_2$ . It is then proven that the Hamiltonian structure contains no dependence on the trivial times so that the Darboux coordinates $\left(q_i,p_i\right)_{1 \leq i \leq g}$ depend uniquely on the isomonodromic times. The existence of two trivial times in the $\mathfrak{sl}_2(\mathbb{C})$ case is directly related to the action of the M\"{o}bius transformations on the Lax matrix. Indeed, it is known that the Hamiltonian system is invariant under the M\"{o}bius transformation while the Lax matrix remains dependent on this choice. Since we have assumed that infinity is always a pole, the M\"{o}bius transformations reduce to only translations and dilations and hence has two independent generators giving rise to two trivial times. The normalization of the Lax matrix at infinity suggests some specific values for the trivial times hence implying some specific values for the irregular times in order to obtain compact results. In \cite{MarchalOrantinAlameddine2022}, such convenient choices were discussed and depend on the degree of the pole at infinity. We shall make the same choice in the rest of the paper that we list here.


\begin{proposition}[Definition of trivial and isomonodromic times \cite{MarchalOrantinAlameddine2022}]\label{PropTrivialTimes}We have:
\begin{itemize}\item If $r_\infty\geq 3$, then we shall take $t_{\infty,r_\infty-1}=1$ and $t_{\infty,r_\infty-2}=0$. The isomonodromic times are defined by:
\bea \tau_{\infty,k}&=&t_{\infty,k}\,\,,\,\, \forall\, k\in \llbracket 1,r_\infty-3\rrbracket\cr
\tau_{X_s,k}&=&t_{X_s,k}\,\,,\,\, \forall\, (s,k)\in \llbracket 1,n\rrbracket\times\llbracket 1,r_s-1\rrbracket\cr
\td{X}_s&=&X_s \,\,,\,\, \forall\, s\in \llbracket 1,n\rrbracket
\eea
\item If $r_\infty=2$, then we shall take $t_{\infty,1}=1$ and $X_1=0$. The isomonodromic times are defined by:
\bea \tau_{X_s,k}&=&t_{X_s,k}\,\,,\,\, \forall\, (s,k)\in \llbracket 1,n\rrbracket\times\llbracket 1,r_s-1\rrbracket\cr
\td{X}_s&=&X_s \,\,,\,\, \forall\, s\in \llbracket 2,n\rrbracket
\eea
\item If $r_\infty=1$ and $n\geq 2$, then we shall take $X_1=0$ and $X_2=1$. The isomonodromic times are defined by:
\bea \tau_{X_s,k}&=&t_{X_s,k}\,\,,\,\, \forall\, (s,k)\in \llbracket 1,n\rrbracket\times\llbracket 1,r_s-1\rrbracket\cr
\td{X}_s&=&X_s \,\,,\,\, \forall\, s\in \llbracket 2,n\rrbracket
\eea
\item If $r_\infty=1$ and $n=1$, then we shall take $X_1=0$ and $t_{X_1,r_1-1}=1$. The isomonodromic times are defined by:
\beq \tau_{X_1,k}=t_{X_1,k}\,\,,\,\, \forall\, k\in \llbracket 1,r_1-2\rrbracket\eeq
\end{itemize}
\end{proposition}

\begin{remark} \normalsize{Note} that in \cite{MarchalOrantinAlameddine2022}, the isomonodromic times corresponding to the irregular times were defined with a factor two (i.e. $\tau_{p,k}=2t_{p,k}$) that we omit in this paper to have simpler relations with isospectral deformations. This simply corresponds to a trivial rescaling of the corresponding isomonodromic times. However, one needs to adapt these factors when matching the current results with the expressions obtained in  \cite{MarchalOrantinAlameddine2022}. 
\end{remark}

\begin{remark} \label{ImportantRemark}\normalsize{For} compactness and to match with results of \cite{BertolaHarnadHurtubise2022}, we may write some formulas keeping all irregular times to arbitrary values but one needs to remember that the upcoming formulas are only valid when the values of $t_{\infty,r_\infty-1}$, $t_{\infty,r_\infty-2}$, $t_{X_1,r_1-1}$, $X_1$ and $X_2$ are exactly the values specified in Proposition \ref{PropTrivialTimes} depending on the values of $r_\infty$ and $n$.
\end{remark}

\section{Lax pairs and Hamiltonians in the oper gauge} \label{SectionOldResults}
This section presents a summary of results of \cite{MarchalOrantinAlameddine2022}  that shall be used in the present article. The main purpose of \cite{MarchalOrantinAlameddine2022} was to obtain explicit formulas for the Lax matrices and the Hamiltonians in the oper gauge where the Lax matrix is a companion matrix. In this gauge, the natural Darboux coordinates are $(\mathbf{q},\mathbf{p}):=\{q_1,\dots,q_g,p_1,\dots,p_g\}$. We refer to \cite{MarchalOrantinAlameddine2022} for details on the derivations of the results presented in this section.  

\medskip

Starting from the Lax system $\partial_\lambda \td{\Psi}(\lambda)=\td{L}(\lambda)\td{\Psi}(\lambda)$ normalized at infinity according to Section \ref{SectionNormalizationInfinity}, the first step of \cite{MarchalOrantinAlameddine2022} is to perform a gauge transformation to switch the setting to the oper form (i.e. a companion-like form for the Lax matrix) of the differential system. This step is also equivalent to write the so-called ``quantum curve'' satisfied by $\td{\Psi}_{1,1}(\lambda)$ and $\td{\Psi}_{1,2}(\lambda)$. 

\begin{proposition}[Gauge transformation to oper form \cite{MarchalOrantinAlameddine2022}]\label{OperForm}Let us define $S_\infty:=\llbracket 0 ,r_\infty-4\rrbracket$ and  $S_{X_s}:=\llbracket 1,r_s\rrbracket\,,\,\, \forall \, s\in \llbracket 1,n\rrbracket$ and
\beq \label{GaugeTransfo} G(\lambda):=\begin{pmatrix} 1&0\\
\td{L}_{1,1}(\lambda)&\td{L}_{1,2}(\lambda)
\end{pmatrix}
=\begin{pmatrix}1&0\\
\frac{-Q(\lambda)-(t_{\infty,r_\infty-1}\lambda+g_0)\underset{j=1}{\overset{g}{\prod}}(\lambda-q_j) }{\underset{s=1}{\overset{n}{\prod}} (\lambda-X_s)^{r_s}}& \frac{\omega\underset{j=1}{\overset{g}{\prod}}(\lambda-q_j)}{\underset{s=1}{\overset{n}{\prod}} (\lambda-X_s)^{r_s}}
\end{pmatrix}\eeq
where $g:=r_\infty-3+\underset{s=1}{\overset{n}{\sum}}r_s$ and
\beq Q(\lambda):=-\sum_{i=1}^g p_i \prod_{s=1}^{n}(q_i-X_s)^{r_s} \prod_{j\neq i}\frac{\lambda-q_j}{q_i-q_j}\eeq
and 
\bea g_0&:=&t_{\infty,r_\infty-2}+t_{\infty,r_\infty-1}\left(\sum_{j=1}^gq_j -\sum_{s=1}^nr_sX_s\right) \,\, \text{ if } r_\infty\geq 2\cr
g_0&:=&t_{\infty^{(1)},0}\left(\sum_{j=1}^gq_j -\sum_{s=1}^nr_sX_s\right)-\sum_{s=1}^n\left[X_s\td{L}^{[X_s,0]}+\td{L}^{[X_s,1]}\right]_{1,1} \,\, \text{ if }r_\infty=1\cr
&&
\eea
The gauge transformation $\Psi(\lambda):=G(\lambda)\td{\Psi}(\lambda)$ implies that $\Psi(\lambda)$ satisfies the companion-like system (also named oper form):
\beq \label{GeneralLSys}\partial_\lambda \Psi(\lambda)=L(\lambda)\Psi(\lambda)\,,\, L(\lambda)=\begin{pmatrix}0&1\\ L_{2,1}(\lambda)&L_{2,2}(\lambda) \end{pmatrix}\eeq
with 
\bea \label{L21}L_{2,1}(\lambda)&=& -\td{P}_2(\lambda) +\sum_{j=0}^{r_\infty-4}H_{\infty,j}\lambda^j+\sum_{s=1}^n\sum_{j=1}^{r_s}H_{X_s,j}(\lambda-X_s)^{-j}\cr
&&-\sum_{j=1}^{g} \frac{
 p_j}{\lambda-q_j}-
t_{\infty^{(1)},r_\infty-1}\lambda^{r_\infty-3}\delta_{r_\infty\geq 3}\cr
L_{2,2}(\lambda)&=&\sum_{j=1}^{g} \frac{
1}{\lambda-q_j} -\sum_{s=1}^n \frac{
 r_s}{\lambda-X_s}
\eea
where 
\beq \label{tdP2}\td{P}_2(\lambda):=\sum_{j=\text{max}(0,r_\infty-3)}^{2r_\infty-4} P_{\infty,j}^{(2)} \lambda^{j} +\sum_{s=1}^n\sum_{j=r_s+1}^{2r_s} \frac{P_{X_s,j}^{(2)}}{(\lambda-X_s)^j}\eeq
is defined by
\bea \label{tdP22}
 P_{\infty,2r_\infty-4-k}^{(2)}&:=&-\sum_{j=0}^k t_{\infty,r_\infty-1-j}t_{\infty,r_\infty-1 -(k-j)} \,\, ,\,\, \forall \,k\in \llbracket 0, r_{\infty}-1\rrbracket,\cr
P_{X_s,2r_s-k}^{(2)}&:=&-\sum_{j=0}^k t_{X_s,r_s-1-j}t_{X_s,r_s-1 -(k-j)} \,\, ,\,\, \forall\, s\in \llbracket 1,n \rrbracket\,\, ,\,\, \forall \,k\in \llbracket 0, r_s-1\rrbracket
\eea
Moreover for $r_\infty=2$, we have the additional relations
\beq \label{ConditionsAddrinftyequal2} P_{\infty,0}^{(2)}=-(t_{\infty,1})^2,\,\,
\sum_{s=1}^n H_{X_s,1}= 
\sum_{j=1}^g p_j+(2t_{\infty,1}t_{\infty,0}-
 t_{\infty,1}) 
\eeq
while, for $r_\infty=1$, we have
\bea \label{ConditionsAddrinftyequal1}\sum_{s=1}^n H_{X_s,1}&=& 
\sum_{j=1}^g p_j,\cr
\sum_{s=1}^n X_s H_{X_s,1}+\sum_{s=1}^n H_{X_s,2}\delta_{r_s\geq 2}&=&
 \sum_{j=1}^g q_j p_j +\sum_{s=1}^n (t_{X_s,0})^2\delta_{r_s=1} 
+t_{\infty,0}(t_{\infty,0}-1
).\cr
&& 
\eea
\end{proposition}

Note that in Proposition \ref{OperForm}, the coordinates $\mathbf{q}$ are defined as the apparent singularities of the system (i.e. zeros of $\det(\td{\Psi}(\lambda))$). The dual coordinates $\mathbf{p}$ are chosen so that for any $i\in \llbracket 1,g\rrbracket$, $(q_i,p_i)$ is a point on the spectral curve, i.e. $\det(p_i I_2-\td{L}(q_i))=0$ for all $i\in \llbracket 1,g\rrbracket$.

\begin{remark} The oper gauge is less natural from geometric perspective because the entries of the Lax matrices have apparent singularities. In particular, they do not belong at first sight in $\hat{F}_{\mathcal{R},\mathbf{r}}$ whose pole structure is fixed. But at the computational level, the oper gauge has many advantages: it is independent of the choice of normalization of $\td{L}$ and it only has two non-trivial entries instead of four, which is more convenient to solve the compatibility equations. Moreover, if it is also straightforward to go from a Lax matrix $\td{L}$ to its oper form (which is equivalent to the scalar ODE satisfied by the first line of the wave matrix $\td{\Psi}$), it is not obvious to rewrite a scalar ODE with apparent singularities into a standard Lax matrix with fixed pole structure. Historically, this problem corresponds to the difference between Garnier systems (scalar ODEs) and the associated Painlev\'{e} property and their Lax formulations.
\end{remark}

The oper gauge is particularly convenient for computations since all the information on the Hamiltonian system is contained into $L_{2,1}(\lambda)$, i.e. in the $g$ unknown coefficients $(H_{p,k})_{p\in \mathcal{R},k\in S_p}$. In fact, the main result of \cite{MarchalOrantinAlameddine2022} is to prove that the coordinates $(\mathbf{q},\mathbf{p})$ satisfy some Hamiltonian evolutions relatively to the isomonodromic times defined in Proposition \ref{PropTrivialTimes} and that the Hamiltonians are easily expressed in terms of $\mathbf{H}:=\left(H_{p,k}\right)_{p\in \mathcal{R},k\in S_p}$.

\begin{theorem}[Hamiltonian evolution of $(\mathbf{q},\mathbf{p})$ (Theorem $7.1$ of \cite{MarchalOrantinAlameddine2022})]\label{TheoHamiltonian}For any isomonodromic time $\tau$ and the choice of irregular times defined in Proposition \ref{PropTrivialTimes}, we have
\beq \partial_{\tau}q_i=\frac{\partial \text{Ham}^{(\alpha_{\tau})}(\mathbf{q},\mathbf{p},\mathbf{t},\mathbf{t}_0)}{\partial p_i}\,\,,\,\, \partial_{\tau}p_i=-\frac{\partial \text{Ham}^{(\alpha_{\tau})}(\mathbf{q},\mathbf{p},\mathbf{t},\mathbf{t}_0)}{\partial q_i} \,,\,\, \forall \, i\in \llbracket 1,g\rrbracket\eeq
where the Hamiltonians $\text{Ham}^{(\alpha_{\tau})}(\mathbf{q},\mathbf{p},\mathbf{t},\mathbf{t}_0)$ are given by:
\bea \label{NewHamReduced}\begin{pmatrix}\text{Ham}^{(\alpha_{t_{\infty,1}})}(\mathbf{q},\mathbf{p},\mathbf{t},\mathbf{t}_0)\\ \vdots \\ (r_\infty-3)\text{Ham}^{(\alpha_{t_{\infty,r_\infty -3}})}(\mathbf{q},\mathbf{p},\mathbf{t},\mathbf{t}_0)\end{pmatrix}&=&\left(M_{\infty}(\mathbf{t})\right)^{-1}\begin{pmatrix}H_{\infty,r_\infty-4}(\mathbf{q},\mathbf{p},\mathbf{t},\mathbf{t}_0)\\ \vdots\\ H_{\infty,0}(\mathbf{q},\mathbf{p},\mathbf{t},\mathbf{t}_0) \end{pmatrix}\cr
\begin{pmatrix}\text{Ham}^{(\alpha_{t_{X_s,1}})}(\mathbf{q},\mathbf{p},\mathbf{t},\mathbf{t}_0)\\ \vdots \\ (r_s-1)\text{Ham}^{(\alpha_{t_{X_s,r_s-1}})}(\mathbf{q},\mathbf{p},\mathbf{t},\mathbf{t}_0)\end{pmatrix}&=&\left(M_{X_s}(\mathbf{t})\right)^{-1}\begin{pmatrix}H_{X_s,r_s}(\mathbf{q},\mathbf{p})\\ \vdots\\  H_{X_s,2}(\mathbf{q},\mathbf{p})\end{pmatrix}\,,\,\,\forall\, s\in \llbracket 1,n\rrbracket\cr
\begin{pmatrix}\text{Ham}^{(\alpha_{X_1})}(\mathbf{q},\mathbf{p},\mathbf{t},\mathbf{t}_0)\\ \vdots \\ \text{Ham}^{(\alpha_{X_n})}(\mathbf{q},\mathbf{p},\mathbf{t},\mathbf{t}_0)\end{pmatrix}&=&\begin{pmatrix}H_{X_1,1}(\mathbf{q},\mathbf{p},\mathbf{t},\mathbf{t}_0)\\ \vdots\\ H_{X_n,1}(\mathbf{q},\mathbf{p},\mathbf{t},\mathbf{t}_0)\end{pmatrix}
\eea
and the matrices $M_{\infty}(\mathbf{t})$ and $\left(M_{X_s}(\mathbf{t})\right)_{1\leq s\leq n}$ are defined by
\beq \label{Minfty} M_\infty(\mathbf{t})=\begin{pmatrix}t_{\infty,r_\infty-1}&0&\dots&& &\dots &0\\
t_{\infty,r_\infty-2}&t_{\infty,r_\infty-1}& 0& && &\vdots\\
t_{\infty,r_\infty-3}&t_{\infty,r_\infty-2}& t_{\infty,r_\infty-1}&\ddots && &\vdots\\
\vdots & \ddots&\ddots &\ddots  && &\vdots\\
\vdots &\ddots&\ddots&\ddots&\ddots&\ddots&\vdots\\
t_{\infty,4} &\ddots &\ddots&&\ddots& t_{\infty,r_\infty-1} &0\\
t_{\infty,3}&t_{\infty,4}& \dots && t_{\infty,r_\infty-3}& t_{\infty,r_\infty-2}& t_{\infty,r_\infty-1}
 \end{pmatrix}\in \mathcal{M}_{r_\infty-3}(\mathbb{C})
\eeq
and 
\beq \label{MXs} M_{X_s}(\mathbf{t})=\begin{pmatrix}t_{X_s,r_s-1}&0&\dots& &\dots &0\\
t_{X_s,r_s-2}&t_{X_s,r_s-1}& 0& & &\vdots\\
\vdots & \ddots&\ddots &\ddots  & &\vdots\\
\vdots &\ddots&\ddots&\ddots&0&\vdots\\
t_{X_s,2}&\ddots &\ddots&\ddots& t_{X_s,r_s-1}&0\\
t_{X_s,1}&t_{X_s,2}& \dots & & t_{X_s,r_s-2}& t_{X_s,r_s-1}
 \end{pmatrix} \,,\,\, \forall \, s\in \llbracket 1,n\rrbracket
\eeq
\end{theorem}

The previous theorem is supplemented by the explicit expression of the coefficients $\mathbf{H}(\mathbf{q},\mathbf{p},\mathbf{t},\mathbf{t}_0)$.
\begin{proposition}[Explicit expression of the $\mathbf{H}(\mathbf{q},\mathbf{p},\mathbf{t},\mathbf{t}_0)$ (Proposition $5.1$ of \cite{MarchalOrantinAlameddine2022})]\label{PropExplicitHpk} We have 
\footnotesize{\beq 
\begin{pmatrix} (V_\infty)^{t}& (V_{X_1})^t&\dots &(V_{X_n})^{t}\end{pmatrix}\begin{pmatrix}\mathbf{H}_{\infty}\\\mathbf{H}_{X_1}\\ \vdots \\ \mathbf{H}_{X_n} \end{pmatrix}=\begin{pmatrix} p_1^2 + p_1\underset{s=1}{\overset{n}{\sum}} \frac{ r_s}{q_1-X_s}+\td{P}_2(q_1)+ \underset{i\neq 1}{\sum}\frac{p_i-p_1}{q_1-q_i}+ t_{\infty^{(1)},r_\infty-1}q_1^{r_\infty-3}\delta_{r_\infty\geq 3}\\
\vdots\\
\vdots\\
p_g^2 + p_g\underset{s=1}{\overset{n}{\sum}} \frac{r_s}{q_g-X_s}+\td{P}_2(q_g)+ \underset{i\neq g}{\sum}\frac{p_i-p_g}{q_g-q_i}+ t_{\infty^{(1)},r_\infty-1}q_g^{r_\infty-3}\delta_{r_\infty\geq 3}
\end{pmatrix}
\eeq}
\sloppy{\normalsize{with} $\mathbf{H}_{\infty}=(H_{\infty,0},\dots H_{\infty,r_\infty-4})^t$ (null is $r_\infty \leq 3$), and, for all $s\in \llbracket 1,n\rrbracket$, $\mathbf{H}_{X_s}=(H_{X_s,1},\dots, H_{X_s,r_s})^t$. Matrices $V_\infty$ and $\left(V_{X_s}\right)_{1\leq s\leq n}$ are defined by}
\beq V_\infty=\begin{pmatrix}1&1 &\dots &\dots &1\\
q_1& q_2&\dots &\dots& q_{g}\\
\vdots & & & & \vdots\\
\vdots & & & & \vdots\\
q_1^{r_\infty-4}& q_2^{r_\infty-4} &\dots & \dots& q_{g}^{r_\infty-4}\end{pmatrix}\,,\, 
V_{X_s}=\begin{pmatrix}\frac{1}{q_1-X_s}& \dots &\dots& \frac{1}{q_g-X_s}\\
\frac{1}{(q_1-X_s)^2}& \dots &\dots& \frac{1}{(q_g-X_s)^2}\\
\vdots & & & \vdots\\
\vdots & & & \vdots\\
\frac{1}{(q_1-X_s)^{r_s}}& \dots &\dots& \frac{1}{(q_g-X_s)^{r_s}}\\
\end{pmatrix}
\eeq
\normalsize{Moreover}, for $r_\infty=2$ (resp. $r_\infty=1$) the previous linear system has to be supplemented with the additional relations \eqref{ConditionsAddrinftyequal2} (resp. \eqref{ConditionsAddrinftyequal1}). Finally for $r_\infty=1$ we have:
\bea\label{g0Specialrinftyequal1} g_0&=&\frac{1}{2t_{\infty^{(1)},0}}\Big[ -\sum_{s=1}^n(2X_sP_{X_s,2}^{(2)}\delta_{r_s=1}+P_{X_s,3}^{(2)}\delta_{r_s=2})\cr
&& + \sum_{s=1}^n(X_s^2H_{X_s,1}+2X_sH_{X_s,2}\delta_{r_s\geq 2}+H_{X_s,3}\delta_{r_s\geq 3})- \sum_{j=1}^{g}p_jq_j^2\cr
&&+t_{\infty^{(1)},0}(2t_{\infty^{(1)},0}-1)\left(\sum_{j=1}^g q_j-\sum_{s=1}^n r_s X_s\right)
\Big]
\eea 
\end{proposition} 

Theorem \ref{TheoHamiltonian} and Proposition \ref{PropExplicitHpk} provide the explicit expression of the Hamiltonian system satisfied by $(\mathbf{q},\mathbf{p})$. This Hamiltonian system is equivalent to the fact that the differential system $\partial_\lambda \Psi(\lambda)=L(\lambda)\Psi(\lambda)$ is supplemented by the general auxiliary system
\beq \label{GeneralAux} \mathcal{L}_{\boldsymbol{\alpha}}[\Psi(\lambda)]=A_{\boldsymbol{\alpha}}(\lambda) \Psi(\lambda)\eeq
where $A_{\boldsymbol{\alpha}}(\lambda)$ is rational in $\lambda$ with poles in $\mathcal{R}$ dominated by those of $L$ and simple poles at the apparent singularities $(q_i)_{1\leq i\leq g}$.
The compatibility of the Lax system \eqref{GeneralLSys} and \eqref{GeneralAux} is equivalent to the isomonodromic compatibility equation:
\beq\label{CompatLA} \mathcal{L}_{\boldsymbol{\alpha}}[L(\lambda)]=\partial_{\lambda}A_{\boldsymbol{\alpha}}(\lambda) -\left[L(\lambda),A_{\boldsymbol{\alpha}}(\lambda)\right]\eeq

In \cite{MarchalOrantinAlameddine2022}, it is proved that the first line of $A_{\boldsymbol{\alpha}}(\lambda)$ is determined by
\bea \label{ExpreA11}\left[A_{\boldsymbol{\alpha}}(\lambda)\right]_{1,1}&=&c_{\infty,0} -\sum_{j=1}^{g} \frac{p_j\mu_{j}^{(\boldsymbol{\alpha})}}{\lambda-q_j}\cr
 \left[A_{\boldsymbol{\alpha}}(\lambda)\right]_{1,2}&=&\nu_{\infty,-1}^{(\boldsymbol{\alpha})}\lambda \delta_{r_\infty=1}+\nu_{\infty,0}^{(\boldsymbol{\alpha})} \delta_{r_\infty\leq 2}+ \sum_{j=1}^{g} \frac{\mu_{j}^{(\boldsymbol{\alpha})}}{\lambda-q_j}
\eea
with\footnote{The derivation of $c_{\infty,0}$ was not done in \cite{MarchalOrantinAlameddine2022} because it is irrelevant in the computation of the Hamiltonian system but only related to the choice of normalization of $\td{L}$. For completeness, the formula is proved in Appendix \ref{Appendixcinfty0}.} 
\beq \label{cinfty0}c_{\infty,0}=\frac{1}{2\omega}\mathcal{L}_{\boldsymbol{\alpha}}[\omega] +\frac{1}{2}\nu^{(\boldsymbol{\alpha})}_{\infty,-1}\delta_{r_\infty=1}\eeq
while the second line of $A_{\boldsymbol{\alpha}}(\lambda)$ is immediately determined by the first line of the isomonodromic compatibility equation \eqref{CompatLA} (which is always trivial in the oper gauge):
\bea\label{A21} \left[A_{\boldsymbol{\alpha}}(\lambda)\right]_{2,1}&=& \partial_{\lambda} \left[A_{\boldsymbol{\alpha}}(\lambda)\right]_{1,1}+\left[A_{\boldsymbol{\alpha}}(\lambda)\right]_{1,2}L_{2,1}(\lambda),\cr
\left[A_{\boldsymbol{\alpha}}(\lambda)\right]_{2,2}&=& \partial_{\lambda} \left[A_{\boldsymbol{\alpha}}(\lambda)\right]_{1,2}+\left[A_{\boldsymbol{\alpha}}(\lambda)\right]_{1,1}+\left[A_{\boldsymbol{\alpha}}(\lambda)\right]_{1,2}L_{2,2}(\lambda),
\eea
Coefficients $\left(\mu_{j}^{(\boldsymbol{\alpha})}\right)_{1\leq j\leq g}$ are determined by
\small{\beq \label{muj}\begin{pmatrix} V_\infty\\  V_1 \\ \vdots \\\vdots \\V_n\end{pmatrix}\begin{pmatrix} \mu^{(\boldsymbol{\alpha})}_1\\ \vdots\\\vdots\\ \mu^{(\boldsymbol{\alpha})}_g\end{pmatrix}= \begin{pmatrix} \boldsymbol{\nu}^{(\boldsymbol{\alpha})}_\infty\\ -\boldsymbol{\nu}^{(\boldsymbol{\alpha})}_{X_1} \\ \vdots\\  -\boldsymbol{\nu}^{(\boldsymbol{\alpha})}_{X_n}\end{pmatrix}\,\text{ with }\,  \boldsymbol{\nu}^{(\boldsymbol{\alpha})}_\infty=\begin{pmatrix}\nu^{(\boldsymbol{\alpha})}_{\infty,1}\\ \nu^{(\boldsymbol{\alpha})}_{\infty,2}\\\vdots \\ \nu^{(\boldsymbol{\alpha})}_{\infty,r_\infty-3}\end{pmatrix} \,\,,\,\, \boldsymbol{\nu}^{(\boldsymbol{\alpha})}_{X_s}=\begin{pmatrix} \nu^{(\boldsymbol{\alpha})}_{{X_s},0} -X_s\nu^{(\boldsymbol{\alpha})}_{{\infty},-1}\delta_{r_\infty=1}-\nu^{(\boldsymbol{\alpha})}_{{\infty},0}\delta_{r_\infty\leq 2}  \\ \nu^{(\boldsymbol{\alpha})}_{{X_s},1}-\nu^{(\boldsymbol{\alpha})}_{{\infty},-1}\delta_{r_\infty=1}\\  \vdots \\ \nu^{(\boldsymbol{\alpha})}_{{X_s},r_s-1}\end{pmatrix} \eeq}
\normalsize{Coefficients} $\left( \nu^{(\boldsymbol{\alpha})}_{p,k}\right)_{p,k}$ correspond to the expansion of $ \left[A_{\boldsymbol{\alpha}}(\lambda)\right]_{1,2}$ at each pole. Under conditions of Proposition \ref{PropTrivialTimes}, we have:
\bea \label{A12nu} \left[A_{\boldsymbol{\alpha}}(\lambda)\right]_{1,2}&\overset{\lambda\to \infty}{=}&\nu^{(\boldsymbol{\alpha})}_{{\infty},-1}\delta_{r_\infty=1}\lambda+\nu^{(\boldsymbol{\alpha})}_{{\infty},0}\delta_{r_\infty\leq 2}+ \sum_{i=1}^{r_\infty-3} \frac{\nu^{(\boldsymbol{\alpha})}_{\infty,i}}{\lambda^i} +O\left(\lambda^{-(r_\infty-2)}\right),\cr
\left[A_{\boldsymbol{\alpha}}(\lambda)\right]_{1,2}&\overset{\lambda\to X_s}{=}&\sum_{i=0}^{r_s-1} \nu^{(\boldsymbol{\alpha})}_{X_s,i}(\lambda-X_s)^i +O\left((\lambda-X_s)^{r_s}\right)
\eea

A crucial technical point is given by the following proposition. 

\begin{proposition}[Expression of $\nu_{p,k}^{(\boldsymbol{\alpha})}$ (Proposition $4.1$ of \cite{MarchalOrantinAlameddine2022})]\label{nus}Let us define
\bea \label{DefY} Y(\lambda)&:=&-\sum_{k=0}^{r_\infty-1}t_{\infty,k}\lambda^{k-1}+\sum_{s=1}^n \sum_{k=0}^{r_s-1} \frac{t_{X_s,k}}{(\lambda-X_s)^{k+1}}\cr
a(\lambda)&:=&-\sum_{k=1}^{r_\infty-3}\frac{\alpha_{\infty,k}}{k}\lambda^{k} -\sum_{s=1}^n \sum_{k=1}^{r_s-1} \frac{\alpha_{X_s,k}}{(\lambda-X_s)^k}-\sum_{s=1}^n\sum_{k=0}^{r_s-1}\frac{\alpha_{X_s}t_{X_s,k}}{(\lambda-X_s)^{k+1}}
\eea
then, coefficients $\left(\nu^{(\boldsymbol{\alpha})}_{p,k}\right)_{p,k}$ are determined by 
\beq \label{YA12} Y(\lambda)\left[A_{\boldsymbol{\alpha}}(\lambda)\right]_{1,2}\overset{\lambda\to \infty}{=}a(\lambda)+ O\left(1\right) \,,\,Y(\lambda)\left[A_{\boldsymbol{\alpha}}(\lambda)\right]_{1,2}\overset{\lambda\to X_s}{=}a(\lambda) +O\left(1\right) \eeq
which is equivalent to (see Appendix $E$ of \cite{MarchalOrantinAlameddine2022} for details) 
\beq \nu^{(\boldsymbol{\alpha})}_{X_s,0}=-\alpha_{X_s}\,,\,M_{X_s}\begin{pmatrix}  \nu^{(\boldsymbol{\alpha})}_{X_s,1}\\ \vdots \\ \nu^{(\boldsymbol{\alpha})}_{X_s,r_s-1}\end{pmatrix}=-\begin{pmatrix} 
\frac{\alpha_{X_s,r_s-1}}{r_s-1}\\
\vdots\\
\frac{\alpha_{X_s,1}}{1}
\end{pmatrix}\,,\,M_\infty\begin{pmatrix} \nu^{(\boldsymbol{\alpha})}_{\infty,1}\\ \vdots \\ \nu^{(\boldsymbol{\alpha})}_{\infty,r_\infty-3}\end{pmatrix}=\begin{pmatrix} \frac{\alpha_{\infty,r_\infty-3}}{r_\infty-3} \\ \vdots \\ \frac{\alpha_{\infty,1}}{1}\end{pmatrix}
\eeq
\end{proposition}

Let us finally make some comments about the advantages and drawbacks of the previous formulas.
\begin{itemize}\item The oper gauge and the associated Darboux coordinates $(\mathbf{q},\mathbf{p})$ allow the explicit resolution of the compatibility equations \eqref{CompatLA} and the explicit expression of the Hamiltonian system via Theorem \ref{TheoHamiltonian} and Proposition \ref{PropExplicitHpk}.
\item The relation between the Hamiltonians and the spectral invariants $\mathbf{H}$ of $L(\lambda)$ is very simple (eq. \eqref{NewHamReduced}) and recovers results of \cite{MartaPaper2022} obtained by confluences of simple poles.
\item The Hamiltonian system has a complicated dependence in the Darboux coordinates $(\mathbf{q},\mathbf{p})$. In particular, it has singularities at $q_i=q_j$ for $i\neq j$. 
\item One may obtain the expression of the geometric Lax matrices $\td{L}(\lambda)$ and $\td{A}(\lambda)$ in terms of $(\mathbf{q},\mathbf{p})$ using the explicit gauge transformation \eqref{GaugeTransfo}:
\bea \td{L}(\lambda)&=& G(\lambda)^{-1} L(\lambda) G(\lambda)-G(\lambda)^{-1}\partial_{\lambda}G(\lambda)\cr
\td{A}_{\boldsymbol{\alpha}}(\lambda)&=& G(\lambda)^{-1} A_{\boldsymbol{\alpha}}(\lambda) G(\lambda)-G(\lambda)^{-1}\mathcal{L}_{\boldsymbol{\alpha}}[G(\lambda)]
\eea
However, the expressions are particularly complicated and it is not obvious in the formulas that $\td{L}(\lambda)$ and $\td{A}_{\boldsymbol{\alpha}}(\lambda)$ do not have poles at $\lambda\in \{q_1,\dots,q_g\}$. 
\end{itemize}

The conclusion of this section is that the Darboux coordinates $(\mathbf{q},\mathbf{p})$ are well-adapted to the oper gauge $(L(\lambda),A_{\boldsymbol{\alpha}}(\lambda))$. The Hamiltonian in these coordinates has also a natural interpretation in terms of interacting particles. However, regarding the initial geometric Lax matrix $\td{L}(\lambda)$ and the geometric setup of Section \ref{SectionMero}, coordinates $(\mathbf{q},\mathbf{p})$ are cumbersome. Thus, the first goal of the present paper is to define some more appropriate Darboux coordinates $(\mathbf{Q},\mathbf{P})$ to obtain explicit formulas for the Lax matrix $\td{L}(\lambda)$ and explicit formulas for the Hamiltonians associated to these Darboux coordinates. 

\section{Geometric Darboux coordinates and associated Hamiltonian system}
\subsection{Geometric Darboux coordinates}
As explained above, we shall introduce a new set of Darboux coordinates $(\mathbf{Q},\mathbf{P})$ that is more adapted to the Lax matrix $\td{L}(\lambda)$. We shall name these coordinates ``Geometric Darboux coordinates'' to stress that they are naturally adapted to the initial gauge and thus to the geometric construction of meromorphic connections on $\mathfrak{gl}_2(\mathbb{C})$.

\begin{definition}[Geometric Darboux coordinates]\label{DefGeometricCoordinates}
We define the ``geometric Darboux coordinates'' $(\mathbf{Q},\mathbf{P})=\left(Q_{p,k},P_{p,k}\right)_{p\in \mathcal{R},k\in S_p}$ by
\bea\label{DefNewcoor} \frac{\omega\underset{j=1}{\overset{g}{\prod}}(\lambda-q_j)}{\underset{s=1}{\overset{n}{\prod}}(\lambda-X_s)^{r_s}}&=&\sum_{s=1}^n\sum_{k=1}^{r_s} \frac{Q_{X_s,k}}{(\lambda-X_s)^k}+\sum_{k=0}^{r_\infty-4} Q_{\infty,k}\lambda^k +\omega\,\delta_{r_\infty\geq 3}\lambda^{r_\infty-3}\cr
p_i&=&\sum_{k=0}^{r_\infty-4} P_{\infty,k}\frac{\partial Q_{\infty,k}(q_1,\dots,q_g)}{\partial q_i}+\sum_{s=1}^n\sum_{k=1}^{r_s} P_{X_s,k}\frac{\partial Q_{X_s,k}(q_1,\dots,q_g)}{\partial q_i} \,\,,\,\, \forall \, i\in \llbracket 1,g\rrbracket\cr
&&
\eea
with the additional relations
\begin{itemize}
\item For $r_\infty=2$: 
\beq \label{AddConstrains} \sum_{s=1}^n Q_{X_s,1}=\omega \,\,\text{ and }\,\, \sum_{m=1}^{r_s}P_{X_s,m}Q_{X_s,m}=0\eeq
\item For $r_\infty=1$:
\bea \label{AddConstrains2} \sum_{s=1}^n Q_{X_s,1}=0 \,&\text{ and }&\, \sum_{s=1}^n Q_{X_s,2}+\sum_{s=1}^{n} X_sQ_{X_s,1}=\omega \cr
 \sum_{s=1}^n\sum_{m=1}^{r_s}P_{X_s,m}Q_{X_s,m}=0\,&\text{ and }&\, \underset{s=1}{\overset{n}{\sum}}\underset{m=1}{\overset{r_s-1}{\sum}}P_{X_s,m}Q_{X_s,m+1} +\underset{s=1}{\overset{n}{\sum}}\underset{m=1}{\overset{r_s}{\sum}} X_sP_{X_s,m}Q_{X_s,m}=0\cr
&&
\eea
\end{itemize}
\end{definition}

Conditions \eqref{AddConstrains} and \eqref{AddConstrains2} are derived for completeness in Lemma \ref{LemmaQExpression} of Appendix \ref{AppendixProofGeoLaxMatrices}. Lemma $6.3$ of \cite{MarchalAlameddineP1Hierarchy2023} proves that \textbf{the change of coordinates $(\mathbf{q},\mathbf{p})\leftrightarrow (\mathbf{Q},\mathbf{P})$ is a time-independent symplectic change of coordinates}. Consequently, one may obtain the Hamiltonian in terms of $(\mathbf{Q},\mathbf{P})$ using Theorem \ref{TheoHamiltonian} and replacing the coordinates $(\mathbf{q},\mathbf{p})$ in Proposition \ref{PropExplicitHpk} in terms of $(\mathbf{Q},\mathbf{P})$. However, this approach is very cumbersome in practice and we shall derive a more convenient formula to obtain the Hamiltonians relatively to $(\mathbf{Q},\mathbf{P})$ in Theorem \ref{TheoHamNew}. 

\begin{remark} \normalsize{The} change of coordinates $\mathbf{q} \leftrightarrow \mathbf{Q}$ is algebraically natural. It simply corresponds to writing a rational function either in its factorized form or in its polar form. It turns out that the factorized form is more convenient in the oper gauge while the polar form is more adapted to the geometric gauge. The remaining part of the change of coordinates giving $\mathbf{p}$ in terms of $(\mathbf{Q},\mathbf{P})$ follows naturally to preserve the symplectic structure.
\end{remark}

\subsection{Expression of the geometric Lax matrices}

Using the new coordinates of Definition \ref{DefGeometricCoordinates}, we may write the expression of the Hamiltonian and of the geometric Lax matrices $\td{L}$ and $\td{A}_{\boldsymbol{\alpha}}$. In order to obtain formulas, we shall introduce the following notation:
\begin{definition}\label{NotationSingularPart}Let $a\in \mathbb{C}$. For any function $f(\lambda)$ admitting a Laurent series at $\lambda\to a$ we shall denote $\left[f(\lambda)\right]_{a,-}$ the singular part at $\lambda\to a$:
\beq f(\lambda)=\sum_{k=-r}^{\infty} F_k (\lambda-a)^{k} \,\, \Rightarrow\,\, \left[f(\lambda)\right]_{a,-}=\sum_{k=1}^{r} F_{-k} (\lambda-a)^{-k} \eeq
For any function $f(\lambda)$ admitting a Laurent series at $\lambda\to\infty$, we shall denote $\left[f(\lambda)\right]_{\infty,+}$ the polynomial part at infinity (including the $O(1)$ term):
\beq f(\lambda)=\sum_{k=-r}^{\infty} F_k\lambda^{-k} \,\,\Rightarrow\,\, \left[f(\lambda)\right]_{\infty,+}=\sum_{k=0}^{r} F_{-k}\lambda^{k} \eeq
\end{definition}

\begin{theorem}[Geometric Lax matrices in terms of geometric Darboux coordinates]\label{GeoLaxMatrices}We have:
\bea \td{L}_{1,2}(\lambda)&=&\sum_{s=1}^n\sum_{k=1}^{r_s} \frac{Q_{X_s,k}}{(\lambda-X_s)^k}+\sum_{k=0}^{r_\infty-4} Q_{\infty,k}\lambda^k+\omega\delta_{r_\infty\geq 3}\lambda^{r_\infty-3}\cr
\td{L}_{1,1}(\lambda)&=&-\omega\sum_{k=0}^{r_\infty-4} P_{\infty,r_\infty-4-k}\lambda^k-\sum_{k=0}^{r_\infty-5}\sum_{m=0}^{r_\infty-5-k}P_{\infty,m}Q_{\infty,k+1+m}\lambda^k\cr
&&+ \sum_{s=1}^n\sum_{k=1}^{r_s}\sum_{m=1}^{r_s+1-k}P_{X_s,m}Q_{X_s,k+m-1}(\lambda-X_s)^{-k}\cr
&&-\frac{(t_{\infty,r_\infty-1}\lambda+g_0)}{\omega}\left(\omega\,\delta_{r_\infty\geq 3}\lambda^{r_\infty-3}+\sum_{k=0}^{r_\infty-4} Q_{\infty,k}\lambda^k+\sum_{s=1}^n\sum_{k=1}^{r_s}\frac{Q_{X_s,k}}{(\lambda-X_s)^k}\right)\cr
\td{L}_{2,2}(\lambda)&=&-\td{L}_{1,1}(\lambda)\cr
\left[\td{A}_{\boldsymbol{\alpha}}(\lambda)\right]_{1,2}&=& \omega \,\nu_{\infty,1}^{(\boldsymbol{\alpha})}\lambda^{r_\infty-4}\delta_{r_\infty\geq 4}
+\sum_{j=0}^{r_\infty-5}\left(\omega \,\nu_{\infty,r_\infty-3-j}^{(\boldsymbol{\alpha})}+\sum_{k=j+1}^{r_\infty-4}\nu_{\infty,k-j}^{(\boldsymbol{\alpha})}Q_{\infty,k}\right)\lambda^j\cr
&&+\sum_{s=1}^n\sum_{j=1}^{r_s}\left(\sum_{k=j}^{r_s} \nu_{X_s,k-j}^{(\boldsymbol{\alpha})}  Q_{X_s,k}\right)(\lambda-X_s)^{-j}\cr
&&
\eea
and
\footnotesize{\bea
\left[\td{A}_{\boldsymbol{\alpha}}(\lambda)\right]_{1,1}&=&\frac{1}{2\omega}\mathcal{L}_{\boldsymbol{\alpha}}[\omega] -t_{\infty,1}\nu^{(\boldsymbol{\alpha})}_{\infty,0}\delta_{r_\infty,2} +\left(\frac{1}{2}-t_{\infty,0}\right)\nu^{(\boldsymbol{\alpha})}_{\infty,-1}\delta_{r_\infty,1}\cr
&&-t_{\infty,r_\infty-1}\nu_{\infty,1}^{(\boldsymbol{\alpha})}\lambda^{r_\infty-3}\delta_{r_\infty\geq 3}-t_{\infty,r_\infty-1}\nu_{\infty,2}^{(\boldsymbol{\alpha})}\lambda^{r_\infty-4}\delta_{r_\infty\geq 4}\cr 
&&-\left(t_{\infty,r_\infty-1}\nu_{\infty,3}^{(\boldsymbol{\alpha})}+\left(\omega P_{\infty,0}+\frac{t_{\infty,r_\infty-1}Q_{\infty,r_\infty-5}+Q_{\infty,r_\infty-4}g_0}{\omega}\right)\nu_{\infty,1}^{(\boldsymbol{\alpha})}\right)\lambda^{r_\infty-5}\delta_{r_\infty\geq 5}\cr
&&-\sum_{j=0}^{r_\infty-6}\Big[t_{\infty,r_\infty-1}\nu_{\infty,r_\infty-2-j}^{(\boldsymbol{\alpha})}+\nu_{\infty,r_\infty-4-j}^{(\boldsymbol{\alpha})}\left(\omega P_{\infty,0}+\frac{t_{\infty,r_\infty-1}Q_{\infty,r_\infty-5}+Q_{\infty,r_\infty-4}g_0}{\omega}\right)\cr
&&+\sum_{i=1}^{r_\infty-5-j}\nu_{\infty,i}^{(\boldsymbol{\alpha})}\left(\omega P_{\infty,r_\infty-4-i-j}+\sum_{m=0}^{r_\infty-5-i-j}P_{\infty,m}Q_{\infty,j+i+1+m}+\frac{t_{\infty,r_\infty-1}Q_{\infty,j+i-1}+g_0Q_{\infty,j+i}}{\omega}\right)\Big]\lambda^j\cr
&&+\sum_{s=1}^n\left(P_{X_s,1}- \frac{t_{\infty,r_\infty-1}X_s+g_0}{\omega}\right)Q_{X_s,r_s}\nu_{X_s,0}^{(\boldsymbol{\alpha})}(\lambda-X_s)^{-r_s}\cr
&&+ \sum_{s=1}^n\sum_{r=1}^{r_s-1}\Big[\sum_{i=0}^{r_s-r}\sum_{m=1}^{r_s+1-r-i}\nu_{X_s,i}^{(\boldsymbol{\alpha})}P_{X_s,m}Q_{X_s,r+i+m-1}-\frac{(t_{\infty,r_\infty-1}X_s+g_0)}{\omega}\sum_{i=0}^{r_s-r}\nu_{X_s,i}^{(\boldsymbol{\alpha})} Q_{X_s,r+i}\cr
&&- \frac{t_{\infty,r_\infty-1}}{\omega}\sum_{i=0}^{r_s-1-r}\nu_{X_s,i}^{(\boldsymbol{\alpha})}Q_{X_s,r+i+1}\Big](\lambda-X_s)^{-r}\cr
\left[\td{A}_{\boldsymbol{\alpha}}(\lambda)\right]_{2,2}&=&-\left[\td{A}_{\boldsymbol{\alpha}}(\lambda)\right]_{1,1}
\eea}
\normalsize{with} the extra condition for $r_\infty=2$:
\beq \label{ExtraConditionsrinftyequal2} \sum_{s=1}^n\left(\sum_{k=1}^{r_s} \nu_{X_s,k-1}^{(\boldsymbol{\alpha})}  Q_{X_s,k}\right)=\omega \,\nu_{\infty,0}^{(\boldsymbol{\alpha})}\eeq
and the extra conditions for $r_\infty=1$:
\small{\bea\label{ExtraConditionsrinftyequal1} \omega \,\nu_{\infty,-1}^{(\boldsymbol{\alpha})}&=&\sum_{s=1}^n\left(\sum_{k=1}^{r_s} \nu_{X_s,k-1}^{(\boldsymbol{\alpha})}  Q_{X_s,k}\right)\cr
\left(\sum_{s=1}^n X_s^2Q_{X_s,1}+2X_sQ_{X_s,2}+Q_{X_s,3}\right)\nu_{\infty,-1}^{(\boldsymbol{\alpha})}+\omega \,\nu_{\infty,0}^{(\boldsymbol{\alpha})}&=&\sum_{s=1}^n\left(\sum_{k=2}^{r_s} \nu_{X_s,k-2}^{(\boldsymbol{\alpha})}  Q_{X_s,k}+ X_s\sum_{k=1}^{r_s} \nu_{X_s,k-1}^{(\boldsymbol{\alpha})}  Q_{X_s,k}\right)\cr
&&
\eea}
\normalsize{The expression} of $\td{L}_{2,1}(\lambda)$ depends on the value of $r_\infty$:
\begin{itemize}\item If $r_\infty\geq 3$:
\bea\label{tdL21rinftygeq3} \td{L}_{2,1}(\lambda)
&=&\left[\frac{\underset{j=r_\infty-3}{\overset{2r_\infty-4}{\sum}}\left(\underset{m=0}{\overset{2r_\infty-4-j}{\sum}} t_{\infty,r_\infty-1-m}t_{\infty,j+m-r_\infty+3}\right) \lambda^{j} -\td{L}_{1,1}(\lambda)^2}{\omega\lambda^{r_\infty-3}+\underset{k=0}{\overset{r_\infty-4}{\sum}} Q_{\infty,k}\lambda^k+\underset{s=1}{\overset{n}{\sum}}\underset{k=1}{\overset{r_s}{\sum}} \frac{Q_{X_s,k}}{(\lambda-X_s)^k}}\right]_{\infty,+} \cr
&&+\sum_{s=1}^n\left[\frac{\underset{j=r_s+1}{\overset{2r_s}{\sum}}\left(\underset{m=0}{\overset{2r_s-j}{\sum}} t_{X_s,r_s-1-m}t_{X_s,j+m-r_s-1}\right) (\lambda-X_s)^{-j} -\td{L}_{1,1}(\lambda)^2}{\omega\lambda^{r_\infty-3}+\underset{k=0}{\overset{r_\infty-4}{\sum}} Q_{\infty,k}\lambda^k+\underset{s=1}{\overset{n}{\sum}}\underset{k=1}{\overset{r_s}{\sum}} \frac{Q_{X_s,k}}{(\lambda-X_s)^k}}\right]_{X_s,-}\cr
&&
\eea
\item If $r_\infty=2$:
\bea \label{tdL21rinftyequal2}\td{L}_{2,1}(\lambda)&=&
\sum_{s=1}^n\left[\frac{\underset{j=r_s+1}{\overset{2r_s}{\sum}}\left(\underset{m=0}{\overset{2r_s-j}{\sum}} t_{X_s,r_s-1-m}t_{X_s,j+m-r_s-1}\right) (\lambda-X_s)^{-j} -\td{L}_{1,1}(\lambda)^2}{\underset{s=1}{\overset{n}{\sum}}\underset{k=1}{\overset{r_s}{\sum}} \frac{Q_{X_s,k}}{(\lambda-X_s)^k}}\right]_{X_s,-}\cr
&&
\eea
\item If $r_\infty=1$:
\bea \label{tdL21rinftyequal1}\td{L}_{2,1}(\lambda)&=&
\sum_{s=1}^n\left[\frac{\underset{j=r_s+1}{\overset{2r_s}{\sum}}\left(\underset{m=0}{\overset{2r_s-j}{\sum}} t_{X_s,r_s-1-m}t_{X_s,j+m-r_s-1}\right) (\lambda-X_s)^{-j} -\mathring{L}_{1,1}(\lambda)^2}{\td{L}_{1,2}(\lambda)}\right]_{X_s,-}\cr
&&
+\frac{2(t_{\infty,0}\lambda+g_0)}{\omega}\mathring{L}_{1,1}(\lambda)-\frac{(t_{\infty,0}\lambda+g_0)^2}{\omega^2}\td{L}_{1,2}(\lambda)+\frac{(t_{\infty,0})^2}{\omega}
\eea
with
\beq \mathring{L}_{1,1}(\lambda):=\td{L}_{1,1}(\lambda)+\frac{(t_{\infty,0}\lambda+g_0)}{\omega}\td{L}_{1,2}(\lambda)=\sum_{s=1}^n\sum_{k=1}^{r_s}\sum_{m=1}^{r_s+1-k}P_{X_s,m}Q_{X_s,k+m-1}(\lambda-X_s)^{-k}
\eeq
\end{itemize}
The expression of $[\td{A}_{\boldsymbol{\alpha}}(\lambda)]_{2,1}$ also depends on the value of $r_\infty$:
\begin{itemize}
\item For $r_\infty\geq 4$:
\small{\bea \label{tdA21rinftygeq4}[\td{A}_{\boldsymbol{\alpha}}(\lambda)]_{2,1}&=&-\frac{t_{\infty,r_\infty-1}}{\omega^2}\mathcal{L}_{\boldsymbol{\alpha}}[\omega] \lambda+\frac{2}{\omega}\,\underset{\lambda\to \infty}{\Res}\left(\td{A}^{(\boldsymbol{\alpha})}_{1,1}(\lambda)-\frac{\td{L}^{(\boldsymbol{\alpha})}_{1,1}(\lambda)}{\td{L}^{(\boldsymbol{\alpha})}_{1,2}(\lambda)} \td{A}^{(\boldsymbol{\alpha})}_{1,2}(\lambda)\right)+\frac{1}{\omega^3}Q_{\infty,r_\infty-4}\mathcal{L}_{\boldsymbol{\alpha}}[\omega]\cr
&&+\left[\frac{\underset{k=r_\infty-3}{\overset{2r_\infty-5}{\sum}}\left(\underset{j=k+1}{\overset{2r_\infty-4}{\sum}}\underset{m=0}{\overset{2r_\infty-4-j}{\sum}} t_{\infty,r_\infty-1-m}t_{\infty,j+m-r_\infty+3}\nu^{(\boldsymbol{\alpha})}_{\infty,j-k}\right) \lambda^{k}}{\omega\lambda^{r_\infty-3}+\underset{k=0}{\overset{r_\infty-4}{\sum}} Q_{\infty,k}\lambda^k+\underset{s=1}{\overset{n}{\sum}}\underset{k=1}{\overset{r_s}{\sum}} \frac{Q_{X_s,k}}{(\lambda-X_s)^k}}\right]_{\infty,+}\cr
&&+ \left[\frac{\td{L}_{1,1}(\lambda)}{\td{L}_{1,2}(\lambda)}\left(\frac{\td{L}_{1,1}(\lambda)}{\td{L}_{1,2}(\lambda)}[\td{A}_{\boldsymbol{\alpha}}(\lambda)]_{1,2}-2[\td{A}_{\boldsymbol{\alpha}}(\lambda)]_{1,1}\right)\right]_{\infty,+}\cr
&&+\sum_{s=1}^n\left[\frac{\underset{k=r_s+1}{\overset{2r_s}{\sum}}\left( \underset{j=k}{\overset{2r_s}{\sum}}\underset{m=0}{\overset{2r_s-j}{\sum}} t_{X_s,r_s-1-m}t_{X_s,j+m-r_s-1} \nu^{(\boldsymbol{\alpha})}_{X_s,j-k}\right) (\lambda-X_s)^{-k}}{\omega \lambda^{r_\infty-3}+\underset{k=0}{\overset{r_\infty-4}{\sum}} Q_{\infty,k}\lambda^k+\underset{s=1}{\overset{n}{\sum}}\underset{k=1}{\overset{r_s}{\sum}} \frac{Q_{X_s,k}}{(\lambda-X_s)^k}}\right]_{X_s,-}\cr
&&+\sum_{s=1}^n \left[\frac{\td{L}_{1,1}(\lambda)}{\td{L}_{1,2}(\lambda)}\left(\frac{\td{L}_{1,1}(\lambda)}{\td{L}_{1,2}(\lambda)}[\td{A}_{\boldsymbol{\alpha}}(\lambda)]_{1,2}-2[\td{A}_{\boldsymbol{\alpha}}(\lambda)]_{1,1}\right)\right]_{X_s,-}
\eea}
\item \normalsize{For} $r_\infty=3$:
\small{\bea\label{tdA21rinftyequal3} [\td{A}_{\boldsymbol{\alpha}}(\lambda)]_{2,1}&=&-\frac{t_{\infty,2}}{\omega^2}\mathcal{L}_{\boldsymbol{\alpha}}[\omega] \lambda+\frac{2}{\omega}\,\underset{\lambda\to \infty}{\Res}\left(\td{A}^{(\boldsymbol{\alpha})}_{1,1}(\lambda)-\frac{\td{L}^{(\boldsymbol{\alpha})}_{1,1}(\lambda)}{\td{L}^{(\boldsymbol{\alpha})}_{1,2}(\lambda)} \td{A}^{(\boldsymbol{\alpha})}_{1,2}(\lambda)\right)+\frac{1}{\omega^3}\left(\underset{s=1}{\overset{n}{\sum}}Q_{X_s,1}\right)\mathcal{L}_{\boldsymbol{\alpha}}[\omega]\cr
&&+\left[\frac{\underset{k=r_\infty-3}{\overset{2r_\infty-5}{\sum}}\left(\underset{j=k+1}{\overset{2r_\infty-4}{\sum}}\underset{m=0}{\overset{2r_\infty-4-j}{\sum}} t_{\infty,r_\infty-1-m}t_{\infty,j+m-r_\infty+3}\nu^{(\boldsymbol{\alpha})}_{\infty,j-k}\right) \lambda^{k}}{\omega\lambda^{r_\infty-3}+\underset{k=0}{\overset{r_\infty-4}{\sum}} Q_{\infty,k}\lambda^k+\underset{s=1}{\overset{n}{\sum}}\underset{k=1}{\overset{r_s}{\sum}} \frac{Q_{X_s,k}}{(\lambda-X_s)^k}}\right]_{\infty,+}\cr
&&+ \left[\frac{\td{L}_{1,1}(\lambda)}{\td{L}_{1,2}(\lambda)}\left(\frac{\td{L}_{1,1}(\lambda)}{\td{L}_{1,2}(\lambda)}[\td{A}_{\boldsymbol{\alpha}}(\lambda)]_{1,2}-2[\td{A}_{\boldsymbol{\alpha}}(\lambda)]_{1,1}\right)\right]_{\infty,+}\cr
&&+\sum_{s=1}^n\left[\frac{\underset{k=r_s+1}{\overset{2r_s}{\sum}}\left( \underset{j=k}{\overset{2r_s}{\sum}}\underset{m=0}{\overset{2r_s-j}{\sum}} t_{X_s,r_s-1-m}t_{X_s,j+m-r_s-1} \nu^{(\boldsymbol{\alpha})}_{X_s,j-k}\right) (\lambda-X_s)^{-k}}{\omega \lambda^{r_\infty-3}+\underset{k=0}{\overset{r_\infty-4}{\sum}} Q_{\infty,k}\lambda^k+\underset{s=1}{\overset{n}{\sum}}\underset{k=1}{\overset{r_s}{\sum}} \frac{Q_{X_s,k}}{(\lambda-X_s)^k}}\right]_{X_s,-}\cr
&&+\sum_{s=1}^n \left[\frac{\td{L}_{1,1}(\lambda)}{\td{L}_{1,2}(\lambda)}\left(\frac{\td{L}_{1,1}(\lambda)}{\td{L}_{1,2}(\lambda)}[\td{A}_{\boldsymbol{\alpha}}(\lambda)]_{1,2}-2[\td{A}_{\boldsymbol{\alpha}}(\lambda)]_{1,1}\right)\right]_{X_s,-}
\eea}
\item \normalsize{For} $r_\infty=2$:
\small{\bea\label{tdA21rinftyequal2}[\td{A}_{\boldsymbol{\alpha}}(\lambda)]_{2,1}&=&
\sum_{s=1}^n \left[\frac{\underset{k=r_s+1}{\overset{2r_s}{\sum}}\left( \underset{j=k}{\overset{2r_s}{\sum}}\underset{m=0}{\overset{2r_s-j}{\sum}} t_{X_s,r_s-1-m}t_{X_s,j+m-r_s-1} \nu^{(\boldsymbol{\alpha})}_{X_s,j-k}\right) (\lambda-X_s)^{-k}}{\underset{s=1}{\overset{n}{\sum}}\underset{k=1}{\overset{r_s}{\sum}} \frac{Q_{X_s,k}}{(\lambda-X_s)^k}}\right]_{X_s,-}\cr
&&+\sum_{s=1}^n \left[\frac{\td{L}_{1,1}(\lambda)}{\td{L}_{1,2}(\lambda)}\left(\frac{\td{L}_{1,1}(\lambda)}{\td{L}_{1,2}(\lambda)}[\td{A}_{\boldsymbol{\alpha}}(\lambda)]_{1,2}-2[\td{A}_{\boldsymbol{\alpha}}(\lambda)]_{1,1}\right)\right]_{X_s,-}
\eea}
\item \normalsize{For} $r_\infty=1$:
\small{\bea\label{tdA21rinftyequal1}[\td{A}_{\boldsymbol{\alpha}}(\lambda)]_{2,1}&=&
-\frac{2}{\omega^2}\mathcal{L}_{\boldsymbol{\alpha}}[\omega]\left(g_0+\frac{t_{\infty,0}}{\omega}\left(\sum_{s=1}^n X_s^2Q_{X_s,1}+2X_sQ_{X_s,2}+Q_{X_s,3}\right)\right)\cr
&&+\left[\frac{\underset{k=r_s+1}{\overset{2r_s}{\sum}}\left( \underset{j=k}{\overset{2r_s}{\sum}}\underset{m=0}{\overset{2r_s-j}{\sum}} t_{X_s,r_s-1-m}t_{X_s,j+m-r_s-1} \nu^{(\boldsymbol{\alpha})}_{X_s,j-k}\right) (\lambda-X_s)^{-k}}{\underset{s=1}{\overset{n}{\sum}}\underset{k=1}{\overset{r_s}{\sum}} \frac{Q_{X_s,k}}{(\lambda-X_s)^k}}\right]_{X_s,-}\cr
&&+\sum_{s=1}^n \left[\frac{\td{L}_{1,1}(\lambda)}{\td{L}_{1,2}(\lambda)}\left(\frac{\td{L}_{1,1}(\lambda)}{\td{L}_{1,2}(\lambda)}[\td{A}_{\boldsymbol{\alpha}}(\lambda)]_{1,2}-2[\td{A}_{\boldsymbol{\alpha}}(\lambda)]_{1,1}\right)\right]_{X_s,-}
\eea}
\end{itemize}
\normalsize{In} the previous formulas, coefficients $\left(\nu_{p,k}^{(\boldsymbol{\alpha})}\right)_{p,k}$ are defined by Proposition \ref{nus} and $\nu_{\infty,r_\infty-2}^{(\boldsymbol{\alpha})}$ is determined (only for $r_\infty\geq 3$) by
\small{\bea \label{nurinftminus2}\omega\,\nu_{\infty,r_\infty-2}^{(\boldsymbol{\alpha})}&=&\sum_{s=1}^n\sum_{k=1}^{r_s}\nu_{X_s,k-1}^{(\boldsymbol{\alpha})}Q_{X_s,k}-\sum_{j=1}^{r_\infty-3}\nu_{\infty,j}^{(\boldsymbol{\alpha})}Q_{\infty,j-1}\cr
\omega\,\nu_{\infty,r_\infty-1}^{(\boldsymbol{\alpha})}&=&\sum_{s=1}^n\sum_{k=2}^{r_s}\nu_{X_s,k-2}^{(\boldsymbol{\alpha})}Q_{X_s,k}+\sum_{s=1}^n\sum_{k=1}^{r_s}\nu_{X_s,k-1}^{(\boldsymbol{\alpha})}X_sQ_{X_s,k}-\sum_{j=2}^{r_\infty-2}\nu_{\infty,j}^{(\boldsymbol{\alpha})}Q_{\infty,j-2}-\nu_{\infty,1}^{(\boldsymbol{\alpha})}\left(\sum_{s=1}^n Q_{X_s,1}\right)\cr
&&
\eea}
\normalsize{Moreover}, we have
\small{\bea \label{Valueg0} g_0&=& t_{\infty,r_\infty-2}-\frac{t_{\infty,r_\infty-1}}{\omega}Q_{\infty,r_\infty-4}\overset{\text{Prop. \ref{PropTrivialTimes}}}{=}-\frac{Q_{\infty,r_\infty-4}}{\omega} \,\, \text{ if }r_\infty\geq 4\cr
g_0&=&t_{\infty,1}-\frac{t_{\infty,2}}{\omega}\left(\sum_{s=1}^nQ_{X_s,1}\right)\overset{\text{Prop. \ref{PropTrivialTimes}}}{=}-\frac{1}{\omega}\left(\sum_{s=1}^nQ_{X_s,1}\right) \,\, \text{ if }\,r_\infty=3\cr
g_0&=&t_{\infty,0}-\frac{t_{\infty,1}}{\omega}\left(\sum_{s=1}^nQ_{X_s,2}+X_sQ_{X_s,1}\right)\overset{\text{Prop. \ref{PropTrivialTimes}}}{=}t_{\infty,0}-\frac{1}{\omega}\left(\sum_{s=1}^nQ_{X_s,2}+X_sQ_{X_s,1}\right) \,\, \text{ if }\,r_\infty=2\cr
g_0&=&\frac{1}{\omega}\left(\frac{1}{2}-t_{\infty,0}\right)\left(\sum_{s=1}^n X_s^2Q_{X_s,1}+2X_sQ_{X_s,2}+Q_{X_s,3}\right)\cr
&&-\frac{1}{2t_{\infty,0}}\underset{\lambda\to \infty}{\Res} \lambda^2\Big[
(\mathring{L}_{1,1}(\lambda))^2+\partial_\lambda \mathring{L}_{1,1}(\lambda)-\frac{t_{\infty,0}}{\omega}\td{L}_{1,2}(\lambda)-\mathring{L}_{1,1}(\lambda)\frac{\partial_\lambda\td{L}_{1,2}(\lambda)}{\td{L}_{1,2}(\lambda)}+\frac{(t_{\infty,0})^2}{\omega}\td{L}_{1,2}(\lambda)\cr
&&+\td{L}_{1,2}(\lambda)\sum_{s=1}^n\left[\frac{\underset{j=r_s+1}{\overset{2r_s}{\sum}}\left(\underset{m=0}{\overset{2r_s-j}{\sum}} t_{X_s,r_s-1-m}t_{X_s,j+m-r_s-1}\right) (\lambda-X_s)^{-j} -\mathring{L}_{1,1}(\lambda)^2}{\td{L}_{1,2}(\lambda)}\right]_{X_s,-}
\Big] \cr
&&\, \text{ if  }\, r_\infty=1
\eea}\normalsize{}
\end{theorem}

\begin{proof}The proof is done in Appendix \ref{AppendixProofGeoLaxMatrices}. The strategy consists in observing that the entries of the matrices $\td{L}(\lambda)$
and $\td{A}_{\boldsymbol{\alpha}}(\lambda)$ are rational functions of $\lambda$ with poles only in $\{\infty,X_1,\dots,X_n\}$. Thus, rewriting the gauge transformation \eqref{GaugeTransfo} using the new set of coordinates (Lemma \ref{LemmaQExpression}), one may obtain the behavior of each entry at each pole and finally obtain the expressions proposed in Theorem \ref{GeoLaxMatrices}. The fact that the expressions depend on $r_\infty$ comes from the normalization at infinity of the Lax matrices that implies to split cases in the study of the local behavior at infinity of some of their entries.
\end{proof}

\begin{remark}\normalsize{For} $r_\infty=1$, it is convenient to introduce $\mathring{L}_{1,1}(\lambda)$ to obtain a formula for $g_0$ that is consistent. In particular, the matrix $L$ is independent of $g_0$ because geometrically the oper gauge is universal and does not depend on the choice of normalization of $\td{L}$. This is the reason why we had to isolate the dependence of $g_0$ in $\td{L}_{2,1}(\lambda)$ for $r_\infty=1$. One could obtain similar simplifications for $r_\infty=2$ and $r_\infty=3$ but it is less relevant since the value of $g_0$ is easily obtained from the knowledge of $\beta_{-1}$ given by Remark \ref{RemarkCoeff} so that one can use it directly in the formulas.
\end{remark}

\subsection{Expression of the Hamiltonian}
There are at least $3$ different ways to obtain the expression of the Hamiltonian in terms of the new coordinates of Definition \ref{DefGeometricCoordinates}. 
\begin{itemize}\item Since the change of coordinates in Definition \ref{DefGeometricCoordinates} is time-independent and symplectic, one may simply replace the variables $(\mathbf{q},\mathbf{p})$ in terms of $(\mathbf{Q},\mathbf{P})$ in the expression of the Hamiltonian. In particular, Theorem \ref{TheoHamiltonian} remains valid so that one only needs to express $(H_{p,k})_{p\in \mathcal{R},k\in S_p}$ in terms of the $(\mathbf{Q},\mathbf{P})$ in Proposition \ref{PropExplicitHpk}.
\item From Theorem \ref{GeoLaxMatrices}, we have the explicit expression of the Lax matrices in terms of the new coordinates $(\mathbf{Q},\mathbf{P})$. Thus, one may obtain the evolution equations of the coordinates $(\mathbf{Q},\mathbf{P})$ by solving the compatibility equation of this system. 
\item One may use the gauge transformation (expressed in terms of the coordinates $(\mathbf{Q},\mathbf{P})$ using Lemma \ref{LemmaQExpression}) and connect $(H_{p,k})_{p\in \mathcal{R},k\in S_p}$ in terms of the expansion of $\det \td{L}(\lambda)$ (i.e. the spectrum of $\td{L}$) at each pole. Using the known expression for $\td{L}(\lambda)$ given by Theorem \ref{GeoLaxMatrices}, one may easily obtain the expression of the Hamiltonian in terms of the coordinates $(\mathbf{Q},\mathbf{P})$. 
\end{itemize}

At the computational level, the third approach seems the easiest. Indeed, the first one requires to express $(q_i)_{1\leq i\leq g}$ in terms of $(Q_{p,k})_{p\in \mathcal{R},k\in S_p}$, that is to say to invert Definition \ref{DefGeometricCoordinates} and plug the results into Proposition \ref{PropExplicitHpk}. Although this is possible in principle, it would require substantial computations and identities similar to the ones obtained in \cite{MarchalAlameddineP1Hierarchy2023} for the simpler case of the Painlev\'{e} $1$ hierarchy. The second option requires to solve the compatibility equations, i.e. to redo the works of \cite{MarchalOrantinAlameddine2022} using a different gauge and different coordinates. It is well-known that solving directly the compatibility equations in the geometric gauge is a difficult task and it was precisely to avoid doing so that the authors of \cite{MarchalOrantinAlameddine2022} turned themselves to the oper gauge.
Consequently, we shall develop the third approach in the rest of this paper. As we will see below, this approach is also very convenient to make the connection with the isospectral deformations.

\medskip

It is a straightforward computation from the gauge transformation \eqref{GaugeTransfo} to prove that
\beq \label{DetL}-\det L(\lambda)=L_{2,1}(\lambda)=(\td{L}_{1,1})^2+\td{L}_{2,1}\td{L}_{1,2} +\td{L}_{1,2}\partial_\lambda\left(\frac{\td{L}_{1,1}}{\td{L}_{1,2}}\right)\eeq
which is also recalled in \eqref{L21bis}.
Since the change of coordinates $(\mathbf{q},\mathbf{p})\to (\mathbf{Q},\mathbf{P})$ is time-independent and symplectic, we immediately obtain from Theorem \ref{TheoHamiltonian} the following main result.

\begin{theorem}[Expression of the Hamiltonian in terms of $(\mathbf{Q},\mathbf{P})$]\label{TheoHamNew} We have
\small{\bea H_{X_s,j}(\mathbf{Q},\mathbf{P},\mathbf{t},\mathbf{t}_0)&=&\Res_{\lambda\to X_s}(\lambda-X_s)^{j-1}\left[ (\td{L}_{1,1})^2+\td{L}_{2,1}\td{L}_{1,2} +\td{L}_{1,2}\partial_\lambda\left(\frac{ \td{L}_{1,1}}{\td{L}_{1,2}}\right)\right]\,\,,\,\, \forall \, (s,j)\in \llbracket 1, n\rrbracket\times\llbracket 1, r_s\rrbracket\cr
H_{\infty,j}(\mathbf{Q},\mathbf{P},\mathbf{t},\mathbf{t}_0)&=& -\Res_{\lambda\to \infty}\lambda^{-j-1} \left[ (\td{L}_{1,1})^2+\td{L}_{2,1}\td{L}_{1,2} +\td{L}_{1,2}\partial_\lambda\left(\frac{ \td{L}_{1,1}}{\td{L}_{1,2}}\right)\right] \,\,,\,\forall\, j\in \llbracket 0, r_\infty-4\rrbracket\cr
&&\eea}
\normalsize{where} $(\td{L}_{1,1},\td{L}_{1,2},\td{L}_{2,1})$ are given by Theorem \ref{GeoLaxMatrices}. For any isomonodromic time $\tau$ and the choice of irregular times defined in Proposition \ref{PropTrivialTimes}, the Hamiltonians relatively to the coordinates $(\mathbf{Q},\mathbf{P})$ are given by 
\bea \label{NewHamReduced2}\begin{pmatrix}\text{Ham}^{(\alpha_{t_{\infty,1}})}(\mathbf{Q},\mathbf{P},\mathbf{t},\mathbf{t}_0)\\ \vdots \\ (r_\infty-3)\text{Ham}^{(\alpha_{t_{\infty,r_\infty -3}})}(\mathbf{Q},\mathbf{P},\mathbf{t},\mathbf{t}_0)\end{pmatrix}&=&\left(M_{\infty}(\mathbf{t})\right)^{-1}\begin{pmatrix}H_{\infty,r_\infty-4}(\mathbf{Q},\mathbf{P},\mathbf{t},\mathbf{t}_0)\\ \vdots\\ H_{\infty,0}(\mathbf{Q},\mathbf{P},\mathbf{t},\mathbf{t}_0) \end{pmatrix}\cr
\begin{pmatrix}\text{Ham}^{(\alpha_{t_{X_s,1}})}(\mathbf{Q},\mathbf{P},\mathbf{t},\mathbf{t}_0)\\ \vdots \\ (r_s-1)\text{Ham}^{(\alpha_{t_{X_s,r_s-1}})}(\mathbf{Q},\mathbf{P},\mathbf{t}_0)\end{pmatrix}&=&\left(M_{X_s}(\mathbf{t})\right)^{-1}\begin{pmatrix}H_{X_s,r_s}(\mathbf{Q},\mathbf{P},\mathbf{t},\mathbf{t}_0)\\ \vdots\\  H_{X_s,2}(\mathbf{Q},\mathbf{P},\mathbf{t},\mathbf{t}_0)\end{pmatrix}\,,\,\,\forall\, s\in \llbracket 1,n\rrbracket\cr
\begin{pmatrix}\text{Ham}^{(\alpha_{X_1})}(\mathbf{Q},\mathbf{P},\mathbf{t},\mathbf{t}_0)\\ \vdots \\ \text{Ham}^{(\alpha_{X_n})}(\mathbf{Q},\mathbf{P},\mathbf{t},\mathbf{t}_0)\end{pmatrix}&=&\begin{pmatrix}H_{X_1,1}(\mathbf{Q},\mathbf{P},\mathbf{t},\mathbf{t}_0)\\ \vdots\\ H_{X_n,1}(\mathbf{Q},\mathbf{P},\mathbf{t},\mathbf{t}_0)\end{pmatrix}
\eea
with the matrices $M_{\infty}(\mathbf{t})$ and $\left(M_{X_s}(\mathbf{t})\right)_{1\leq s\leq n}$ defined by \eqref{Minfty} and \eqref{MXs}.
\end{theorem}

\begin{remark}\normalsize{The} matrix $L$ is independent of the choice of normalization of $\td{L}$ because the oper gauge is geometrically independent of this choice. Consequently, $\mathbf{H}$ is independent of this choice too. In particular, for $r_\infty=1$, it may be convenient to use $\mathring{L}_{1,1}(\lambda)$ instead of $\td{L}_{1,1}(\lambda)$ in order to avoid the use of $g_0$ that is more involved. The computation is detailed in \eqref{SuperIdL21}
and we have the alternative expression
\bea L_{2,1}(\lambda)&=&(\mathring{L}_{1,1}(\lambda))^2+\partial_\lambda \mathring{L}_{1,1}(\lambda)-t_{\infty,0}\td{L}_{1,2}(\lambda)-\mathring{L}_{1,1}(\lambda)\frac{\partial_\lambda\td{L}_{1,2}(\lambda)}{\td{L}_{1,2}(\lambda)}+t_{\infty,0}\td{L}_{1,2}(\lambda)\cr
&&+\td{L}_{2,1}(\lambda)\sum_{s=1}^n\left[\frac{\underset{j=r_s+1}{\overset{2r_s}{\sum}}\left(\underset{m=0}{\overset{2r_s-j}{\sum}} t_{X_s,r_s-1-m}t_{X_s,j+m-r_s-1}\right) (\lambda-X_s)^{-j} -\mathring{L}_{1,1}(\lambda)^2}{\td{L}_{1,2}(\lambda)}\right]_{X_s,-}\cr
&&
\eea  
\end{remark}

\subsection{Geometric Lax coordinates and associated expressions for the geometric Lax matrices}\label{SectionGeoLaxCoord}
The previous change of coordinates $(\mathbf{q},\mathbf{p})\to (\mathbf{Q},\mathbf{P})$ is interesting for the Hamiltonian system because it is a time-independent symplectic change of coordinates so that one may simply replace coordinates into the Hamiltonians. However, the price to pay is that entry $\td{L}_{1,1}(\lambda)$ involves a complicated mix of variables $\mathbf{Q}$ and $\mathbf{P}$. In preparation for the upcoming link with isospectral deformations, it is interesting to look for another set of coordinates $(\mathbf{Q},\mathbf{R})$ in which $\td{L}_{1,1}(\lambda)$ would be expressed in a simpler way. This leads to the following definition.

\begin{definition}\label{DefR} Let us define the coordinates $\mathbf{R}:=\{R_{\infty,0},\dots,R_{\infty,r_\infty-3}\}\cup \underset{s=1}{\overset{n}{\bigcup}}\{R_{X_s,1},\dots,R_{X_s,r_s}\}$ by:
\bea R_{X_s,k}&=&\sum_{m=1}^{r_s+1-k} P_{X_s,m}Q_{X_s,k+m-1}-\frac{(g_0+t_{\infty,r_\infty-1}X_s)}{\omega}Q_{X_s,k}\cr
&&-\frac{t_{\infty,r_\infty-1}}{\omega}Q_{X_s,k+1}\delta_{k\leq r_s-1}\,\,,\,\, \forall \, (s,k)\in \llbracket 1,n\rrbracket\times\llbracket 1,r_s\rrbracket\cr
R_{\infty,r_\infty-4}&=&-\omega \, P_{\infty,0}-\frac{g_0}{\omega} Q_{\infty,r_\infty-4}-\frac{t_{\infty,r_\infty-1}}{\omega}Q_{\infty,r_\infty-5} \,\, \text{ if } \, \, r_\infty\geq 4\cr
R_{\infty,k}&=&-\omega \,P_{\infty,r_\infty-4-k}-\sum_{m=0}^{r_\infty-5-k}P_{\infty,m}Q_{\infty,k+1+m}-\frac{t_{\infty,r_\infty-1}}{\omega}Q_{\infty,k-1}\cr
&&-\frac{g_0}{\omega}Q_{\infty,k} \,\,,\,\, \forall \, k\in \llbracket 1, r_\infty-5\rrbracket\cr
R_{\infty,0}&=&-\omega \, P_{\infty,r_\infty-4}-\sum_{m=0}^{r_\infty-5}P_{\infty,m}Q_{\infty,m+1}-\frac{g_0}{\omega}Q_{\infty,0}-\frac{t_{\infty,r_\infty-1}}{\omega}\sum_{s=1}^nQ_{X_s,1}
\eea
with the additional relations
\bea \label{AddconstraintsR}\sum_{s=1}^nR_{X_s,1}&=&-t_{\infty,0} \,\text{ if }\, r_\infty\leq 2\cr
\sum_{s=1}^nX_s R_{X_s,1}+R_{X_s,2}\delta_{r_s\geq 2}
&=&-g_0-\frac{t_{\infty,0}}{\omega}\left(\sum_{s=1}^nX_s^2Q_{X_s,1}+2X_sQ_{X_s,2}+Q_{X_s,3}\right) \cr
&=&\beta_{-1}\,\, \text{ if }\, r_\infty= 1
\eea
We shall name the coordinates $(\mathbf{Q},\mathbf{R})$ the ``geometric Lax coordinates'' since as proved below in Lemma \ref{LemmatdL11}, they are particularly convenient to express the geometric Lax matrix $\td{L}(\lambda)$. 
\end{definition}

Note that the additional constraints \eqref{AddconstraintsR} are direct consequences of the additional constraints for the coordinates $\mathbf{Q}$ in \eqref{AddConstrains} and of the value of $g_0$ provided by \eqref{Valueg0} for $r_\infty\in\{1,2\}$. The previous coordinates are well-suited to express $\td{L}_{1,1}(\lambda)$ because of the following lemma.
\begin{lemma}\label{LemmatdL11} We have
\beq \td{L}_{1,1}(\lambda)=\sum_{s=1}^n\sum_{k=1}^{r_s}\frac{R_{X_s,k}}{(\lambda-X_s)^k}-t_{\infty,r_\infty-1}\delta_{r_\infty\geq 2}\lambda^{r_\infty-2}-t_{\infty,r_\infty-2}\delta_{r_\infty\geq 3}\lambda^{r_\infty-3} +  \sum_{k=0}^{r_\infty-4}R_{\infty,k}\lambda^k\eeq
\end{lemma}

\begin{proof}The proof is obvious from the expression of $\td{L}_{1,1}(\lambda)$ given by Theorem \ref{GeoLaxMatrices} but is detailed for completeness in Appendix \ref{AppendixLemmatdL11}.
\end{proof}

For completeness, we shall mention that the inverse change of coordinates amounts to
\small{\beq \label{InversChangeCoo}\begin{pmatrix}R_{\infty,r_\infty-4}\\ \vdots\\ R_{\infty,1}\\ R_{\infty,0}\end{pmatrix}=-\begin{pmatrix}\omega&0&\dots&0\\
Q_{\infty,r_\infty-4}&\omega& &0\\
\vdots& \ddots &\ddots &\vdots \\
Q_{\infty,1}&\dots&Q_{\infty,r_\infty-4}&\omega
\end{pmatrix}\begin{pmatrix} P_{\infty,0}\\ \vdots\\ P_{\infty,r_\infty-5}\\ P_{\infty, r_\infty-4}\end{pmatrix}-\frac{g_0}{\omega}\begin{pmatrix}Q_{\infty,r_\infty-4}\\ \vdots\\ Q_{\infty,1}\\ Q_{\infty,0}\end{pmatrix}-\frac{t_{\infty,r_\infty-1}}{\omega}\begin{pmatrix}Q_{\infty,r_\infty-5}\\ \vdots\\ Q_{\infty,0}\\ \underset{s=1}{\overset{n}{\sum}}Q_{X_s,1}\end{pmatrix}
\eeq}
and for all $s\in\llbracket 1,n\rrbracket$:
\small{\bea \label{InversChangeCoo2}\begin{pmatrix}R_{X_s,r_s}\\  R_{X_s,r_s-1}\\ \vdots\\ R_{X_s,1}\end{pmatrix}&=&\begin{pmatrix}Q_{X_s,r_s}&0&\dots& &\dots &0\\
Q_{X_s,r_s-1}&Q_{X_s,r_s}& 0& & &\vdots\\
\vdots & \ddots&\ddots &\ddots  & &\vdots\\
\vdots &\ddots&\ddots&\ddots&0&\vdots\\
Q_{X_s,2}&\ddots &\ddots&\ddots& Q_{X_s,r_s}&0\\
Q_{X_s,1}&Q_{X_s,2}& \dots & & Q_{X_s,r_s-1}& Q_{X_s,r_s}
 \end{pmatrix}\begin{pmatrix}P_{X_s,1}\\P_{X_s,2} \\ \vdots\\ P_{X_s,r_s}\end{pmatrix} \cr
&&-\frac{(g_0+t_{\infty,r_\infty-1}X_s)}{\omega}\begin{pmatrix}Q_{X_s,r_s}\\Q_{X_s,r_s-1}\\ \vdots\\ Q_{X_s,1}\end{pmatrix}-\frac{t_{\infty,r_\infty-1}}{\omega}\begin{pmatrix} 0\\Q_{X_s,r_s}\\ \vdots\\ Q_{X_s,2}\end{pmatrix}
\eea}

\normalsize{}
\begin{remark} \label{RemarkNotSymplectic} \normalsize{Let} us remark that the change of coordinates $(\mathbf{Q},\mathbf{P})\leftrightarrow (\mathbf{Q},\mathbf{R})$ is time-independent but \textbf{not symplectic}. Therefore, one may not obtain the Hamiltonians for the coordinates $(\mathbf{Q},\mathbf{R})$ by just replacing the coordinates $(\mathbf{Q},\mathbf{P})$ in terms of the coordinates $(\mathbf{Q},\mathbf{R})$ in the Hamiltonians of Theorem \ref{TheoHamNew}.   
\end{remark}

As we will see in Section \ref{SectionConnectionIso}, the coordinates $(\mathbf{Q},\mathbf{R})$ are particularly convenient to make the connections with isospectral deformations. For completeness, we also provide the expression of the geometric Lax matrices in terms of $(\mathbf{Q},\mathbf{R})$ in the following theorem.

\begin{theorem}[Expression of the geometric Lax matrices in terms of $(\mathbf{Q},\mathbf{R})$]\label{GeoLaxMatricesQR}We have:
\bea \td{L}_{1,1}(\lambda)&=&\sum_{s=1}^n\sum_{k=1}^{r_s}\frac{R_{X_s,k}}{(\lambda-X_s)^k}-t_{\infty,r_\infty-1}\delta_{r_\infty\geq 2}\lambda^{r_\infty-2}-t_{\infty,r_\infty-2}\delta_{r_\infty\geq 3}\lambda^{r_\infty-3} +  \sum_{k=0}^{r_\infty-4}R_{\infty,k}\lambda^k\cr
 \td{L}_{1,2}(\lambda)&=&\sum_{s=1}^n\sum_{k=1}^{r_s} \frac{Q_{X_s,k}}{(\lambda-X_s)^k}+\sum_{k=0}^{r_\infty-4} Q_{\infty,k}\lambda^k+\omega\delta_{r_\infty\geq 3}\lambda^{r_\infty-3}\cr
\td{L}_{2,2}(\lambda)&=&-\td{L}_{1,1}(\lambda)
\eea
Expression of $\td{L}_{2,1}(\lambda)$ depends on the value of $r_\infty$ and is given by \eqref{tdL21rinftygeq3} for $r_\infty\geq 3$, \eqref{tdL21rinftyequal2} for $r_\infty\geq 2$. For $r_\infty=1$, we have
\bea \td{L}_{2,1}(\lambda)&=&\frac{t_{\infty,0}}{\omega}-\frac{(t_{\infty,0})^2}{\omega}\cr
&&+\sum_{s=1}^n\left[\frac{\underset{j=r_s+1}{\overset{2r_s}{\sum}}\left(\underset{m=0}{\overset{2r_s-j}{\sum}} t_{X_s,r_s-1-m}t_{X_s,j+m-r_s-1}\right) (\lambda-X_s)^{-j} -\td{L}_{1,1}(\lambda)^2}{\underset{s=1}{\overset{n}{\sum}}\underset{k=1}{\overset{r_s}{\sum}} \frac{Q_{X_s,k}}{(\lambda-X_s)^k}}\right]_{X_s,-}\cr
g_0&=&-\sum_{s=1}^nX_sR_{X_s,1}+ R_{X_s,2}\delta_{r_s\geq 2}-\frac{t_{\infty,0}}{\omega}\left(\sum_{s=1}^n X_s^2Q_{X_s,1}+2X_sQ_{X_s,2}+Q_{X_s,3}\right)\cr
&&
\eea
The auxiliary matrix $\td{A}_{\boldsymbol{\alpha}}(\lambda)$ is given by
\bea \left[\td{A}_{\boldsymbol{\alpha}}(\lambda)\right]_{1,2}&=&\omega \,\nu_{\infty,1}^{(\boldsymbol{\alpha})}\lambda^{r_\infty-4}\delta_{r_\infty\geq 4}
+\sum_{j=0}^{r_\infty-5}\left(\omega \,\nu_{\infty,r_\infty-3-j}^{(\boldsymbol{\alpha})}+\sum_{k=j+1}^{r_\infty-4}\nu_{\infty,k-j}^{(\boldsymbol{\alpha})}Q_{\infty,k}\right)\lambda^j\cr
&&+\sum_{s=1}^n\sum_{j=1}^{r_s}\left(\sum_{k=j}^{r_s} \nu_{X_s,k-j}^{(\boldsymbol{\alpha})}  Q_{X_s,k}\right)(\lambda-X_s)^{-j}\cr
\left[\td{A}_{\boldsymbol{\alpha}}(\lambda)\right]_{1,1}&=&
\frac{1}{2\omega}\mathcal{L}_{\boldsymbol{\alpha}}[\omega] -t_{\infty,1}\nu^{(\boldsymbol{\alpha})}_{\infty,0}\delta_{r_\infty,2} +\left(\frac{1}{2}-t_{\infty,0}\right)\nu^{(\boldsymbol{\alpha})}_{\infty,-1}\delta_{r_\infty,1}\cr
&&-t_{\infty,r_\infty-1}\sum_{j=0}^{r_\infty-3}\nu_{\infty,r_\infty-2-j}^{(\boldsymbol{\alpha})}\lambda^{j}-t_{\infty,r_\infty-2}\sum_{j=0}^{r_\infty-4}\nu_{\infty,r_\infty-3-j}^{(\boldsymbol{\alpha})}\lambda^{j}\cr
&&+\sum_{j=0}^{r_\infty-5}\left(\sum_{i=1}^{r_\infty-4-j}\nu_{\infty,i}^{(\boldsymbol{\alpha})}R_{\infty,j+i}\right)\lambda^j+\sum_{s=1}^n\sum_{j=1}^{r_s}\left(\sum_{i=0}^{r_s-j}  \nu_{X_s,i}^{(\boldsymbol{\alpha})}R_{X_s,i+j} \right) (\lambda-X_s)^{-j}\cr
\left[\td{A}_{\boldsymbol{\alpha}}(\lambda)\right]_{2,2}&=&-\left[\td{A}_{\boldsymbol{\alpha}}(\lambda)\right]_{1,1}
\eea
with the extra conditions \eqref{ExtraConditionsrinftyequal2} and \eqref{ExtraConditionsrinftyequal1} for $r_\infty\leq 2$ while coefficients $\left(\nu_{\infty,r_\infty-2}^{(\boldsymbol{\alpha})},\nu_{\infty,r_\infty-1}^{(\boldsymbol{\alpha})}\right)$ are determined by \eqref{nurinftminus2} for $r_\infty\geq 3$ . Finally entry $\left[\td{A}_{\boldsymbol{\alpha}}(\lambda)\right]_{2,1}$ is determined by \eqref{tdA21rinftygeq4}, \eqref{tdA21rinftyequal3}, \eqref{tdA21rinftyequal2} or \eqref{tdA21rinftyequal1} depending on the value of $r_\infty$.
\end{theorem}

\begin{proof}Computations are straightforward for $\td{L}(\lambda)$. The only non-trivial case is $r_\infty=1$. For $r_\infty=1$, the value of $g_0$ is given by \eqref{betaMinus1} and the expression of $\td{L}_{2,1}(\lambda)$ follows from Section \ref{SectionAppL21}. Computations for $[\td{A}_{\boldsymbol{\alpha}}(\lambda)]_{1,1}$ are done in Appendix \ref{AppendixRtdA}.
\end{proof}

If the Lax matrices $\td{L}(\lambda)$ and $\td{A}_{\boldsymbol{\alpha}}(\lambda)$ have some nice expressions in terms of the geometric Lax coordinates $(\mathbf{Q},\mathbf{R})$, obtaining the Hamiltonian evolutions relatively to these coordinates is not that simple. Indeed, the change of coordinates $(\mathbf{Q},\mathbf{P})\to (\mathbf{Q},\mathbf{R})$ is not symplectic so that one cannot use the results on $(\mathbf{Q},\mathbf{P})$ directly. Solving the compatibility equations in these coordinates would require substantial work equivalent to the one done in \cite{MarchalOrantinAlameddine2022}.

\section{Connection with isospectral deformations}\label{SectionConnectionIso}
In the previous section, a direct approach to isomonodromic deformations was used to obtain the underlying Hamiltonian structure. The strategy consisted 
\begin{enumerate}\item Go to the oper gauge and use appropriate Darboux coordinates to describe the Lax matrices (Proposition \eqref{OperForm}).
\item Solve the isomonodromic compatibility equation in the oper gauge to obtain the Hamiltonian system (Theorem \ref{TheoHamiltonian}).
\item Get back to the initial gauge and perform a time-independent symplectic change of coordinates to write the Hamiltonian and Lax matrices in this gauge (Theorems \ref{GeoLaxMatrices} and \ref{TheoHamNew}).
\end{enumerate}

If this strategy provides a direct way to obtain the symplectic structure, another approach exists in the literature to tackle isomonodromic deformations. Indeed, the historical approach of R. Fuchs (for Fuchsian singularities) pursued by the Japanese school (M. Jimbo, T. Miwa, M. Sato, K. Ueno, etc.) and by the Montr\'{e}al school (J. Harnad, J. Hurtubise, M. Bertola, etc.) and also by many others consists in using first isospectral deformations and then impose some additional constraints to match them with isomonodromic deformations. It turns out that for Fuchsian singularities the isospectral deformations are identical to the isomonodromic deformations making this strategy a very powerful method to obtain the symplectic structure and the Lax pairs. For non-Fuchsian singularities, the situation is more complicated and the additional constraints are non-trivial. Many cases, including all six Painlev\'{e} equations have been dealt with and other case by case studies can be found in the literature. Recently, the connection with isospectral deformations in $\mathfrak{sl}_d(\mathbb{C})$ has been made complete by the Montr\'{e}al school in \cite{BertolaHarnadHurtubise2022}. Their main result (simplified to the $\mathfrak{sl}_2(\mathbb{C})$ setting) is summarized in the following theorem.

\begin{theorem}[Theorems $3.2$, $3.3$ and $4.6$ of \cite{BertolaHarnadHurtubise2022}]\label{MontrealResults} Let $\hat{L}(\lambda)\in \hat{F}_{\mathcal{R},\mathbf{r}}$ defining the meromorphic connection $\partial_\lambda \hat{\Psi}(\lambda)= \hat{L}(\lambda)\hat{\Psi}(\lambda)$. Define $\left(\lambda_+(\lambda),-\lambda_+(\lambda)\right)$ the eigenvalues of $\hat{L}(\lambda)$. The spectral invariants $\mathbf{I}:=\left(I_{p,k}\right)_{p\in \mathcal{R},k\in \llbracket 1, r_p-1\rrbracket}$ are defined using the expansion of the eigenvalues at each pole:
\bea \label{lambdaplusdef}\lambda_{+}(\lambda)&\overset{\lambda\to\infty}{:=}&\sum_{j=1}^{r_\infty-1} t_{\infty,j} \lambda^{j-1}+t_{\infty,0}\lambda^{-1} +\sum_{j=1}^{r_\infty-1} j I_{\infty,j}\lambda^{-j-1}+O(\lambda^{-r_\infty-1})\cr
\lambda_{+}(\lambda)&\overset{\lambda\to X_s}{:=}&-\sum_{j=1}^{r_s-1} t_{X_s,j} (\lambda-X_s)^{-(j+1)}- t_{X_s,0}(\lambda-X_s)^{-1} \cr
&&-\sum_{j=1}^{r_s-1}j I_{X_s,j}(\lambda-X_s)^{j-1}+O((\lambda-X_s)^{r_s-1}) \,\,\,,\,\,\, \forall \, s\in \llbracket 1,n\rrbracket\cr
&&
\eea
Let $\hat{A}_{\boldsymbol{\alpha}}(\lambda)$ be the associated rational matrix satisfying $\mathcal{L}_{\boldsymbol{\alpha}}[\hat{\Psi}(\lambda)]=\hat{A}_{\boldsymbol{\alpha}}(\lambda)\hat{\Psi}(\lambda)$. Let us define the ``isospectral condition''
\beq \label{IsoCondition}\delta^{(\boldsymbol{\alpha})}_{\mathbf{t}} \hat{L}(\lambda)=\partial_\lambda \hat{A}_{\boldsymbol{\alpha}}(\lambda)\eeq
where $\delta^{(\boldsymbol{\alpha})}_{\mathbf{t}}$ is the exterior derivative relatively to the irregular times and positions of poles (i.e. we do not derive the Darboux coordinates but only the explicit dependence relatively to the irregular times and positions of poles):
\beq \label{DefDelta}\delta^{(\boldsymbol{\alpha})}_{\mathbf{t}}=\sum_{i=1}^{r_\infty-1} \alpha_{\infty,i} \delta_{t_{\infty,i}}+ \sum_{s=1}^n\sum_{k=1}^{r_s-1}\alpha_{X_s,k} \delta_{t_{X_s,k}}+\sum_{s=1}^n \alpha_{X_s} \partial_{X_s}
\eeq
If the isospectral condition \eqref{IsoCondition} is satisfied then for any $(p,k)\in \mathcal{R}\times \llbracket 1, r_p-1\rrbracket$, the spectral invariant $I_{p,k}$ identifies with the Hamiltonian corresponding to the isomonodromic deformation relatively to $t_{p,k}$. We shall name ``isospectral Darboux coordinates'', any set of Darboux coordinates $(\mathbf{u},\mathbf{v})$ satisfying the isospectral condition \eqref{IsoCondition}.
\end{theorem}  

\begin{remark}\normalsize{ In} \cite{BertolaHarnadHurtubise2022}, the authors do not deal with deformations relatively to the position of the finite poles. However, these deformations that correspond to the usual Fuchsian deformations have already been understood for a long time so that only deformations relatively to irregular times remain difficult to describe. For completeness, we shall include the deformations relatively to the position of the poles in this article. 
\end{remark}

In \cite{BertolaHarnadHurtubise2022}, the authors observed that the main difficulty of this strategy is that the isospectral condition \eqref{IsoCondition} fixes the appropriate Darboux coordinates (up to a trivial time-independent change of coordinates) but that in practice, obtaining some isospectral Darboux coordinates is difficult and left this question opened. We shall address this question providing the link between our geometric Lax coordinates $(\mathbf{Q},\mathbf{R})$ and some isospectral Darboux coordinates $(\mathbf{u},\mathbf{v})$ (Cf. Theorem \ref{MainTheoIsospectral}).

\subsection{Expression of $\det \td{L}$ and verification on irregular times and monodromies}
The formalism presented in \cite{BertolaHarnadHurtubise2022} uses the eigenvalues of $\td{L}(\lambda)$ to define the spectral invariants $\mathbf{I}$. In $\mathfrak{sl}_{2}(\mathbb{C})$, one may relate them to $\det \td{L}(\lambda)$ and then using the gauge transformation \eqref{GaugeTransfo}, relate them to our set $\mathbf{H}$. Indeed, it is straightforward to observe that
\bea \label{Equivalence}\det \td{L}&=&\det(L+ G(\partial_\lambda (G^{-1})))=-L_{2,1}-[G (\partial_\lambda (G^{-1}))]_{2,1}=-L_{2,1}+\td{L}_{1,2}\partial_\lambda\left(\frac{ \td{L}_{1,1}}{\td{L}_{1,2}}\right)\cr
&=&\td{P}_2(\lambda) -\sum_{j=0}^{r_\infty-4}H_{\infty,j}\lambda^j-\sum_{s=1}^n\sum_{j=1}^{r_s}H_{X_s,j}(\lambda-X_s)^{-j}+ \delta_{r_\infty\geq 3}t_{\infty,r_\infty-1}\lambda^{r_\infty-3}+\sum_{j=1}^{g} \frac{p_j}{\lambda-q_j}\cr
&&+\td{L}_{1,2}\partial_\lambda\left(\frac{ \td{L}_{1,1}}{\td{L}_{1,2}}\right)
\eea

Since $\td{L}$ is rational in $\lambda$ with poles in $\{\infty,X_1,\dots,X_n\}$, $\det \td{L}$ is a rational function of $\lambda$ with poles in $\{\infty,X_1,\dots,X_n\}$. Thus, we only need to obtain the singular part at each pole in \eqref{Equivalence} to determine $\det \td{L}$. Let us first observe that the term $\underset{j=1}{\overset{g}{\sum}} \frac{p_j}{\lambda-q_j}$ is regular at each pole and hence does not contribute to $\det \td{L}$. Thus, we only need to obtain the behavior of $\td{L}_{1,2}\partial_\lambda\left(\frac{ \td{L}_{1,1}}{\td{L}_{1,2}}\right)$ at each pole to complete the computation of $\det L$.

\begin{proposition}[Expression of $\det \td{L}$ in terms of geometric Darboux coordinates]\label{PropDettdL} We have
\bea \det \td{L}(\lambda)&=& \td{P}_2(\lambda) -\sum_{j=0}^{r_\infty-4}H_{\infty,j}\lambda^j-\sum_{s=1}^n\sum_{j=1}^{r_s}H_{X_s,j}(\lambda-X_s)^{-j}+ \delta_{r_\infty\geq 3}t_{\infty,r_\infty-1}\lambda^{r_\infty-3}\cr
&&+ \left[\td{L}_{1,2}\partial_\lambda\left(\frac{ \td{L}_{1,1}}{\td{L}_{1,2}}\right) \right]_{\infty,+}+\sum_{s=1}^n \left[\td{L}_{1,2}\partial_\lambda\left(\frac{ \td{L}_{1,1}}{\td{L}_{1,2}}\right)\right]_{X_s,-}
\eea
with $\td{P}_2(\lambda)$ given by \eqref{tdP2} and $\td{L}_{1,2}(\lambda)$, $\td{L}_{1,1}(\lambda)$ given by either Theorem \ref{GeoLaxMatrices} or Theorem \ref{GeoLaxMatricesQR} depending on the choice of Darboux coordinates.
\end{proposition}

The second purpose of this section is to verify that our irregular times $\mathbf{t}$ and monodromies $\mathbf{t_0}$ identify with those of \cite{BertolaHarnadHurtubise2022}. We have from Theorem \ref{MontrealResults}:

\bea\label{MontrealInfinity} \det\td{L}(\lambda)&=& -\lambda_+(\lambda)^2\overset{\lambda\to\infty}{=}-\left(\sum_{j=0}^{r_\infty-1} t_{\infty,j} \lambda^{j-1}+\sum_{j=1}^{r_\infty-1} j I_{\infty,j}\lambda^{-j-1}+O(\lambda^{-r_\infty-1})\right)^2\cr
&\overset{\lambda\to\infty}{=}& -\sum_{i=0}^{r_\infty-1}\sum_{j=0}^{r_\infty-1} t_{\infty,i}t_{\infty,j}\lambda^{i+j-2}-2\sum_{j=1}^{r_\infty-1}\sum_{i=0}^{r_\infty-1} j I_{\infty,j} t_{\infty,i}\lambda^{i-j-2}+O(\lambda^{-2}) \cr
&\overset{\lambda\to\infty}{=}&-\sum_{k=r_\infty-3}^{2r_\infty-4}\sum_{j=k+3-r_\infty}^{r_\infty-1} t_{\infty,k+2-j}t_{\infty,j}\lambda^{k}-\sum_{k=0}^{r_\infty-4}\sum_{j=0}^{k+2} t_{\infty,k+2-j}t_{\infty,j}\lambda^{k}\cr
&&-2\sum_{k=-1}^{r_\infty-4}\sum_{j=1}^{r_\infty-3-k} j I_{\infty,j} t_{\infty,j+k+2}\lambda^{k}+O(\lambda^{-2})
\eea

We have from \eqref{tdP22}:
\beq P_{\infty,k}^{(2)}=-\sum_{j=k+3-r_\infty}^{r_\infty-1} t_{\infty,j}t_{\infty,k-j+2}\,\, ,\,\, \forall\, k\in \llbracket r_\infty-3,2r_\infty-4 \rrbracket\eeq
so that from Proposition \ref{PropDettdL} we also have:
\bea \det \td{L}(\lambda)&=& \td{P}_2(\lambda) -\sum_{j=0}^{r_\infty-4}H_{\infty,j}\lambda^j-\sum_{s=1}^n\sum_{j=1}^{r_s}H_{X_s,j}(\lambda-X_s)^{-j}+ \delta_{r_\infty\geq 3}t_{\infty,r_\infty-1}\lambda^{r_\infty-3}\cr
&&+ \left[\td{L}_{1,2}\partial_\lambda\left(\frac{ \td{L}_{1,1}}{\td{L}_{1,2}}\right) \right]_{\infty,+}+\sum_{s=1}^n \left[\td{L}_{1,2}\partial_\lambda\left(\frac{ \td{L}_{1,1}}{\td{L}_{1,2}}\right)\right]_{X_s,-}\cr
&\overset{\lambda\to\infty}{=}&-\sum_{j=k+3-r_\infty}^{r_\infty-1} t_{\infty,j}t_{\infty,k-j+2}+ \delta_{r_\infty\geq 3}t_{\infty,r_\infty-1}\lambda^{r_\infty-3}- \delta_{r_\infty\geq 3}t_{\infty,r_\infty-1}\lambda^{r_\infty-3}\cr
&&+O(\lambda^{r_\infty-4})\cr
&\overset{\lambda\to\infty}{=}&-\sum_{j=k+3-r_\infty}^{r_\infty-1} t_{\infty,j}t_{\infty,k-j+2}+O(\lambda^{r_\infty-4})
\eea
because $\td{L}_{1,1}=-t_{\infty,r_\infty-1}\lambda^{r_\infty-2}+O(\lambda^{r_\infty-3})$ and $\td{L}_{1,2}=\omega \lambda^{r_\infty-3}+O(\lambda^{r_\infty-4})$.
Thus, for $r_\infty\geq 3$, coefficients of order $\lambda^{r_\infty-3}$ up to $\lambda^{2r_\infty-4}$ match, meaning that the coefficients appearing in the expansion of $\lambda_+(\lambda)$ at infinity coincide with our irregular times and monodromy at infinity.

A similar computation can be carried out for any $s\in \llbracket 1,n\rrbracket$:
\small{\bea \label{MontrealXs}&&\det\td{L}(\lambda)
\overset{\lambda\to X_s}{=}-\left(\sum_{j=0}^{r_s-1} t_{X_s,j} (\lambda-X_s)^{-(j+1)} +\sum_{j=1}^{r_s-1}j I_{X_s,j}(\lambda-X_s)^{j-1}+O((\lambda-X_s)^{r_s-1})\right)^2\cr
&& -\sum_{i=0}^{r_s-1}\sum_{j=0}^{r_s-1} t_{X_s,j}t_{X_s,i} (\lambda-X_s)^{-(j+i+2)} -2\sum_{i=0}^{r_s-1}\sum_{j=1}^{r_s-1}j I_{X_s,j} t_{X_s,i} (\lambda-X_s)^{j-i-2}+O((\lambda-X_s)^{-1})\cr
&&=-\sum_{k=r_s+1}^{2r_s}\sum_{j=k-r_s-1}^{r_s-1} t_{X_s,j}t_{X_s,k-j-2} (\lambda-X_s)^{-k}-\sum_{k=2}^{r_s}\sum_{j=0}^{k-2} t_{X_s,j}t_{X_s,k-j-2} (\lambda-X_s)^{-k}\cr
&&-2\sum_{k=2}^{r_s}\sum_{j=1}^{r_s+1-k}j I_{X_s,j} t_{X_s,k+j-2} (\lambda-X_s)^{-k}+O((\lambda-X_s)^{-1})
\eea}
\normalsize{but} we have from \eqref{tdP22}:
\beq P_{X_s,k}^{(2)}=\sum_{j=k-r_s-1}^{r_s-1} t_{X_s,j}t_{X_s,k-j-2} \,\,,\,\, \forall k\in \llbracket r_s+1,2r_s\rrbracket\eeq 
so that from Proposition \eqref{PropDettdL} we also have:
\bea \det\td{L}(\lambda)&=& \td{P}_2(\lambda) -\sum_{j=0}^{r_\infty-4}H_{\infty,j}\lambda^j-\sum_{s=1}^n\sum_{j=1}^{r_s}H_{X_s,j}(\lambda-X_s)^{-j}+ \delta_{r_\infty\geq 3}t_{\infty,r_\infty-1}\lambda^{r_\infty-3}\cr
&&+ \left[\td{L}_{1,2}\partial_\lambda\left(\frac{ \td{L}_{1,1}}{\td{L}_{1,2}}\right) \right]_{\infty,+}+\sum_{s=1}^n \left[\td{L}_{1,2}\partial_\lambda\left(\frac{ \td{L}_{1,1}}{\td{L}_{1,2}}\right)\right]_{X_s,-}\cr
&\overset{\lambda\to X_s}{=}& -\sum_{k=r_s+1}^{2r_s}\sum_{j=k-r_s-1}^{r_s-1} t_{X_s,j}t_{X_s,k-j-2} (\lambda-X_s)^{-k} +O((\lambda-X_s)^{-r_s})
\eea
because $\td{L}_{1,2}\overset{\lambda\to X_s}{=}O((\lambda-X_s)^{-r_s})$ and $\frac{ \td{L}_{1,1}}{\td{L}_{1,2}}\overset{\lambda\to X_s}{=}O(1)$.
Thus, coefficients of order $(\lambda-X_s)^{-(r_s+1)}$ up to $(\lambda-X_s)^{-2r_s}$ match, meaning that the coefficients appearing in the expansion of $\lambda_+(\lambda)$ at $\lambda\to X_s$ coincide with our irregular times and monodromy at $X_s$.

\begin{remark} \normalsize{The} situation at infinity for $r_\infty\leq 2$ requires special attention. Indeed, we have for $r_\infty=2$:
\beq \label{SpecialdettdL}\det\td{L}(\lambda)=-\left(t_{\infty,1} +t_{\infty,0}\lambda^{-1} +O(\lambda^{-2})\right)^2\overset{\lambda\to\infty}{=}-(t_{\infty,1})^2-2t_{\infty,1}t_{\infty,0}\lambda^{-1}+O(\lambda^{-2})
\eeq
From \eqref{Equivalence} and \eqref{ConditionsAddrinftyequal2} we have:
\bea\label{SpecialdettdL2} \det \td{L}(\lambda)&=&-(t_{\infty,1})^2 -\left(\sum_{s=1}^nH_{X_s,1} -\sum_{j=1}^n p_j\right)\lambda^{-1}+ \td{L}_{1,2}\partial_\lambda\left(\frac{ \td{L}_{1,1}}{\td{L}_{1,2}}\right) +O(\lambda^{-2})\cr
&=&-(t_{\infty,1})^2 +(-2t_{\infty,1}t_{\infty,0}+t_{\infty,1})\lambda^{-1}+ \td{L}_{1,2}\partial_\lambda\left(\frac{ \td{L}_{1,1}}{\td{L}_{1,2}}\right)+O(\lambda^{-2})\cr
&=&-(t_{\infty,1})^2-2t_{\infty,1}t_{\infty,0}\lambda^{-1} +O(\lambda^{-2})
\eea
because $\td{L}_{1,2}\partial_\lambda\left(\frac{ \td{L}_{1,1}}{\td{L}_{1,2}}\right)=(\omega \lambda^{-1}+O(\lambda^{-2}))\partial_\lambda\left(-\frac{t_{\infty,r_\infty-1}}{\omega}\lambda+O(1)\right)=-t_{\infty,r_\infty-1}\lambda^{-1}$. Thus, comparing \eqref{SpecialdettdL} and \eqref{SpecialdettdL2} implies that the time $t_{\infty,1}$ and the monodromy $t_{\infty,0}$ coincide in both setups.

For $r_\infty=1$, we have:
\beq \label{SpecialdettdL3}\det\td{L}(\lambda)=-\left(t_{\infty,0}\lambda^{-1} +O(\lambda^{-2})\right)^2=-(t_{\infty,0})^2\lambda^{-2}+O(\lambda^{-3})\eeq
From \eqref{Equivalence} and \eqref{ConditionsAddrinftyequal1} we have:
\bea \det \td{L}(\lambda)&=&-\left(-\sum_{s=1}^nH_{X_s,1} +\sum_{j=1}^n p_j\right)\lambda^{-1}\cr
&&-\left(\sum_{s=1}^n(H_{X_s,2}\delta_{r_s\geq 2}+X_sH_{X_s,1}) -\sum_{j=1}^n q_jp_j-\sum_{s=1}^n (t_{X_s,0})^2\delta_{r_s=1} \right)\lambda^{-2}\cr
&&+ \td{L}_{1,2}\partial_\lambda\left(\frac{ \td{L}_{1,1}}{\td{L}_{1,2}}\right)+O(\lambda^{-3})\cr
&=&-\left(t_{\infty,0}(t_{\infty,0}-1)\right)\lambda^{-2}+\td{L}_{1,2}\partial_\lambda\left(\frac{ \td{L}_{1,1}}{\td{L}_{1,2}}\right)+O(\lambda^{-3})
\eea
Since 
\beq \td{L}_{1,2}\partial_\lambda\left(\frac{ \td{L}_{1,1}}{\td{L}_{1,2}}\right)=(\omega \lambda^{-2}+O(\lambda^{-3}))\partial_{\lambda}\left( -\frac{t_{\infty,0}}{\omega}\lambda+O(1)\right)=-t_{\infty,0}\lambda^{-2}+O(\lambda^{-3})\eeq
we end up with
\beq \label{SpecialdettdL4}\det \td{L}(\lambda)=-(t_{\infty,0})^2\lambda^{-2}+O(\lambda^{-3}) \eeq
Thus, comparing \eqref{SpecialdettdL3} and \eqref{SpecialdettdL4} implies that the monodromy $t_{\infty,0}$ coincides in both setups.
\end{remark}

Hence, we conclude that \textbf{the irregular times $\mathbf{t}$ and monodromies $\mathbf{t_0}$ defined in this article identify exactly with those of \cite{BertolaHarnadHurtubise2022}}. The next step is to relate the spectral invariants $\mathbf{I}$ of \cite{BertolaHarnadHurtubise2022} with our set $\mathbf{H}$.

\subsection{Relation between spectral invariants $\mathbf{I}$ and $\mathbf{H}$}
The purpose of this section is to relate the spectral invariants $\mathbf{I}$ defined by \eqref{lambdaplusdef} with our set $\mathbf{H}$ defined by \eqref{L21}. In order to obtain the relation, let us look at the next orders of the expansion of $\det \td{L}(\lambda)$ at each pole.
For any $s\in \llbracket 1,n\rrbracket$, we have from Proposition \ref{PropDettdL}:
\bea\label{IdXs1} \det \td{L}(\lambda)&\overset{\lambda\to X_s}{=}& -\sum_{k=r_s+1}^{2r_s}\sum_{j=k-r_s-1}^{r_s-1} t_{X_s,j}t_{X_s,k-j-2} (\lambda-X_s)^{-k}-\sum_{k=1}^{r_s}H_{X_s,k}(\lambda-X_s)^{-k}\cr
&&+ \left[\td{L}_{1,2}\partial_\lambda\left(\frac{ \td{L}_{1,1}}{\td{L}_{1,2}}\right)\right]_{X_s,-}+O(1)\cr
&&
\eea
while from \eqref{MontrealXs} we have 
\bea \label{IdXs2} \det\td{L}(\lambda)&\overset{\lambda\to X_s}{=}&-\sum_{k=r_s+1}^{2r_s}\sum_{j=k-r_s-1}^{r_s-1} t_{X_s,j}t_{X_s,k-j-2} (\lambda-X_s)^{-k}-\sum_{k=2}^{r_s}\sum_{j=0}^{k-2} t_{X_s,j}t_{X_s,k-j-2} (\lambda-X_s)^{-k}\cr
&&-2\sum_{k=2}^{r_s}\sum_{j=1}^{r_s+1-k}j I_{X_s,j} t_{X_s,k+j-2} (\lambda-X_s)^{-k}+O((\lambda-X_s)^{-1})
\eea
Hence, identifying \eqref{IdXs1} and \eqref{IdXs2}, we get for all $k\in \llbracket 2, r_s\rrbracket$:
\beq H_{X_s,k}=\Res_{\lambda\to X_s}(\lambda-X_s)^{k-1}\td{L}_{1,2}(\lambda)\partial_\lambda\left(\frac{ \td{L}_{1,1}(\lambda)}{\td{L}_{1,2}(\lambda)}\right)+\sum_{j=0}^{k-2} t_{X_s,j}t_{X_s,k-j-2}+2\sum_{j=1}^{r_s+1-k}j I_{X_s,j} t_{X_s,k+j-2}
\eeq
These relations may be rewritten into a matrix form:
\small{\beq\label{MatrixIHXs} 2M_{X_s}(\mathbf{t})\begin{pmatrix}I_{X_s,1}\\2I_{X_s,2}\\\vdots\\ (r_{s}-1)I_{X_s,r_s-1} \end{pmatrix}= \begin{pmatrix}H_{X_s,r_s}\\ H_{X_s,r_s-1}\\\vdots \\ H_{X_s,2}\end{pmatrix} -\begin{pmatrix}\underset{j=0}{\overset{r_s-2}{\sum}} t_{X_s,j}t_{X_s,r_s-j-2}+\underset{\lambda\to X_s}{\Res}(\lambda-X_s)^{r_s-1}\td{L}_{1,2}(\lambda)\partial_\lambda\left(\frac{ \td{L}_{1,1}(\lambda)}{\td{L}_{1,2}(\lambda)}\right)\\
\underset{j=0}{\overset{r_s-3}{\sum}} t_{X_s,j}t_{X_s,r_s-j-3}+\underset{\lambda\to X_s}{\Res}(\lambda-X_s)^{r_s-2}\td{L}_{1,2}(\lambda)\partial_\lambda\left(\frac{ \td{L}_{1,1}(\lambda)}{\td{L}_{1,2}(\lambda)}\right)\\
\vdots\\
(t_{X_s,0})^2+\underset{\lambda\to X_s}{\Res}(\lambda-X_s)\td{L}_{1,2}(\lambda)\partial_\lambda\left(\frac{ \td{L}_{1,1}(\lambda)}{\td{L}_{1,2}(\lambda)}\right)
  \end{pmatrix}
\eeq}

\normalsize{At} infinity, for $r_\infty\geq 3$, we have from Proposition \ref{PropDettdL}:
\bea \label{IdInfty1}\det \td{L}(\lambda)&\overset{\lambda\to \infty}{=}& -\sum_{k=r_\infty-3}^{2r_\infty-4}\sum_{j=k+3-r_\infty}^{r_\infty-1} t_{\infty,j}t_{\infty,k-j+2}\lambda^k -\sum_{k=0}^{r_\infty-4}H_{\infty,k}\lambda^k+ t_{\infty,r_\infty-1}\lambda^{r_\infty-3}\cr
&&+ \left[\td{L}_{1,2}\partial_\lambda\left(\frac{ \td{L}_{1,1}}{\td{L}_{1,2}}\right) \right]_{\infty,+}+O(\lambda^{-1})
\eea
From \eqref{MontrealInfinity}, we have:
\bea \label{IdInfty2}\det\td{L}(\lambda)&\overset{\lambda\to\infty}{=}&-\sum_{k=r_\infty-3}^{2r_\infty-4}\sum_{j=k+3-r_\infty}^{r_\infty-1} t_{\infty,k+2-j}t_{\infty,j}\lambda^{k}-\sum_{k=0}^{r_\infty-4}\sum_{j=0}^{k+2} t_{\infty,k+2-j}t_{\infty,j}\lambda^{k}\cr
&&-2\sum_{k=-1}^{r_\infty-4}\sum_{j=1}^{r_\infty-3-k} j I_{\infty,j} t_{\infty,j+k+2}\lambda^{k}+O(\lambda^{-2})
\eea
For all $k \in \llbracket 0,r_\infty-4\rrbracket$, identifying \eqref{IdInfty1} and \eqref{IdInfty2} implies that:
\beq H_{\infty,k}=-\underset{\lambda\to \infty}{\Res}\lambda^{-k-1} \td{L}_{1,2}(\lambda)\partial_\lambda\left(\frac{ \td{L}_{1,1}(\lambda)}{\td{L}_{1,2}(\lambda)}\right) +\sum_{j=0}^{k+2} t_{\infty,k+2-j}t_{\infty,j}+2\sum_{j=1}^{r_\infty-3-k} j I_{\infty,j} t_{\infty,j+k+2}\eeq
These relations may be rewritten into a matrix form:
\small{\beq\label{MatrixIHInfty} 2M_{\infty}(\mathbf{t})\begin{pmatrix}I_{\infty,1}\\2I_{\infty,2}\\\vdots\\ (r_{\infty}-3)I_{\infty,r_\infty-3} \end{pmatrix}= \begin{pmatrix}H_{\infty,r_\infty-4}\\ H_{\infty,r_\infty-5}\\\vdots \\ H_{\infty,0}\end{pmatrix} -\begin{pmatrix}\underset{j=0}{\overset{r_\infty-2}{\sum}} t_{\infty,r_\infty-2-j}t_{\infty,j}-\underset{\lambda\to \infty}{\Res}\lambda^{-r_\infty+3} \td{L}_{1,2}(\lambda)\partial_\lambda\left(\frac{ \td{L}_{1,1}(\lambda)}{\td{L}_{1,2}(\lambda)}\right)\\
\underset{j=0}{\overset{r_\infty-3}{\sum}} t_{\infty,r_\infty-3-j}t_{\infty,j}-\underset{\lambda\to \infty}{\Res}\lambda^{-r_\infty+4} \td{L}_{1,2}(\lambda)\partial_\lambda\left(\frac{ \td{L}_{1,1}(\lambda)}{\td{L}_{1,2}(\lambda)}\right)\\
\vdots\\
\underset{j=0}{\overset{2}{\sum}} t_{\infty,2-j}t_{\infty,j}-\underset{\lambda\to \infty}{\Res}\lambda^{-1}\td{L}_{1,2}(\lambda)\partial_\lambda\left(\frac{ \td{L}_{1,1}(\lambda)}{\td{L}_{1,2}(\lambda)}\right)
  \end{pmatrix}
\eeq}

\normalsize{We} may combine \eqref{MatrixIHXs} and \eqref{MatrixIHInfty} with Theorem \ref{TheoHamiltonian} to obtain the following statement.

\begin{theorem}[Relation between Hamiltonians and spectral invariants]\label{TheoHamSpectral}The Hamiltonians relatively to the geometric Darboux coordinates $(\mathbf{Q},\mathbf{P})$ are related to the spectral invariants by
\bea \label{NewHamReduced3}&&\begin{pmatrix}\text{Ham}^{(\alpha_{t_{\infty,1}})}(\mathbf{Q},\mathbf{P},\mathbf{t},\mathbf{t}_0)\\ 2\text{Ham}^{(\alpha_{t_{\infty,2}})}(\mathbf{Q},\mathbf{P},\mathbf{t},\mathbf{t}_0)\\ \vdots \\ (r_\infty-3)\text{Ham}^{(\alpha_{t_{\infty,r_\infty -3}})}(\mathbf{Q},\mathbf{P},\mathbf{t},\mathbf{t}_0)\end{pmatrix}=2\begin{pmatrix}I_{\infty,1}\\2I_{\infty,2}\\\vdots\\ (r_{\infty}-3)I_{\infty,r_\infty-3} \end{pmatrix}\cr
&&+\left(M_\infty(\mathbf{t})\right)^{-1}\begin{pmatrix}\underset{j=0}{\overset{r_\infty-2}{\sum}} t_{\infty,r_\infty-2-j}t_{\infty,j}-\underset{\lambda\to \infty}{\Res}\lambda^{-r_\infty+3} \td{L}_{1,2}(\lambda)\partial_\lambda\left(\frac{ \td{L}_{1,1}(\lambda)}{\td{L}_{1,2}(\lambda)}\right)\\
\underset{j=0}{\overset{r_\infty-3}{\sum}} t_{\infty,r_\infty-3-j}t_{\infty,j}-\underset{\lambda\to \infty}{\Res}\lambda^{-r_\infty+4} \td{L}_{1,2}(\lambda)\partial_\lambda\left(\frac{ \td{L}_{1,1}(\lambda)}{\td{L}_{1,2}(\lambda)}\right)\\
\vdots\\
\underset{j=0}{\overset{2}{\sum}} t_{\infty,2-j}t_{\infty,j}-\underset{\lambda\to \infty}{\Res}\lambda^{-1}\td{L}_{1,2}(\lambda)\partial_\lambda\left(\frac{ \td{L}_{1,1}(\lambda)}{\td{L}_{1,2}(\lambda)}\right)
  \end{pmatrix}\cr
	&&
	\eea
and
	\bea
&&\begin{pmatrix}\text{Ham}^{(\alpha_{t_{X_s,1}})}(\mathbf{Q},\mathbf{P},\mathbf{t},\mathbf{t}_0)\\2\text{Ham}^{(\alpha_{t_{X_s,2}})}(\mathbf{Q},\mathbf{P},\mathbf{t},\mathbf{t}_0)\\ \vdots \\ (r_s-1)\text{Ham}^{(\alpha_{t_{X_s,r_s-1}})}(\mathbf{Q},\mathbf{P},\mathbf{t},\mathbf{t}_0)\end{pmatrix}=2\begin{pmatrix}I_{X_s,1}\\2I_{X_s,2}\\\vdots\\ (r_{s}-1)I_{X_s,r_s-1} \end{pmatrix}\cr
&&+ \left(M_{X_s}(\mathbf{t})\right)^{-1}\begin{pmatrix}\underset{j=0}{\overset{r_s-2}{\sum}} t_{X_s,j}t_{X_s,r_s-j-2}+\underset{\lambda\to X_s}{\Res}(\lambda-X_s)^{r_s-1}\td{L}_{1,2}(\lambda)\partial_\lambda\left(\frac{ \td{L}_{1,1}(\lambda)}{\td{L}_{1,2}(\lambda)}\right)\\
\underset{j=0}{\overset{r_s-3}{\sum}} t_{X_s,j}t_{X_s,r_s-j-3}+\underset{\lambda\to X_s}{\Res}(\lambda-X_s)^{r_s-2}\td{L}_{1,2}(\lambda)\partial_\lambda\left(\frac{ \td{L}_{1,1}(\lambda)}{\td{L}_{1,2}(\lambda)}\right)\\
\vdots\\
(t_{X_s,0})^2+\underset{\lambda\to X_s}{\Res}(\lambda-X_s)\td{L}_{1,2}(\lambda)\partial_\lambda\left(\frac{ \td{L}_{1,1}(\lambda)}{\td{L}_{1,2}(\lambda)}\right)
  \end{pmatrix}\cr
&&
\eea\normalsize{for} all $s\in \llbracket 1,n\rrbracket$.
\end{theorem}

Theorem \ref{TheoHamSpectral} shows that the spectral invariants are indeed related to the Hamiltonian of the system. However, there is no equality between both sets because of the non-trivial terms in the right-hand-sides. These terms are directly related to the gauge transformation $G(\lambda)$ going from the initial gauge $\td{L}$ to the oper gauge $L$ via the additional $(\partial_\lambda G)G^{-1}$ term that arises in the process. Note that this extra term does not appear when considering isospectral transformations, so it is not a surprise that this term is at the core of the difference between the isospectral world and the isomonodromic world. The main issue is thus to find a way to cancel this extra-term and it is precisely the purpose of \cite{BertolaHarnadHurtubise2022} to propose a solution by imposing the isospectral condition. As proved in \cite{BertolaHarnadHurtubise2022} and recalled in Theorem \ref{MontrealResults}, a possible solution is to impose some additional conditions (eq. \eqref{IsoCondition}) on the Lax system. In our setting, this additional condition is equivalent to perform a non-trivial (in the sense time-dependent and not symplectic) change of coordinates that modifies the Hamiltonians to cancel this extra-term. 

\subsection{Isospectral Darboux coordinates}
In this section, we look for an explicit change of coordinates $(\mathbf{Q},\mathbf{P})\rightarrow (\mathbf{u},\mathbf{v})$ so that the coordinates $(\mathbf{u},\mathbf{v})$ satisfy the isospectral condition \eqref{IsoCondition}. Consequently, for these isospectral coordinates, the spectral invariants $\mathbf{I}$ identify with the Hamiltonian of the system. Note that the choice of isospectral coordinates is not unique since any time-independent symplectic change of coordinates from any isospectral coordinates would still provide coordinates satisfying \eqref{IsoCondition}. Thus, only the explicit dependence in the irregular times is determined by condition \eqref{IsoCondition}, leaving any symplectic choice of coordinates to determine the time-independent part of the underlying symplectic form. In particular our isospectral Darboux coordinates $(\mathbf{u},\mathbf{v})$ are not canonical in general and therefore a time-independent but non-symplectic change of coordinates is required to trivialize the symplectic structure.
As explained above, the change of coordinates $(\mathbf{Q},\mathbf{P})\rightarrow (\mathbf{u},\mathbf{v})$ cannot be time-independent and symplectic since it is required to cancel the extra terms in Theorem \ref{TheoHamSpectral}. More specifically, we look for an explicit time dependence so that the new Hamiltonians shall equal the old ones plus the extra term on the r.h.s. of Theorem \ref{TheoHamSpectral}. When performing a time-dependent change of coordinates, it is easy to compute the additional terms obtained in the Hamiltonian. On the contrary, our problem goes in the opposite direction: we know the additional terms in the Hamiltonian system that we need to cancel and we want to obtain the time-dependence of the change of coordinates that would correspond to cancel them. To our knowledge, there does not exist any general method to ensure that such a change of coordinates exists nor a constructive method to obtain it in general.

\medskip

Fortunately, Theorem \ref{GeoLaxMatrices} (resp. Theorem \ref{GeoLaxMatricesQR}) provides the explicit expressions of the Lax matrices in terms of the coordinates $(\mathbf{Q},\mathbf{P})$ (resp. $(\mathbf{Q},\mathbf{R})$). Consequently a simple strategy consists in solving directly the isospectral condition \eqref{IsoCondition} using these Lax matrices and obtain the required time-dependence for the appropriate change of coordinates. 

\subsubsection{Differential system for the $\mathbf{Q}$ coordinates and definition of the $\mathbf{u}$ isospectral coordinates}

Let us start with entry $(1,2)$ of condition \eqref{IsoCondition}. From Theorem \ref{GeoLaxMatrices}, we have:
\small{\bea \label{EqEQ}&& \sum_{s=1}^n\sum_{k=2}^{r_s+1} (k-1)\frac{Q_{X_s,k-1}\delta^{(\boldsymbol{\alpha})}_{\mathbf{t}}[X_s]}{(\lambda-X_s)^{k}}+\sum_{s=1}^n\sum_{k=1}^{r_s} \frac{\delta^{(\boldsymbol{\alpha})}_{\mathbf{t}}[Q_{X_s,k}]}{(\lambda-X_s)^k}+\sum_{k=0}^{r_\infty-4} \delta^{(\boldsymbol{\alpha})}_{\mathbf{t}}[Q_{\infty,k}]\lambda^k+ \delta^{(\boldsymbol{\alpha})}_{\mathbf{t}}[\omega] \delta_{r_\infty\geq 3}\lambda^{r_\infty-3}\cr
&&=(r_\infty-4)\omega\,\nu_{\infty,1}^{(\boldsymbol{\alpha})}\lambda^{r_\infty-5}\delta_{r_\infty\geq 4}+\sum_{j=1}^{r_\infty-5}j\left(\omega\,\nu_{\infty,r_\infty-3-j}^{(\boldsymbol{\alpha})}+\sum_{k=j+1}^{r_\infty-3}\nu_{\infty,k-j}^{(\boldsymbol{\alpha})}Q_{\infty,k}\right)\lambda^{j-1}\cr
&&-\sum_{s=1}^n\sum_{k=2}^{r_s+1}(k-1)\left(\sum_{i=k-1}^{r_s} \nu_{X_s,i+1-k}^{(\boldsymbol{\alpha})}  Q_{X_s,i}\right)(\lambda-X_s)^{-k}+\text{reg}(\lambda)\cr
&&\eea}
\normalsize{where} the regular term is $O(\lambda^{-1})$ when $\lambda\to \infty$ and $O(1)$ when $\lambda\to X_s$ for any $s\in \llbracket 1,n\rrbracket$. Identifying the coefficients implies the following proposition:

\begin{proposition}\label{PropQCondition} Entry $(1,2)$ of the condition $\delta^{(\boldsymbol{\alpha})}_{\mathbf{t}}[\td{L}(\lambda)]=\partial_\lambda \td{A}_{\boldsymbol{\alpha}}(\lambda)$ is equivalent to
\bea  \delta^{(\boldsymbol{\alpha})}_{\mathbf{t}}[\omega]&=&-\omega\, \nu_{\infty,-1}^{(\boldsymbol{\alpha})} \delta_{r_\infty=1} \cr 
\delta^{(\boldsymbol{\alpha})}_{\mathbf{t}}[X_s]&=&\alpha_{X_s,0} \,\,,\,\, \forall\, s\in \llbracket 1,n\rrbracket\cr
\delta^{(\boldsymbol{\alpha})}_{\mathbf{t}}[Q_{X_s,1}]&=&0\,\,,\,\, \forall\, s\in \llbracket 1,n\rrbracket\cr
\delta^{(\boldsymbol{\alpha})}_{\mathbf{t}}[Q_{X_s,k}]&=&- (k-1)\sum_{i=k}^{r_s} \nu_{X_s,i+1-k}^{(\boldsymbol{\alpha})}  Q_{X_s,i}\,\,,\,\,\forall (s,k)\in \llbracket 1,n\rrbracket\times\llbracket 2,r_s\rrbracket\cr
\delta^{(\boldsymbol{\alpha})}_{\mathbf{t}}[Q_{\infty,r_\infty-4}]&=&0\cr
\delta^{(\boldsymbol{\alpha})}_{\mathbf{t}}[Q_{\infty,r_\infty-5}]&=&(r_\infty-4)\omega\,\nu_{\infty,1}^{(\boldsymbol{\alpha})}\cr
\delta^{(\boldsymbol{\alpha})}_{\mathbf{t}}[Q_{\infty,m}]&=&(m+1)\left(\omega\,\nu_{\infty,r_\infty-4-m}^{(\boldsymbol{\alpha})}+\sum_{k=m+2}^{r_\infty-4}\nu_{\infty,k-m-1}^{(\boldsymbol{\alpha})}Q_{\infty,k}\right)\,\,,\,\, \forall\, m\in \llbracket 0, r_\infty-6\rrbracket\cr
&&
\eea
\end{proposition}

\begin{proof}
Looking at order $(\lambda-X_s)^{-r_s-1}$ provides:
\beq \label{deltatXs}\delta^{(\boldsymbol{\alpha})}_{\mathbf{t}}[X_s]=- \nu_{X_s,0}^{(\boldsymbol{\alpha})}=\alpha_{X_s,0} \,\, \forall\, s\in \llbracket 1,n\rrbracket\eeq
Order $(\lambda-X_s)^{-k}$ with $k\in \llbracket 2,r_s\rrbracket$ provides:
\beq (k-1)Q_{X_s,k-1}\delta^{(\boldsymbol{\alpha})}_{\mathbf{t}}[X_s]+\delta^{(\boldsymbol{\alpha})}_{\mathbf{t}}[Q_{X_s,k}]=-(k-1)\left(\nu_{X_s,0}^{(\boldsymbol{\alpha})}Q_{X_s,k-1} +\sum_{i=k}^{r_s} \nu_{X_s,i+1-k}^{(\boldsymbol{\alpha})}  Q_{X_s,i}\right)\eeq
i.e. using \eqref{deltatXs}
\beq \delta^{(\boldsymbol{\alpha})}_{\mathbf{t}}[Q_{X_s,k}]=- (k-1)\sum_{i=k}^{r_s} \nu_{X_s,i+1-k}^{(\boldsymbol{\alpha})}  Q_{X_s,i}\,\,,\,\,\forall (s,k)\in \llbracket 1,n\rrbracket\times\llbracket 2,r_s\rrbracket\eeq
Order $(\lambda-X_s)^{-1}$ simply provides $\delta_{\mathbf{t}}[Q_{X_s,1}]=0$ for all $s\in \llbracket 1,n\rrbracket$. 
Let us look at the situation at infinity depending on the value of $r_\infty$.
\begin{itemize}\item If $r_\infty\geq 3$, order $\lambda^{r_\infty-3}$ provides $\delta^{(\boldsymbol{\alpha})}_{\mathbf{t}}[\omega]=0$. Then if $r_\infty\geq 4$, order $\lambda^{r_\infty-4}$ provides $\delta^{(\boldsymbol{\alpha})}_{\mathbf{t}}[Q_{\infty,r_\infty-4}]=0$. For $r_\infty\geq 5$, order $\lambda^{r_\infty-5}$ provides $\delta^{(\boldsymbol{\alpha})}_{\mathbf{t}}[Q_{\infty,r_\infty-5}]=(r_\infty-4)\nu_{\infty,1}^{(\boldsymbol{\alpha})}$ and for all $m\in \llbracket 0, r_\infty-6\rrbracket$ order $\lambda^m$ provides
\beq \delta^{(\boldsymbol{\alpha})}_{\mathbf{t}}[Q_{\infty,m}]=(m+1)\left(\nu_{\infty,r_\infty-4-m}^{(\boldsymbol{\alpha})}+\sum_{k=m+2}^{r_\infty-4}\nu_{\infty,k-m-1}^{(\boldsymbol{\alpha})}Q_{\infty,k}\right)\eeq
\item For $r_\infty=2$, order $\lambda^{-1}$ of \eqref{EqEQ} is equivalent to
\beq \sum_{s=1}^n \delta^{(\boldsymbol{\alpha})}_{\mathbf{t}}[Q_{X_s,1}]=0\eeq
But from \eqref{AddConstrains}, this is equivalent to $\delta^{(\boldsymbol{\alpha})}_{\mathbf{t}}[\omega]=0$.
\item For $r_\infty=1$, order $\lambda^{-1}$ of \eqref{EqEQ} is equivalent to $\underset{s=1}{\overset{n}{\sum}} \delta^{(\boldsymbol{\alpha})}_{\mathbf{t}}[Q_{X_s,1}]=0$ which is consistent with \eqref{AddConstrains2}. Order $\lambda^{-2}$ of \eqref{EqEQ} is equivalent to
\beq \sum_{s=1}^n Q_{X_s,1}\delta^{(\boldsymbol{\alpha})}_{\mathbf{t}}[X_s]+\sum_{s=1}^n\sum_{k=1}^{r_s} X_s\delta^{(\boldsymbol{\alpha})}_{\mathbf{t}}[Q_{X_s,1}]+\delta^{(\boldsymbol{\alpha})}_{\mathbf{t}}[Q_{X_s,2}]=-\sum_{s=1}^n\sum_{i=1}^{r_s} \nu_{X_s,i-1}^{(\boldsymbol{\alpha})}  Q_{X_s,i}\eeq
From \eqref{AddConstrains2}, this is equivalent to
\beq \delta^{(\boldsymbol{\alpha})}_{\mathbf{t}}[\omega]=-\sum_{s=1}^n\sum_{i=1}^{r_s} \nu_{X_s,i-1}^{(\boldsymbol{\alpha})}  Q_{X_s,i}\overset{\eqref{ExtraConditionsrinftyequal1}}{=}-\omega\, \nu_{\infty,-1}^{(\boldsymbol{\alpha})}\eeq
\end{itemize}
\end{proof}

A solution to Proposition \ref{PropQCondition} is given by the following theorem.

\begin{theorem}\label{ExplicitDependenceQ}The conditions $\delta^{(\boldsymbol{\alpha})}_{\mathbf{t}}[\td{L}_{1,2}(\lambda)]=\partial_\lambda \td{A}_{1,2}(\lambda)$ are equivalent, for all $s\in \llbracket 1,n\rrbracket$, to $Q_{X_s,1}=u_{X_s,1}$ independent of the irregular times and to solve the $(r_s-1)\times(r_s-1)$ differential system
\footnotesize{\bea\label{SystemXss} &&\begin{pmatrix}t_{X_s,r_s-1}&0&\dots&0\\
t_{X_s,r_s-2}&t_{X_s,r_s-1}& &0\\
\vdots& \ddots &\ddots &\vdots \\
t_{X_s,1}&\dots&t_{X_s,r_s-2}&t_{X_s,r_s-1}
\end{pmatrix}
\begin{pmatrix}\frac{1}{r_s-1}&0&\dots&0\\
0&\ddots& &0\\
\vdots& & \ddots&0\\
0&\dots&0 &\frac{1}{1} \end{pmatrix} \begin{pmatrix} \delta_{t_{X_s,r_s-1}}[Q_{X_s,r_s}]&\dots& \delta_{t_{X_s,1}}[Q_{X_s,r_s}]\\
\vdots&\vdots&\vdots\\
\delta_{t_{X_s,r_s-1}}[Q_{X_s,2}]&\dots& \delta_{t_{X_s,1}}[Q_{X_s,2}]
\end{pmatrix}\cr
&&=\begin{pmatrix}Q_{X_s,r_s}&0&\dots& &\dots &0\\
Q_{X_s,r_s-1}&Q_{X_s,r_s}& 0& & &\vdots\\
\vdots & \ddots&\ddots &\ddots  & &\vdots\\
\vdots &\ddots&\ddots&\ddots&0&\vdots\\
Q_{X_s,3}&\ddots &\ddots&\ddots& Q_{X_s,r_s}&0\\
Q_{X_s,2}&Q_{X_s,3}& \dots & & Q_{X_s,r_s-1}& Q_{X_s,r_s}
 \end{pmatrix}\begin{pmatrix}\frac{1}{r_s-1}&0&\dots&0\\
0&\ddots& &0\\
\vdots& & \ddots&0\\
0&\dots&0 &\frac{1}{1} \end{pmatrix}
\cr
&&
\eea}
\normalsize{whose} solutions are of the form:
\footnotesize{\beq \begin{pmatrix}Q_{X_s,r_s}\\ Q_{X_s,r_s-1}\\ \vdots\\\vdots\\ Q_{X_s,2}\end{pmatrix}=
\begin{pmatrix} f^{(X_s)}_{1,1}(t_{X_s,r_s-1})&0&\dots&\dots& 0\\
 f^{(X_s)}_{2,1}(t_{X_s,r_s-2},t_{X_s,r_s-1})& f^{(X_s)}_{2,2}(t_{X_s,r_s-1})&0&&\vdots\\
\vdots&& \ddots& \ddots &\vdots\\
f^{(X_s)}_{r_s-2,1}(t_{X_s,2},\dots,t_{X_s,r_s-1})&\dots&\dots& f^{(X_s)}_{r_s-2,r_s-1}(t_{X_s,r_s-1})&0\\
f^{(X_s)}_{r_s-1,1}(t_{X_s,1},\dots,t_{X_s,r_s-1})&\dots&&\dots& f^{(X_s)}_{r_s-1,r_s-1}(t_{X_s,r_s-1})
\end{pmatrix}
\begin{pmatrix}u_{X_s,r_s}\\ u_{X_s,r_s-1}\\\vdots\\ \vdots\\ u_{X_s,2}
\end{pmatrix} 
\eeq}
\normalsize{with} all $\left(u_{X_s,k}\right)_{1\leq k\leq r_s}$ independent of the irregular times. Moreover, we have:
\bea f^{(X_s)}_{j,j}(t_{X_s,r_s-1})&=&(t_{X_s,r_s-1})^{\frac{r_s-j}{r_s-1}}\,\,,\,\, \forall\, j\in \llbracket 1,r_s-1\rrbracket\cr
f^{(X_s)}_{j+1,j}(t_{X_s,r_s-2},t_{X_s,r_s-1})&=&\frac{r_s-j-1}{r_s-2}(t_{X_s,r_s-1})^{\frac{1-j}{r_s-1}}t_{X_s,r_s-2}\,\,,\,\, \forall\, j\in \llbracket 1,r_s-2\rrbracket\cr
f^{(X_s)}_{j,1}(t_{X_s,r_s-j},\dots,t_{X_s,r_s-1})&=&t_{X_s,r_s-j}\,\,,\,\, \forall\, j\in \llbracket 1,r_s-1\rrbracket
\eea

At infinity we have for $r_\infty\geq 2$ that $\omega$ is independent of the irregular times and for $r_\infty\geq 4$, $Q_{\infty,r_\infty-4}=\omega \,u_{\infty,r_\infty-4}$ independent of the irregular times. Moreover, for $r_\infty\geq 5$ we have to solve the $(r_\infty-4)\times(r_\infty-4)$ differential system
\footnotesize{\bea&&\begin{pmatrix}1&0&\dots&&0\\
0&1& &&0\\
t_{\infty,r_\infty-3}&0&\ddots &\ddots &\vdots\\
\vdots& \ddots &\ddots &\ddots&\vdots \\
t_{\infty,4}&\dots&t_{\infty,r_\infty-3}&0&1
\end{pmatrix}\begin{pmatrix}\frac{1}{r_\infty-4}&0&\dots&0\\
0&\frac{1}{r_\infty-5}& &0\\
\vdots& \ddots &\ddots &\vdots \\
0&\dots&\dots &\frac{1}{1}
\end{pmatrix} \begin{pmatrix} \delta_{t_{\infty,r_\infty-3}}[Q_{\infty,r_\infty-5}]& \dots& \delta_{t_{\infty,2}}[Q_{\infty,r_\infty-5}]\\
\vdots\\
\delta_{t_{\infty,r_\infty-3}}[Q_{\infty,0}]& \dots &\delta_{t_{\infty,2}}[Q_{\infty,0}]\end{pmatrix}\cr 
&&=\begin{pmatrix}\omega&0&\dots&0\\
Q_{\infty,r_\infty-4}&\omega& &0\\
\vdots& \ddots &\ddots &\vdots \\
Q_{\infty,2}&\dots&Q_{\infty,r_\infty-4}&\omega
\end{pmatrix}\begin{pmatrix}\frac{1}{r_\infty-3}&0&\dots&0\\
0&\frac{1}{r_\infty-4}& &0\\
\vdots& \ddots &\ddots &\vdots \\
0&\dots&\dots &\frac{1}{2}
\end{pmatrix}\cr
&&\eea
}
\normalsize{whose} solutions are of the form
\footnotesize{\beq \begin{pmatrix}\omega\\ Q_{\infty,r_\infty-4}\\ Q_{\infty,r_\infty-5}\\\vdots\\ \vdots\\ Q_{\infty,0}\end{pmatrix}=\omega\,
\begin{pmatrix}1&0&0&\dots&\dots&0\\
0&1&0&\ddots&&\vdots\\
f^{(\infty)}_{3,1}(t_{\infty,r_\infty-3})&0&1&0&&0\\
f^{(\infty)}_{4,1}(t_{\infty,r_\infty-4},t_{\infty,r_\infty-3})&f^{(\infty)}_{4,2}(t_{\infty,r_\infty-3})&0&1&\ddots&\vdots\\
\vdots& \ddots&&&\ddots&0\\
f^{(\infty)}_{r_\infty-2,1}(,t_{\infty,2},\dots,t_{\infty,r_\infty-3})&\dots&&f^{(\infty)}_{r_\infty-2,r_\infty-4}(t_{\infty,r_\infty-3})&0&1
\end{pmatrix}
\begin{pmatrix}1\\ u_{\infty,r_\infty-4}\\ u_{\infty,r_\infty-5}\\ \vdots\\ \vdots\\ u_{\infty,0}\end{pmatrix}
\eeq}
\normalsize{with} all $\left(u_{\infty,k}\right)_{0\leq k\leq r_\infty-4}$ independent of the irregular times. Moreover, we have
\bea f^{(\infty)}_{k+2,k}(t_{\infty,r_\infty-1})&=&\frac{r_\infty-3-k}{r_\infty-3} t_{\infty,r_\infty-3}  \,,\,\, \forall \, k\in \llbracket 1, r_\infty-4\rrbracket\cr
f^{(\infty)}_{k+3,k}(t_{\infty,r_\infty-2},t_{\infty,r_\infty-1})&=&\frac{r_\infty-4-k}{r_\infty-4} t_{\infty,r_\infty-4}  \,,\,\, \forall \, k\in \llbracket 1, r_\infty-5\rrbracket
\eea

Finally, for $r_\infty=2$, we have the extra condition:
\beq \omega=\sum_{s=1}^n u_{X_s,1}\eeq
while for $r_\infty=1$, we have the extra conditions:
\beq  \label{OmegaDeform}\sum_{s=1}^n u_{X_s,1}=0\,\,,\,\,\, \omega=\sum_{s=1}^n X_s u_{X_s,1}+ \sum_{s=1}^n\sum_{j=1}^{r_s-1} f^{(X_s)}_{r_s-1,j}(t_{X_s,j},\dots,t_{X_s,r_s-1}) u_{X_s,r_s+1-j}
\eeq
\end{theorem}

\begin{proof}The proof is done in Appendix \ref{AppendixExplicitDependenceQ}.
\end{proof}

\begin{remark}\label{ImportantRemarkOmega} \normalsize{For} $r_\infty\geq 2$, the isospectral condition implies that $\delta_{\mathbf{t}}^{(\boldsymbol{\alpha})}[\omega]=0$ so that $\omega$ can be set to a given non-zero constant (i.e. independent of $\lambda$, the irregular times and the positions of the finite poles). Up to a trivial rescaling, this is equivalent to set $\omega=1$ as done in \cite{MarchalOrantinAlameddine2022} or in most of the literature. On the contrary for $r_\infty=1$, one observes from Proposition \ref{PropQCondition} and its solution \eqref{OmegaDeform} that the normalization coefficient $\omega$ must depend on the irregular times at infinity in a non-trivial way to ensure the isospectral condition. This difference originates from the fact that for $r_\infty=1$, the normalization of $\td{L}_{1,2}(\lambda)$ is done on a regular order at $\lambda\to \infty$ (i.e. $\lambda^{-2}$) that is therefore connected to all the integrable structure at all poles. Hence, imposing the isospectral condition requires that this normalization adapts with the whole integrable structure.
\end{remark}

\begin{remark} \normalsize{We} \sloppy{could not obtain a closed form for the expression of the coefficients $\left(f^{(\infty)}_{i,j}\right)_{1\leq i,j\leq r_\infty-2}$ nor $\left(f^{(X_s)}_{i,j}\right)_{1\leq i,j\leq r_s-1}$. However, we stress that these differential systems are universal and that the corresponding coefficients $\left(f^{(\infty)}_{i,j}\right)_{1\leq i,j\leq r_\infty-2}$ and $\left(f^{(X_s)}_{i,j}\right)_{1\leq i,j\leq r_s-1}$ only depend on the irregular times in a way that may have some interest outside of the scope of the present problem. The exact formulas proposed in Theorem \ref{ExplicitDependenceQ} provide the explicit expressions up to $r_s=4$ and $r_\infty=6$. Solving each differential system is not obvious and coefficients rapidly get more involved because each line is determined recursively by the formers. The main difficulty is that the r.h.s. of the differential system \eqref{SystemXss} is (after conjugating by the constant matrix) a Toeplitz lower-triangular matrix, while the l.h.s. of the differential system is not a priori. In the end, solving the differential system amounts to choosing a dependence in the irregular times so that the l.h.s. becomes specifically Toeplitz lower-triangular matrix (up to the conjugation by the constant matrix). It turns out that imposing the Toeplitz like condition is non trivial for a linear differential system.}
\end{remark}

\begin{remark} \normalsize{The} system looks different at infinity because of our choice of normalization at this point. Indeed, imposing the normalization ($\td{L}_{1,2}(\lambda)=\omega\,\lambda^{r_\infty-3}+O(\lambda^{r_\infty-4})$) is equivalent to freeze a possible coordinate $Q_{\infty,r_\infty-3}$ to $\omega$. Moreover, the fact that $t_{\infty,r_\infty-1}=1$ and $t_{\infty,r_\infty-2}=0$ is also easily seen in the Toeplitz matrix. Consequently, one could make other choices and immediately adapt the results of Theorem \ref{ExplicitDependenceQ}.
\end{remark}

Apart from double poles, one can see that the explicit dependence of the Darboux coordinates $\mathbf{Q}(\mathbf{t},\mathbf{t_0})$ that is required to satisfy the isospectral condition \eqref{IsoCondition} gets substantially involved since solving the differential systems becomes complicated (but can be easily done with an appropriate software). This is not a surprise since as soon as one gets the correct explicit dependence in the irregular times, results of \cite{BertolaHarnadHurtubise2022} prove that the Hamiltonians $\text{Ham}(\mathbf{u},\mathbf{v},\mathbf{t},\mathbf{t_0})$ identify with the spectral invariants $\mathbf{I}$ that can be computed from Theorem \ref{MontrealResults} and the explicit expression of $\td{L}(\lambda)$ given by Theorem \ref{GeoLaxMatrices}. Thus, it seems legitimate that getting the explicit dependence in the irregular times to satisfy the isospetral condition should be as hard as solving the compatibility equations of the Lax system.

\subsubsection{Differential system for the $\mathbf{R}$ coordinates and definition of the $\mathbf{v}$ isospectral coordinates}
Adapting the computations done for Theorem \ref{ExplicitDependenceQ} and Appendix \ref{AppendixExplicitDependenceQ}, we obtain the differential system satisfied by the coordinates $\mathbf{R}$ that have been defined in Definition \ref{DefR}.

\begin{theorem}\label{TheoExplicitDependenceP} The isospectral condition $\delta_{\mathbf{t}}^{(\boldsymbol{\alpha})} \td{L}_{1,1}(\lambda)=\partial_\lambda [\td{A}_{\boldsymbol{\alpha}}(\lambda)]_{1,1}$ is equivalent to the fact that for any $s\in \llbracket 1,n\rrbracket$, the coordinates $\left(R_{X_s,k}\right)_{1\leq k\leq r_s}$ satisfy the same differential system as \eqref{SystemXss}. Thus, for any $s\in \llbracket 1,n\rrbracket$ we have $R_{X_s,1}=v_{X_s,1}$ and
\footnotesize{\beq \begin{pmatrix}R_{X_s,r_s}\\ R_{X_s,r_s-1}\\ \vdots\\\vdots\\ R_{X_s,2}\end{pmatrix}=
\begin{pmatrix} f^{(X_s)}_{1,1}(t_{X_s,r_s-1})&0&\dots&\dots& 0\\
 f^{(X_s)}_{2,1}(t_{X_s,r_s-2},t_{X_s,r_s-1})& f^{(X_s)}_{2,2}(t_{X_s,r_s-1})&0&&\vdots\\
\vdots&& \ddots& \ddots &\vdots\\
f^{(X_s)}_{r_s-2,1}(t_{X_s,2},\dots,t_{X_s,r_s-1})&\dots&\dots& f^{(X_s)}_{r_s-2,r_s-1}(t_{X_s,r_s-1})&0\\
f^{(X_s)}_{r_s-1,1}(t_{X_s,1},\dots,t_{X_s,r_s-1})&\dots&&\dots& f^{(X_s)}_{r_s-1,r_s-1}(t_{X_s,r_s-1})
\end{pmatrix}
\begin{pmatrix}v_{X_s,r_s}\\ v_{X_s,r_s-1}\\\vdots\\ \vdots\\ v_{X_s,2}
\end{pmatrix} 
\eeq}
\normalsize{with} all $\left(v_{X_s,k}\right)_{1\leq k\leq r_s}$ independent of the irregular times. Moreover, for $r_\infty=2$ we have the additional relation
\beq \sum_{s=1}^n v_{X_s,1} =-t_{\infty,0}\eeq 
and for $r_\infty=1$ we have the additional relations
\small{\bea \label{AddRRelations} \sum_{s=1}^n v_{X_s,1}&=&-t_{\infty,0}\cr
\sum_{s=1}^n X_s v_{X_s,1}+ \sum_{s=1}^n\sum_{j=1}^{r_s-1} f^{(X_s)}_{r_s-1,j}(t_{X_s,j},\dots,t_{X_s,r_s-1}) v_{X_s,r_s+1-j}&=&-g_0-\frac{t_{\infty,0}}{\omega}\left(\sum_{s=1}^nX_s^2Q_{X_s,1}+2X_sQ_{X_s,2}+Q_{X_s,3}\right)\cr
&&
\eea}

\normalsize{Similarly}, for $r_\infty\geq 4$, the coordinates $\left(R_{\infty,k}\right)_{0\leq k\leq r_\infty-4}$ satisfy the $(r_\infty-3)\times(r_\infty-3)$ differential system 
\footnotesize{\bea &&\begin{pmatrix}1&0&\dots&\dots&0\\
0&1& &&\vdots\\
t_{\infty,r_\infty-3}&0&1& &\vdots\\
\vdots& \ddots &\ddots &\ddots&0 \\
t_{\infty,3}&\dots&t_{\infty,r_\infty-3}&0&1 
\end{pmatrix}
\begin{pmatrix}\frac{1}{r_\infty-3}&0&\dots&0\\
0&\ddots& &0\\
\vdots& & \ddots&0\\
0&\dots&0 &\frac{1}{1} \end{pmatrix} \begin{pmatrix} \delta_{t_{\infty,r_\infty-3}}[R_{\infty,r_\infty-4}]&\dots& \delta_{t_{\infty,1}}[R_{\infty,r_\infty-4}]\\
\vdots&\vdots&\vdots\\
\delta_{t_{\infty,r_\infty-3}}[R_{\infty,0}]&\dots& \delta_{t_{\infty,1}}[R_{\infty,0}]
\end{pmatrix}\cr
&&=\begin{pmatrix}-1&0&\dots&\dots& \dots&\dots &0\\
0&-1& 0&& & &\vdots\\
R_{\infty,r_\infty-4}&0& -1& && &\vdots\\
\vdots & \ddots&\ddots &\ddots&\ddots  & &\vdots\\
\vdots &\ddots&\ddots&\ddots&-1&0&\vdots\\
R_{\infty,r_\infty,3}&\ddots &\ddots&\ddots&\ddots& -1&0\\
R_{\infty,r_\infty,2}&R_{\infty,r_\infty,3}& \dots &\dots&R_{\infty,r_\infty-4}& 0& -1
 \end{pmatrix}\begin{pmatrix}\frac{1}{r_\infty-3}&0&\dots&0\\
0&\ddots& &0\\
\vdots& & \ddots&0\\
0&\dots&0 &\frac{1}{1} \end{pmatrix}
\cr
&&
\eea}
\normalsize{whose} solutions are 
\footnotesize{\bea \begin{pmatrix}R_{\infty,r_\infty-4}\\ R_{\infty,r_\infty-5}\\ \vdots\\\vdots\\ R_{\infty,0}\end{pmatrix}&=&-\begin{pmatrix}t_{\infty,r_\infty-3}\\ t_{\infty,r_\infty-4}\\ \vdots\\ t_{\infty,1}\end{pmatrix}\cr
&&+
\begin{pmatrix}1&0&0&\dots&\dots&0\\
0&1&0&\ddots&&\vdots\\
g^{(\infty)}_{3,1}(t_{\infty,r_\infty-3})&0&1&0&&0\\
g^{(\infty)}_{4,1}(t_{\infty,r_\infty-4},t_{\infty,r_\infty-3})&g^{(\infty)}_{4,2}(t_{\infty,r_\infty-3})&0&1&\ddots&\vdots\\
\vdots& \ddots&&&\ddots&0\\
g^{(\infty)}_{r_\infty-3,1}(t_{\infty,3},\dots,t_{\infty,r_\infty-3})&\dots&&g^{(\infty)}_{r_\infty-3,r_\infty-4}(t_{\infty,r_\infty-3})&0&1
\end{pmatrix}
\begin{pmatrix}v_{\infty,r_\infty-4}\\ v_{\infty,r_\infty-5}\\\vdots\\ \vdots\\ v_{\infty,0}
\end{pmatrix} \cr
&&
\eea}
\normalsize{with} all $\left(v_{\infty,k}\right)_{0\leq k\leq r_\infty-4}$ independent of the irregular times. Moreover, we have:
\bea g^{(\infty)}_{j+2,j}(t_{\infty,r_\infty-3})&=&\frac{r_\infty-4-j}{r_\infty-3}t_{\infty,r_\infty-3}\,\,,\,\, \forall\, j\in \llbracket 1,r_\infty-5\rrbracket\cr
  g^{(\infty)}_{j+3,j}(t_{\infty,r_\infty-4},t_{\infty,r_\infty-3})&=&\frac{r_\infty-5-j}{r_\infty-4}t_{\infty,r_\infty-4}\,\,,\,\,  \forall\, j\in \llbracket 1,r_\infty-6\rrbracket
\eea
\end{theorem}

\begin{proof}The proof is done in Appendix \ref{AppendixExplicitDependenceR}.
\end{proof}

Finally, one may obtain the explicit dependence of the coordinates $\mathbf{P}$ in terms of the irregular times so that the isospectral condition \eqref{IsoCondition} is realized.

\begin{theorem}\label{MainTheoIsospectral}The isospectral condition \eqref{IsoCondition} is realized if and only if the coordinates $\mathbf{Q}$ satisfy Theorem \ref{ExplicitDependenceQ} and if the coordinates $\mathbf{P}$ are given by \eqref{InversChangeCoo} and \eqref{InversChangeCoo2} where the coordinates $\mathbf{R}$ satisfy Theorem \ref{TheoExplicitDependenceP}.
\end{theorem}

\begin{proof}It is obvious that the isospectral condition \eqref{IsoCondition} implies the specific time dependence of the coordinates $(\mathbf{Q},\mathbf{R})$ providing a necessary condition. Moreover, \cite{BertolaHarnadHurtubise2022} proves that a set of coordinates satisfying the isospectral condition always exists providing the sufficient part of the theorem. Another way to obtain the sufficient part of the theorem is to study entry $(2,1)$ of the isospectral condition.  However since this entry is directly expressed from entries $(1,1)$ and $(1,2)$ and is computed specifically to match the isomonodromic compatibility equations, it shall not generate new relations on coordinates $(\mathbf{Q},\mathbf{R})$ but rather only corresponds to consistency relations.
\end{proof}

\section{Conclusion and outlooks}
In this paper we provided several sets of Darboux coordinates to describe the Lax pairs associated to isomonodromic deformations of non-twisted meromorphic connections in $\mathfrak{sl}_{2}(\mathbb{C})$. First, we recalled in Section \ref{SectionOldResults} the results of \cite{MarchalOrantinAlameddine2022} and the description of the Lax pair $(L,A_{\boldsymbol{\alpha}})$ in the oper gauge using Darboux coordinates $(\mathbf{q},\mathbf{p})$ in relation with the apparent singularities of the quantum curve. This gauge and these coordinates are universal and provide a natural interpretation of the Hamiltonians in terms of interacting particles. However, from the geometric point of view, the fact that the Lax pair and Hamiltonian have singularities at $q_i\neq q_j$ when $i\neq j$ is not practical.

We then introduced geometric Darboux coordinates $(\mathbf{Q},\mathbf{P})$ that are related to $(\mathbf{q},\mathbf{p})$ through a time-independent and symplectic change of coordinates. Using these geometric Darboux coordinates, we provided the explicit expressions of the geometric Lax matrices $(\td{L},\td{A}_{\boldsymbol{\alpha}})$ in Theorem \ref{GeoLaxMatrices} and of the associated Hamiltonians in Theorem \ref{TheoHamNew}. The main interest of these coordinates is to provide expressions directly in the geometric gauge rather than the oper gauge so that the connections with the geometry and previous results in the literature (for example Painlev\'{e} cases) are simpler. Moreover, the fact that they are related to $(\mathbf{q},\mathbf{p})$ using a time-independent and symplectic change of coordinates makes the connection with the oper gauge and the corresponding symplectic structure immediate. On the negative side, this makes the formula for $\td{L}_{1,1}$ rather complicated. 

In Section \ref{SectionGeoLaxCoord}, we defined some geometric Lax coordinates $(\mathbf{Q},\mathbf{R})$ for which the expression of the geometric Lax matrices $(\td{L},\td{A}_{\boldsymbol{\alpha}})$ is very simple (Theorem \ref{GeoLaxMatricesQR}). On the contrary, since the relation with geometric Darboux coordinates $(\mathbf{Q},\mathbf{P})$ is not symplectic, the corresponding Hamiltonians are more difficult to express. 

Finally, we defined from the geometric Lax coordinates $(\mathbf{Q},\mathbf{R})$ (Theorems \ref{ExplicitDependenceQ} and \ref{TheoExplicitDependenceP}) some isospectral Darboux coordinates $(\mathbf{u},\mathbf{v})$ that are solutions of the isospectral condition \eqref{IsoCondition} of \cite{BertolaHarnadHurtubise2022}. In particular, the Hamiltonians corresponding to these coordinates are simply given by the expansion of $\det \td{L}(\lambda)$ at each pole and thus are straightforward to obtain. These isospectral coordinates also provide the connection with isospectral deformations that have been extensively studied in the literature. Our main observation in that matter is that for Fuchsian singularities the four sets of coordinates are related in a trivial way. For irregular singularities of low orders ($r_s\leq 2$ or $r_\infty\leq 6$), isospectral Darboux coordinates identify with geometric Darboux coordinates $(\mathbf{Q},\mathbf{P})$ when choosing the appropriate normalization of the Lax pairs. Nevertheless, this is no longer the case for irregular singularities of higher orders where the relations between both sets become non-trivial. This observation may explain why the isospectral deformations strategy has been successful in some numerous case-by-case analysis but had not been generalized so far. In $\mathfrak{sl}_2(\mathbb{C})$, our results provide an answer to the question left opened in \cite{BertolaHarnadHurtubise2022} to find isospectral coordinates.

\medskip

Geometrically, the correspondence between the isospectral Darboux coordinates $(\mathbf{u},\mathbf{v})$ and the Darboux coordinates $(\mathbf{q},\mathbf{p})$ based on the apparent singularities is highly time-dependent for irregular singularities implying that both symplectic Ehresmann connections differ. More precisely, in addition to the isomonodromy connection, whose horizontal sections are given by Stokes data= \textit{constant}, the isospectral approach of Yamakawa \cite{Yamakawa2017TauFA} and the Montr\'{e}al school \cite{BertolaHarnadHurtubise2022}  uses the so-called Birkhoff connection, whose horizontal sections correspond to $(\mathbf{u},\mathbf{v})=\textit{constant}$ in this article. It is well-known that the difference between the isomonodromy connection and the Birkhoff connection is the Jimbo-Miwa-Ueno one-form (denoted $\omega_{IM}$ in \cite{BertolaHarnadHurtubise2022}) that has been extensively studied. On the contrary in this paper, we proposed another Ehresmann connection whose horizontal sections correspond to $(\mathbf{q},\mathbf{p})=\textit{constant}$ (or equivalently $(\mathbf{Q},\mathbf{P})=\textit{constant}$ because the change is symplectic and time-independent) where $(\mathbf{q},\mathbf{p})$ are the Darboux coordinates associated to the apparent singularities. The difference between this Ehresmann connection and the isomonodromy connection is $\omega=\sum_{p,k} \text{Ham}^{(\alpha_{p,k})} d t_{p,k}$ given by either Theorem \ref{TheoHamiltonian} (See. \cite{MarchalOrantinAlameddine2022}) or Theorem \ref{TheoHamNew}. Although the map between the two setups is explicit (Theorems \ref{ExplicitDependenceQ} and \ref{TheoExplicitDependenceP}, it would be interesting to study the global properties of this new connection. In particular, the flatness and completeness (proved for the Birkhoff connection in \cite{Yamakawa2017TauFA}) would deserve attention.

\medskip

There are also several other extensions of our work that would deserve investigations.

\begin{itemize}
\item The first natural extension consists in obtaining similar results in $\mathfrak{sl}_d(\mathbb{C})$ with $d\geq 3$ for unramified poles. Indeed, the isospectral strategy developed in \cite{BertolaHarnadHurtubise2022} is valid for any $d\geq 2$. It seems to the authors that using properties of the gauge transformation from the geometric gauge to the oper gauge in a smart way should be sufficient to obtain the expression of the Lax pairs and of the Hamiltonian system using similar computations as the one done in \cite{MarchalOrantinAlameddine2022}. We let this very interesting question for future works.
\item The differential systems solving the isospectral condition \eqref{IsoCondition} (Theorems \ref{ExplicitDependenceQ} and \ref{TheoExplicitDependenceP})  amount to transform a matrix of derivative into a lower-triangular Toeplitz invertible matrix. It is worth noticing that $m\times m$ lower-triangular Toeplitz matrices are a commuting subalgebra of dimension $m$ of $\text{GL}_m(\mathbb{C})$ (that can be extended to a subfield by considering lower-triangular Toeplitz matrices with non-zero diagonal entries). From this point of view, the isospectral condition \eqref{IsoCondition} seems equivalent to transform a problem from the algebra of matrices to a commuting subalgebra (of lower-triangular Toeplitz matrices). This observation would definitely deserve additional investigations from a purely algebraic point of view and possible interpretation from a geometric point of view.
\item Another natural extension of this work is to try to coalesce the leading order at a singular pole to zero eigenvalues, for example to take $t_{\infty,r_\infty-1}\to 0$ to obtain a resonant case. However, the question of coalescing the leading order eigenvalues is a subtle one. In our case, one could try for example to take $t_{\infty,r_\infty-1}=0$. In this case, the matrix $M_\infty$ is no longer invertible and the Hamiltonians in Theorem \ref{TheoHamiltonian} diverge (not that a priori the spectral invariants remain well-defined). Many other quantities throughout the paper would also diverge. Nevertheless, such limits would naturally correspond to a decrease in the order of the singularity at one pole (in the example just passing from $r_\infty$ to $r_\infty-1$) and truncate the lower-triangular matrices accordingly. But note that the limit $t_{\infty,r_\infty-1}\to 0$ is very different from the switch between the untwisted case and the twisted case. 

\item In \cite{MarchalAlameddineP1Hierarchy2023}, the authors derived the expression of the Lax matrices $(\td{L},\td{A}_{\boldsymbol{\alpha}})$ for the case of twisted meromorphic connections with an irregular pole at infinity with a nilpotent leading order (i.e. an element of $F_{\{\infty\},r_\infty}$ but not in $\hat{F}_{\{\infty\},r_\infty}$). Note that the limit $t_{\infty,r_\infty-1}\to 0$ (or $t_{X_s,r_s-1}\to 0$) does not provide the twisted case but rather simply decreases the order of singularity at the pole, i.e. replacing  $r_\infty\to r_\infty-1$ (or $r_s \to r_s-1$). At the algebraic level, it corresponds to truncate the lower-triangular matrices (like $M_\infty$ or $M_{X_s}$) obtained in this article that would no longer be invertible. Nevertheless, if the twisted case is geometrically very different from the untwisted case, the strategy of \cite{MarchalAlameddineP1Hierarchy2023} is the same as the one developed in \cite{MarchalOrantinAlameddine2022} except from the fact that the local diagonalizations have to be understood in $z=\lambda^{-\frac{1}{2}}$ instead of $\lambda$and that definition of irregular times has to be adapted accordingly. In the end, results in the twisted and untwisted cases in $\mathfrak{sl}_2(\mathbb{C})$ give rise to very similar expressions. Thus, if one could write down the isospectral condition for this case (i.e. only Section \ref{SectionConnectionIso} of the present work which is missing in the twisted case) then one could easily obtain the isospectral coordinates in this twisted case. It would be interesting to see if the differential systems obtained in this twisted case are similar to the ones obtained in this article as well as the potential differences on the isospectral coordinates.
\end{itemize}

\section*{Acknowledgements}The authors would like to thank John Harnad for fruitful discussions on the isospectral condition. The authors would also like to thank Nicolas Orantin for his support at the early stage of this project.

\appendix

\renewcommand{\theequation}{\thesection-\arabic{equation}}

\section{Proof of Theorem \ref{GeoLaxMatrices}}\label{AppendixProofGeoLaxMatrices}
\subsection{Entry $(1,2)$}
The choice of coordinates $\mathbf{Q}$ immediately provides the expression for $\td{L}_{1,2}(\lambda)$. Moreover, the gauge transformation \eqref{GaugeTransfo} implies that for $r_\infty\geq 3$
\beq \label{ExpresstdA}G(\lambda)\td{A}_{\boldsymbol{\alpha}}(\lambda)=A_{\boldsymbol{\alpha}}(\lambda)G(\lambda)-\mathcal{L}_{\boldsymbol{\alpha}}[G(\lambda)]\eeq
In particular, 
\beq \label{RefA12}[\td{A}_{\boldsymbol{\alpha}}(\lambda)]_{1,2}=\td{L}_{1,2}(\lambda) [A_{\boldsymbol{\alpha}}(\lambda)]_{1,2}.\eeq 
Since we know the expansion of $[A_{\boldsymbol{\alpha}}(\lambda)]_{1,2}$ at each pole we get for $r_\infty\geq3$:
\bea \label{A12Infinity}[\td{A}_{\boldsymbol{\alpha}}(\lambda)]_{1,2}&\overset{\lambda\to \infty}{=}& \left(\sum_{s=1}^n\sum_{k=1}^{r_s} \frac{Q_{X_s,k}}{(\lambda-X_s)^k}+\sum_{k=0}^{r_\infty-4} Q_{\infty,k}\lambda^k+\omega\,\delta_{r_\infty\geq 3}\lambda^{r_\infty-3}\right)\cr
&&\left(\sum_{i=1}^{r_\infty-3}\nu_{\infty,i}^{(\boldsymbol{\alpha})} \lambda^{-i} +O\left(\lambda^{-r_\infty+2}\right) \right)\cr
&=&\omega\,\delta_{r_\infty\geq 3}\sum_{i=1}^{r_\infty-3}\nu_{\infty,i}^{(\boldsymbol{\alpha})} \lambda^{r_\infty-3-i}+\sum_{k=0}^{r_\infty-4}\sum_{i=1}^{r_\infty-3}Q_{\infty,k}\nu_{\infty,i}^{(\boldsymbol{\alpha})}\lambda^{k-i} + O\left(\lambda^{-1}\right) \cr
&=&\omega\,\delta_{r_\infty\geq 3}\sum_{j=0}^{r_\infty-4}\nu_{\infty,r_\infty-3-j}^{(\boldsymbol{\alpha})} \lambda^{j}+ \sum_{j=0}^{r_\infty-5}\sum_{k=j+1}^{r_\infty-4}Q_{\infty,k}\nu_{\infty,k-j}^{(\boldsymbol{\alpha})} \lambda^j + O\left(\lambda^{-1}\right) \cr
&=&\omega\,\nu_{\infty,1}^{(\boldsymbol{\alpha})}\delta_{r_\infty\geq 4}\lambda^{r_\infty-4}+\sum_{j=0}^{r_\infty-5}\left(\omega\,\nu_{\infty,r_\infty-3-j}^{(\boldsymbol{\alpha})}+\sum_{k=j+1}^{r_\infty-4}\nu_{\infty,k-j}^{(\boldsymbol{\alpha})}Q_{\infty,k}\right)\lambda^j +O(\lambda^{-1})\cr
&&
\eea
For $r_\infty=2$ we have:
\beq\label{A12Infinityrinftyequal2}[\td{A}_{\boldsymbol{\alpha}}(\lambda)]_{1,2}\overset{\lambda\to \infty}{=}\left(\omega\lambda^{-1}+O(\lambda^{-2})\right)\left(\nu_{\infty,0}^{(\boldsymbol{\alpha})} +O(\lambda^{-1})\right)= \omega\,\nu_{\infty,0}^{(\boldsymbol{\alpha})} \lambda^{-1}+ O(\lambda^{-2})
\eeq
For $r_\infty=1$, we have:
\bea\label{A12Infinityrinftyequal1}[\td{A}_{\boldsymbol{\alpha}}(\lambda)]_{1,2}&\overset{\lambda\to \infty}{=}&\left(\omega\lambda^{-2}+\left(\sum_{s=1}^n X_s^2Q_{X_s,1}+2X_sQ_{X_s,2}+Q_{X_s,3}\right)\lambda^{-3}+O(\lambda^{-4})\right)\cr
&&\left(\nu_{\infty,-1}^{(\boldsymbol{\alpha})}\lambda+\nu_{\infty,0}^{(\boldsymbol{\alpha})} +O(\lambda^{-1})\right)\cr
&=& \omega\,\nu_{\infty,-1}^{(\boldsymbol{\alpha})} \lambda^{-1}+\left(\left(\sum_{s=1}^n X_s^2Q_{X_s,1}+2X_sQ_{X_s,2}+Q_{X_s,3}\right)\nu_{\infty,-1}^{(\boldsymbol{\alpha})}+ \omega\,\nu_{\infty,0}^{(\boldsymbol{\alpha})} \right) \lambda^{-2}\cr
&&+O(\lambda^{-3})
\eea
In the same way, for all $s\in \llbracket 1,n\rrbracket$:
\small{\bea [\td{A}_{\boldsymbol{\alpha}}(\lambda)]_{1,2}&\overset{\lambda\to X_s}{=}&\left(\sum_{s'=1}^n\sum_{k=1}^{r_{s'}} \frac{Q_{X_{s'},k}}{(\lambda-X_{s'})^k}+\sum_{k=0}^{r_\infty-4} Q_{\infty,k}\lambda^k+\omega\,\delta_{r_\infty\geq 3}\lambda^{r_\infty-3}\right)\cr
&&\left(\sum_{i=0}^{r_s-1}\nu_{X_s,i}^{(\boldsymbol{\alpha})} (\lambda-X_s)^{i} +O\left((\lambda-X_s)^{r_s}\right) \right)\cr
&=&\sum_{k=1}^{r_s}\sum_{i=0}^{r_s-1} \nu_{X_s,i}^{(\boldsymbol{\alpha})}Q_{X_s,k} (\lambda-X_{s})^{-(k-i)}+O(1)\cr
&=&\sum_{j=1}^{r_s}\left(\sum_{k=j}^{r_s} \nu_{X_s,k-j}^{(\boldsymbol{\alpha})}  Q_{X_s,k}\right)(\lambda-X_s)^{-j} +O(1)
\eea}

\normalsize{Since} we know that $[\td{A}_{\boldsymbol{\alpha}}(\lambda)]_{1,2}$ is a rational function of $\lambda$ with only poles at $\lambda\in\{X_1,\dots,X_n,\infty\}$ we end up with
\bea\left[\td{A}_{\boldsymbol{\alpha}}(\lambda)\right]_{1,2}&=&\omega\,\nu_{\infty,1}^{(\boldsymbol{\alpha})}\lambda^{r_\infty-4}\delta_{r_\infty\geq 4}
+\sum_{j=0}^{r_\infty-5}\left(\omega\,\nu_{\infty,r_\infty-3-j}^{(\boldsymbol{\alpha})}+\sum_{k=j+1}^{r_\infty-4}\nu_{\infty,k-j}^{(\boldsymbol{\alpha})}Q_{\infty,k}\right)\lambda^j\cr
&&+\sum_{s=1}^n\sum_{j=1}^{r_s}\left(\sum_{k=j}^{r_s} \nu_{X_s,k-j}^{(\boldsymbol{\alpha})}  Q_{X_s,k}\right)(\lambda-X_s)^{-j}
\eea
and \eqref{A12Infinityrinftyequal2} and \eqref{A12Infinityrinftyequal1} implies that we get some extra conditions for $r_\infty=2$ and $r_\infty=1$:
\begin{itemize}
\item For $r_\infty=2$: 
\beq \sum_{s=1}^n\left(\sum_{k=1}^{r_s} \nu_{X_s,k-1}^{(\boldsymbol{\alpha})}  Q_{X_s,k}\right)=\omega\,\nu_{\infty,0}^{(\boldsymbol{\alpha})}\eeq
\item For $r_\infty=1$: 
\footnotesize{\bea \omega\,\nu_{\infty,-1}^{(\boldsymbol{\alpha})}&=&\sum_{s=1}^n\left(\sum_{k=1}^{r_s} \nu_{X_s,k-1}^{(\boldsymbol{\alpha})}  Q_{X_s,k}\right)\cr
\left(\sum_{s=1}^n X_s^2Q_{X_s,1}+2X_sQ_{X_s,2}+Q_{X_s,3}\right)\nu_{\infty,-1}^{(\boldsymbol{\alpha})}+\omega\,\nu_{\infty,0}^{(\boldsymbol{\alpha})}&=&\sum_{s=1}^n\left(\sum_{k=2}^{r_s} \nu_{X_s,k-2}^{(\boldsymbol{\alpha})}  Q_{X_s,k}+ X_s\sum_{k=1}^{r_s} \nu_{X_s,k-1}^{(\boldsymbol{\alpha})}  Q_{X_s,k}\right)\cr
&&
\eea}\normalsize{}
\end{itemize}

\subsection{Entry $(1,1)$}
\subsubsection{Computation of $\td{L}_{1,1}(\lambda)$}
Let us start with the following lemma:
\begin{lemma}\label{LemmaQ}We have for all $i\in \llbracket 1,g\rrbracket$:
\bea \partial_{q_i}[Q_{X_s,m}]&=&\sum_{j=m}^{r_s}  (q_i-X_s)^{m-j-1}Q_{X_s,j} \,\,,\,\, \forall \, (s,m)\in \llbracket 1,n\rrbracket\times\llbracket 1,r_s\rrbracket\cr
\partial_{q_i}[Q_{\infty,r_\infty-4}]&=&-\omega\,\,,\,\, \text{ if }r_\infty\geq 4\cr
\partial_{q_i}[Q_{\infty,m}]&=&-\omega \,q_i^{r_\infty-4-m}-\sum_{j=m+1}^{r_\infty-4} Q_{\infty,j} q_i^{j-1-m} \,\,,\,\, \forall \, m\in \llbracket 1,r_\infty-5\rrbracket
\eea
\end{lemma}

\begin{proof}Let $i\in \llbracket 1,g\rrbracket$. Taking the derivative relatively to $q_i$ in the definition of the change of coordinates of Definition \ref{DefGeometricCoordinates} provides
\beq \label{IdentityNew} \frac{-\omega\underset{j\neq i}{\overset{g}{\prod}}(\lambda-q_j)}{\underset{s=1}{\overset{n}{\prod}}(\lambda-X_s)^{r_s}}=\sum_{s=1}^n\sum_{k=1}^{r_s} \frac{\partial_{q_i}[Q_{X_s,k}]}{(\lambda-X_s)^k}+\sum_{k=0}^{r_\infty-4} \partial_{q_i}[Q_{\infty,k}]\lambda^k\eeq
But the l.h.s. is a rational function of $\lambda$ with poles in $\left\{\infty,X_1,\dots,X_n\right\}$. It can be rewritten as
\beq \frac{-\omega\underset{j\neq i}{\overset{g}{\prod}}(\lambda-q_j)}{\underset{s=1}{\overset{n}{\prod}}(\lambda-X_s)^{r_s}}=\frac{-\omega\underset{j=1}{\overset{g}{\prod}}(\lambda-q_j)}{(\lambda-q_i)\underset{s=1}{\overset{n}{\prod}}(\lambda-X_s)^{r_s}}\eeq
We may expand the last quantity at each pole:
\small{\bea \frac{-\omega\underset{j=1}{\overset{g}{\prod}}(\lambda-q_j)}{(\lambda-q_i)\underset{s=1}{\overset{n}{\prod}}(\lambda-X_s)^{r_s}}&\overset{\lambda\to \infty}{=}&-\left(\sum_{s=1}^n\sum_{k=1}^{r_s} \frac{Q_{X_s,k}}{(\lambda-X_s)^k}+\sum_{k=0}^{r_\infty-4} Q_{\infty,k}\lambda^k +\omega\,\delta_{r_\infty\geq 3}\lambda^{r_\infty-3}\right)\sum_{r=0}^{\infty} q_i^r \lambda^{-1-r} \cr
&=&-\sum_{k=0}^{r_\infty-4}\sum_{r=0}^{\infty} Q_{\infty,k} q_i^r \lambda^{k-1-r}-\omega\sum_{r=0}^{r_\infty-4} q_i^r \lambda^{r_\infty-4-r} +O(\lambda^{-1})\cr
&\overset{m=k-1-r}{=}&-\sum_{m=0}^{r_\infty-5}\sum_{k=m+1}^{r_\infty-4} Q_{\infty,k} q_i^{k-1-m} \lambda^{m} -\omega\sum_{m=0}^{r_\infty-4} q_i^{r_\infty-4-m} \lambda^{m} +O(\lambda^{-1})\cr
&&
\eea}
\normalsize{Note} that the expression is also valid for $r_\infty\leq 3$ in which case the r.h.s. is only $O(\lambda^{-1})$. For any $s\in \llbracket 1,n\rrbracket$, we also have:
\small{\bea \frac{-\omega\underset{j=1}{\overset{g}{\prod}}(\lambda-q_j)}{(\lambda-q_i)\underset{s=1}{\overset{n}{\prod}}(\lambda-X_s)^{r_s}}&\overset{\lambda\to X_s}{=}&-\left(\sum_{s=1}^n\sum_{k=1}^{r_s} \frac{Q_{X_s,k}}{(\lambda-X_s)^k}+\sum_{k=0}^{r_\infty-4} Q_{\infty,k}\lambda^k +\omega\,\delta_{r_\infty\geq 3}\lambda^{r_\infty-3}\right)\cr
&&\,\left(\sum_{r=0}^{\infty} (-1)^{r}(X_s-q_i)^{-(r+1)} (\lambda-X_s)^{r}\right) \cr
&=&-\sum_{k=1}^{r_s}\sum_{r=0}^{\infty} Q_{X_s,k} (-1)^{r}(X_s-q_i)^{-(r+1)} (\lambda-X_s)^{r-k} +O(1)\cr
&\overset{m=k-r}{=}&-\sum_{m=1}^{r_s}\sum_{k=m}^{r_s} Q_{X_s,k} (-1)^{k-m}(X_s-q_i)^{m-k-1} (\lambda-X_s)^{-m} +O(1)\cr
&=&\sum_{m=1}^{r_s}\sum_{k=m}^{r_s} Q_{X_s,k} (q_i-X_s)^{m-k-1} (\lambda-X_s)^{-m} +O(1)
\eea}
\normalsize{Thus} we get:
\bea  \frac{-\omega\underset{j=1}{\overset{g}{\prod}}(\lambda-q_j)}{(\lambda-q_i)\underset{s=1}{\overset{n}{\prod}}(\lambda-X_s)^{r_s}}&=&-\omega\sum_{m=0}^{r_\infty-4} q_i^{r_\infty-4-m} \lambda^{m}-\sum_{m=0}^{r_\infty-5}\sum_{j=m+1}^{r_\infty-4} Q_{\infty,j} q_i^{j-1-m} \lambda^{m} \cr
&&+\sum_{s=1}^n\sum_{m=1}^{r_s}\sum_{j=m}^{r_s} Q_{X_s,j} (q_i-X_s)^{m-j-1} (\lambda-X_s)^{-m}\eea
For $r_\infty=2$, the l.h.s. is $O(\lambda^{-2})$ at infinity and for $r_\infty=1$, the l.h.s. is $O(\lambda^{-3})$ at infinity so that we have the additional relations for all $i\in \llbracket 1,g\rrbracket$:
\small{\bea \label{NewRelaations} 0&=&\sum_{s=1}^n\sum_{j=1}^{r_s} Q_{X_s,j} (q_i-X_s)^{-j}\,\,,\,\, \text{ if } r_\infty=2\cr
0&=&\sum_{s=1}^n\sum_{j=1}^{r_s} Q_{X_s,j} (q_i-X_s)^{-j} \,\text{ and }0=\sum_{s=1}^n\sum_{j=2}^{r_s} Q_{X_s,j} (q_i-X_s)^{1-j}\cr
&&+ \sum_{s=1}^n\sum_{j=1}^{r_s} X_sQ_{X_s,j} (q_i-X_s)^{-j}  \,\,,\,\, \text{ if } r_\infty=1
\eea}

\normalsize{Finally}, we may identify the coefficients at each pole with \eqref{IdentityNew} to determine $\partial_{q_i}[Q_{p,k}]$ ending the proof of Lemma \ref{LemmaQ}. Note that the additional relations \eqref{NewRelaations} are consistent with the additional constraints \eqref{AddConstrains} and \eqref{AddConstrains2}. 
\end{proof}

The next step is to provide an expression for $-\frac{Q(\lambda)}{\underset{s=1}{\overset{n}{\prod}}(\lambda-X_s)^{r_s}}$.

\begin{lemma}\label{LemmaQExpression}We have
\bea \label{LLL}-\frac{Q(\lambda)}{\underset{s=1}{\overset{n}{\prod}}(\lambda-X_s)^{r_s}}&=&-\omega \,\delta_{r_\infty\geq 4}P_{\infty,r_\infty-4} -\sum_{m=0}^{r_\infty-5}P_{\infty,m}\left(\lambda^{r_\infty-4-m}+\omega \sum_{j=m+1}^{r_\infty-4} Q_{\infty,j} \lambda^{j-1-m}\right)\cr
&&+\sum_{s=1}^n\sum_{m=1}^{r_s} P_{X_s,m}\left(\sum_{j=m}^{r_s}  (\lambda-X_s)^{m-j-1}Q_{X_s,j}\right)\cr
&=&-\omega \sum_{k=0}^{r_\infty-4} P_{\infty,r_\infty-4-k}\lambda^k-\sum_{k=0}^{r_\infty-5}\sum_{m=0}^{r_\infty-5-k}P_{\infty,m}Q_{\infty,k+1+m}\lambda^k\cr
&&+ \sum_{s=1}^n\sum_{k=1}^{r_s}\sum_{m=1}^{r_s+1-k}P_{X_s,m}Q_{X_s,k+m-1}(\lambda-X_s)^{-k}\eea
\end{lemma}

\begin{proof}Let $r_\infty\geq 3$. By definition, both sides are rational function of $\lambda$ with poles of degree $r_\infty-4$ at infinity and $r_s$ at $X_s$ for all $s\in \llbracket 1,n\rrbracket$. Moreover, for any $i\in \llbracket 1,g\rrbracket$ we have $Q(q_i)=-p_i \underset{s=1}{\overset{n}{\prod}} (q_i-X_s)^{r_s}$ so that
\bea -\frac{Q(q_i)}{\underset{s=1}{\overset{n}{\prod}}(q_i-X_s)^{r_s}}&=&p_i\cr
&\overset{\text{Def. } \ref{DefGeometricCoordinates}}{=}&\delta_{r_\infty\geq 4}P_{\infty,r_\infty-4}\partial_{q_i}[Q_{\infty,r_\infty-4}] + \sum_{k=0}^{r_\infty-5} P_{\infty,k}\partial_{q_i}[Q_{\infty,k}]\cr
&&+\sum_{s=1}^n\sum_{k=1}^{r_s} P_{X_s,k}\partial_{q_i}[Q_{X_s,k}]\cr
&\overset{\text{Lemma }\ref{LemmaQ}}{=}&-\omega\,\delta_{r_\infty\geq 4}P_{\infty,r_\infty-4}-\sum_{k=0}^{r_\infty-5} P_{\infty,k}\left(\omega \,q_i^{r_\infty-4-k}+\sum_{j=k+1}^{r_\infty-4} Q_{\infty,j} q_i^{j-1-k}\right)\cr
&&+\sum_{s=1}^n\sum_{k=1}^{r_s} P_{X_s,k}\left(\sum_{j=k}^{r_s}  (q_i-X_s)^{k-j-1}Q_{X_s,j}\right)
\eea
We recognize the r.h.s. of Lemma \ref{LemmaQExpression} evaluated at $\lambda=q_i$. Thus, for all $i\in \llbracket 1,g\rrbracket$, both sides coincide at $\lambda=q_i$. Multiplying, by $\underset{s=1}{\overset{n}{\prod}}(\lambda-X_s)^{r_s}$ on both sides, we get polynomials of order $g-1$ that coincide at $g$ distinct values so that they are equal ending the proof of Lemma \ref{LemmaQExpression}.
For $r_\infty=2$, the l.h.s. of \eqref{LLL} is of order $\lambda^{-2}$ so that we need to impose the condition $\underset{s=1}{\overset{n}{\sum}}\underset{m=1}{\overset{r_s}{\sum}} P_{X_s,m}Q_{X_s,m}=0$ so that the r.h.s. has the same property. Similarly for $r_\infty=1$, the l.h.s. of \eqref{LLL} is of order $\lambda^{-3}$ so that we need to impose the conditions $\underset{s=1}{\overset{n}{\sum}}\underset{m=1}{\overset{r_s}{\sum}} P_{X_s,m}Q_{X_s,m}=0$ and $\underset{s=1}{\overset{n}{\sum}}\underset{m=1}{\overset{r_s-1}{\sum}}P_{X_s,m}Q_{X_s,m+1} +\underset{s=1}{\overset{n}{\sum}}\underset{m=1}{\overset{r_s}{\sum}} X_sP_{X_s,m}Q_{X_s,m}=0$ so that the r.h.s. has the same property.
\end{proof}

Finally, the gauge transformation \eqref{GaugeTransfo} implies that
\small{\beq \label{tdL11} \td{L}_{1,1}(\lambda)= -\frac{Q(\lambda)}{\underset{s=1}{\overset{n}{\prod}}(\lambda-X_s)^{r_s}} -\frac{(t_{\infty,r_\infty-1}\lambda+g_0)}{\omega}\left(\sum_{s=1}^n\sum_{k=1}^{r_s} \frac{Q_{X_s,k}}{(\lambda-X_s)^k}+\sum_{k=0}^{r_\infty-4} Q_{\infty,k}\lambda^k+\omega \,\delta_{r_\infty\geq 3}\lambda^{r_\infty-3}\right)\eeq}
\normalsize{Let} us now look at the behavior at infinity. 
\begin{itemize}\item If $r_\infty\geq 4$ then from Remark \ref{RemarkCoeff}: $\td{L}_{1,1}(\lambda)=-t_{\infty,r_\infty-1}\lambda^{r_\infty-2}-t_{\infty,r_\infty-2}\lambda^{r_\infty-3}+O(\lambda^{r_\infty-4})$. We have that $\frac{Q(\lambda)}{\underset{s=1}{\overset{n}{\prod}}(\lambda-X_s)^{r_s}}=O(\lambda^{r_\infty-4})$ so that the r.h.s. of \eqref{tdL11} is
\beq \td{L}_{1,1}(\lambda)=-t_{\infty,r_\infty-1}\lambda^{r_\infty-2}-\left(\frac{t_{\infty,r_\infty-1}}{\omega}Q_{\infty,r_\infty-4}+g_0\right)\lambda^{r_\infty-3}+O(\lambda^{r_\infty-4})\eeq
Hence we get
\beq -t_{\infty,r_\infty-2}=-\frac{t_{\infty,r_\infty-1}}{\omega}Q_{\infty,r_\infty-4}-g_0 \,\Leftrightarrow\, g_0=t_{\infty,r_\infty-2}-\frac{t_{\infty,r_\infty-1}}{\omega}Q_{\infty,r_\infty-4}\eeq
\item If $r_\infty=3$ then from Remark \ref{RemarkCoeff}: $\td{L}_{1,1}(\lambda)=-t_{\infty,2}\lambda-t_{\infty,1}+O(\lambda^{-1})$.
We have that $\frac{Q(\lambda)}{\underset{s=1}{\overset{n}{\prod}}(\lambda-X_s)^{r_s}}=O(\lambda^{-1})$ so that the r.h.s. of \eqref{tdL11} is
\beq \td{L}_{1,1}(\lambda)=-t_{\infty,2}\lambda-\left(\frac{t_{\infty,2}}{\omega}\sum_{s=1}^n Q_{X_s,1}+g_0\right)+O(\lambda^{-1})\eeq
Hence we get
\beq -t_{\infty,1}=-\frac{t_{\infty,2}}{\omega}\sum_{s=1}^n Q_{X_s,1}-g_0 \,\Leftrightarrow\, g_0=t_{\infty,1}-\frac{t_{\infty,2}}{\omega}\sum_{s=1}^n Q_{X_s,1}\eeq
\item If $r_\infty=2$ then from Remark \ref{RemarkCoeff}: $\td{L}_{1,1}(\lambda)=-t_{\infty,1}-t_{\infty,0}\lambda^{-1}+O(\lambda^{-2})$. We have also from Lemma \ref{LemmaQExpression}:
\beq -\frac{Q(\lambda)}{\underset{s=1}{\overset{n}{\prod}}(\lambda-X_s)^{r_s}}=
 O(\lambda^{-2})\eeq
so that using, $\underset{s=1}{\overset{n}{\sum}}Q_{X_s,1}=\omega$, the r.h.s. of \eqref{tdL11} is:
\beq -t_{\infty,1} - \left(\frac{t_{\infty,1}}{\omega}\left(\sum_{s=1}^nQ_{X_s,2}+X_sQ_{X_s,1}\right) +g_0
\right)\lambda^{-1} +O(\lambda^{-2})\eeq
Hence we get
\beq \label{g0rfintyequal2}g_0=t_{\infty,0}-\frac{t_{\infty,1}}{\omega}\left(\sum_{s=1}^nQ_{X_s,2}+X_sQ_{X_s,1}\right)
\eeq
\end{itemize}
Hence, Lemma \ref{LemmaQExpression} provides the expression for Theorem \ref{GeoLaxMatrices}.

\subsubsection{Computation of $c_{\infty,0}$}\label{Appendixcinfty0}
Let us now observe that entry $(2,1)$ of the compatibility equations implies:
\beq \mathcal{L}_{\boldsymbol{\alpha}}[\td{L}_{1,2}(\lambda)]=\partial_{\lambda}\td{A}^{(\boldsymbol{\alpha})}_{1,2}(\lambda)-2\left(\td{L}_{1,1}(\lambda)\td{A}^{(\boldsymbol{\alpha})}_{1,2}(\lambda)-\td{L}_{1,2}(\lambda)\td{A}^{(\boldsymbol{\alpha})}_{1,1}(\lambda)\right)\eeq
From the gauge transformation we have $\td{A}^{(\boldsymbol{\alpha})}_{1,2}(\lambda)=\td{L}_{1,2}(\lambda)A^{(\boldsymbol{\alpha})}_{1,2}(\lambda)$ and $\td{A}^{(\boldsymbol{\alpha})}_{1,1}(\lambda)=A^{(\boldsymbol{\alpha})}_{1,1}(\lambda)+\td{L}_{1,1}(\lambda)A^{(\boldsymbol{\alpha})}_{1,2}(\lambda)$ so that we end up with
\beq \mathcal{L}_{\boldsymbol{\alpha}}[\td{L}_{1,2}(\lambda)]=\partial_{\lambda}\td{A}^{(\boldsymbol{\alpha})}_{1,2}(\lambda)+2\td{L}_{1,2}(\lambda)A^{(\boldsymbol{\alpha})}_{1,1}(\lambda)
\eeq
From \eqref{ExpreA11}, we have $A^{(\boldsymbol{\alpha})}_{1,1}(\lambda)=c_{\infty,0}+\frac{a_{1}}{\lambda}+O(\lambda^{-2})$. Let us now discuss the various cases:
\begin{itemize}\item If $r_\infty\geq 4$, we have $\partial_{\lambda}\td{A}^{(\boldsymbol{\alpha})}_{1,2}(\lambda)=O(\lambda^{r_\infty-5})$ and $\td{L}_{1,2}(\lambda)=\omega\lambda^{r_\infty-3}+Q_{\infty,r_\infty-4}\lambda^{r_\infty-4}+ O(\lambda^{r_\infty-5})$  so
\beq \mathcal{L}_{\boldsymbol{\alpha}}[\td{L}_{1,2}(\lambda)]=2\omega c_{\infty,0}\lambda^{r_\infty-3}+ \left(2 \omega a_{1}+2c_{\infty,0} Q_{\infty,r_\infty-4} \right) \lambda^{r_\infty-4}+ O(\lambda^{r_\infty-5})\eeq
Identifying coefficient $\lambda^{r_\infty-3}$ at infinity provides
\beq \mathcal{L}_{\boldsymbol{\alpha}}[\omega]=2\omega c_{\infty,0} \,\,\Rightarrow\,\, c_{\infty,0}=\frac{1}{2\omega}\mathcal{L}_{\boldsymbol{\alpha}}[\omega] \eeq
Identifying coefficient $\lambda^{r_\infty-4}$ at infinity provides
\beq \label{LQinftyrinftygeq4}\mathcal{L}_{\boldsymbol{\alpha}}[Q_{\infty,r_\infty-4}]
=2 \omega a_{1}+\frac{1}{\omega}\mathcal{L}_{\boldsymbol{\alpha}}[\omega]Q_{\infty,r_\infty-4}\eeq
\item For $r_\infty=3$, we have $\td{A}^{(\boldsymbol{\alpha})}_{1,2}(\lambda)=O(\lambda^{-1})$ so that $\partial_\lambda \td{A}^{(\boldsymbol{\alpha})}_{1,2}(\lambda)=O(\lambda^{-2})$. Moreover $\td{L}_{1,2}(\lambda)=\omega+\left(\underset{s=1}{\overset{n}{\sum}} Q_{X_s,1}\right) \lambda^{-1}+ O(\lambda^{-2})$ so that
 \beq \mathcal{L}_{\boldsymbol{\alpha}}[\td{L}_{1,2}(\lambda)]=2\omega c_{\infty,0}+\left(2\omega a_{1}+2c_{\infty,0}\sum_{s=1}^n Q_{X_s,1}\right)\lambda^{-1}+ O(\lambda^{-2})\eeq
Identifying coefficient $\lambda^{0}$ at infinity provides
\beq \mathcal{L}_{\boldsymbol{\alpha}}[\omega]=2\omega c_{\infty,0} \,\,\Rightarrow\,\, c_{\infty,0}=\frac{1}{2\omega}\mathcal{L}_{\boldsymbol{\alpha}}[\omega] \eeq
Identifying coefficient $\lambda^{-1}$ at infinity provides
\beq \label{LQinftyrinftyequal3}\sum_{s=1}^n\mathcal{L}_{\boldsymbol{\alpha}}[Q_{X_s,1}]=2\omega a_1+\frac{1}{\omega}\mathcal{L}_{\boldsymbol{\alpha}}[\omega]\sum_{s=1}^n Q_{X_s,1}\eeq
\item For $r_\infty=2$, we have from \eqref{ExtraConditionsrinftyequal2} $\td{A}^{(\boldsymbol{\alpha})}_{1,2}(\lambda)=\omega\, \nu^{(\boldsymbol{\alpha})}_{\infty,0}\lambda^{-1}+O(\lambda^{-2})$ so that $\partial_{\lambda}\td{A}^{(\boldsymbol{\alpha})}_{1,2}(\lambda)=-\omega\,\nu^{(\boldsymbol{\alpha})}_{\infty,0}\lambda^{-2}+O(\lambda^{-3})$. We have also $\td{L}_{1,2}(\lambda)=\omega\lambda^{-1}+\left(\underset{s=1}{\overset{n}{\sum}} X_sQ_{X_s,1}+Q_{X_s,2}\right)\lambda^{-2}+O(\lambda^{-3})$ and $[A_{\boldsymbol{\alpha}}]_{1,1}=c_{\infty,0}+\frac{a_1}{\lambda} +O(\lambda^{-2})$ so that
\bea \mathcal{L}_{\boldsymbol{\alpha}}[\td{L}_{1,2}(\lambda)]&=&2\omega c_{\infty,0}\lambda^{-1}+ \left(-\omega\, \nu^{(\boldsymbol{\alpha})}_{\infty,0}+2\omega a_1 +2c_{\infty,0}\left(\underset{s=1}{\overset{n}{\sum}} X_sQ_{X_s,1}+Q_{X_s,2}\right)\right)\lambda^{-2}\cr
&&+O(\lambda^{-3})\eea
Identifying coefficient $\lambda^{-1}$ at infinity provides
\beq \mathcal{L}_{\boldsymbol{\alpha}}[\omega]=2\omega c_{\infty,0} \,\,\Rightarrow\,\, c_{\infty,0}=\frac{1}{2\omega}\mathcal{L}_{\boldsymbol{\alpha}}[\omega]\eeq
Identifying coefficient $\lambda^{-2}$ at infinity provides
\beq \label{LQinftyrinftyequal2} \sum_{s=1}^n\mathcal{L}_{\boldsymbol{\alpha}}[X_sQ_{X_s,1}+Q_{X_s,2}]=\frac{1}{\omega}\mathcal{L}_{\boldsymbol{\alpha}}[\omega]\left(\underset{s=1}{\overset{n}{\sum}} X_sQ_{X_s,1}+Q_{X_s,2}\right)-\omega\,\nu^{(\boldsymbol{\alpha})}_{\infty,0}+2\omega a_1\eeq
\item For $r_\infty=1$, we have from \eqref{ExtraConditionsrinftyequal1} 
\beq\td{A}^{(\boldsymbol{\alpha})}_{1,2}(\lambda)=\omega\,\nu^{(\boldsymbol{\alpha})}_{\infty,-1}\lambda^{-1}+\left(\omega\,\nu^{(\boldsymbol{\alpha})}_{\infty,0}+\nu^{(\boldsymbol{\alpha})}_{\infty,-1}\sum_{s=1}^n X_s^2Q_{X_s,1}+2X_sQ_{X_s,2}+Q_{X_s,3}\right)\lambda^{-2}+ O(\lambda^{-3})\eeq
so that
\small{\beq \partial_\lambda \td{A}^{(\boldsymbol{\alpha})}_{1,2}(\lambda)=-\omega\,\nu^{(\boldsymbol{\alpha})}_{\infty,-1}\lambda^{-2}-2\left(\omega\,\nu^{(\boldsymbol{\alpha})}_{\infty,0}+\nu^{(\boldsymbol{\alpha})}_{\infty,-1}\sum_{s=1}^n X_s^2Q_{X_s,1}+2X_sQ_{X_s,2}+Q_{X_s,3}\right)\lambda^{-3}+ O(\lambda^{-4})\eeq}
\normalsize{Moreover}, $\td{L}_{1,2}(\lambda)=\omega\lambda^{-2}+\left(\underset{s=1}{\overset{n}{\sum}} X_s^2Q_{X_s,1}+2X_sQ_{X_s,2}+Q_{X_s,3}\right)\lambda^{-3}+O(\lambda^{-4})$ and $A^{(\boldsymbol{\alpha})}_{1,1}(\lambda)=c_{\infty,0}+\frac{a_1}{\lambda}+O(\lambda^{-2})$ so that
\small{\bea \mathcal{L}_{\boldsymbol{\alpha}}[\td{L}_{1,2}(\lambda)]&=&\left(2\omega c_{\infty,0}-\omega\,\nu^{(\boldsymbol{\alpha})}_{\infty,-1}\right)\lambda^{-2}\cr
&&+ 2\left(\omega a_1-\omega\, \nu^{(\boldsymbol{\alpha})}_{\infty,0}+(c_{\infty,0}-\nu^{(\boldsymbol{\alpha})}_{\infty,-1})\sum_{s=1}^n X_s^2Q_{X_s,1}+2X_sQ_{X_s,2}+Q_{X_s,3}\right)\lambda^{-3}\cr
&&+ O(\lambda^{-4})\eea}
\normalsize{Identifying} coefficient $\lambda^{-2}$ at infinity provides
\beq \mathcal{L}_{\boldsymbol{\alpha}}[\omega]=2\omega c_{\infty,0}-\omega\,\nu^{(\boldsymbol{\alpha})}_{\infty,-1}\,\,\Rightarrow\,\, c_{\infty,0}=\frac{1}{2\omega}\mathcal{L}_{\boldsymbol{\alpha}}[\omega] +\frac{1}{2}\nu^{(\boldsymbol{\alpha})}_{\infty,-1}\eeq
Identifying coefficient $\lambda^{-3}$ at infinity provides
\bea \label{LQinftyrinftyequal1}
&&\sum_{s=1}^n\mathcal{L}_{\boldsymbol{\alpha}}[X_s^2Q_{X_s,1}+2X_sQ_{X_s,2}+Q_{X_s,3}]=2\omega\,a_1-2\omega\,\nu^{(\boldsymbol{\alpha})}_{\infty,0}\cr
&&+\left(\frac{1}{\omega}\mathcal{L}_{\boldsymbol{\alpha}}[\omega]-\nu^{(\boldsymbol{\alpha})}_{\infty,-1}\right)\sum_{s=1}^n X_s^2Q_{X_s,1}+2X_sQ_{X_s,2}+Q_{X_s,3}
\eea
\end{itemize}

\subsubsection{Computation of $[\td{A}_{\boldsymbol{\alpha}}(\lambda)]_{1,1}$}
Entry $(1,1)$ of the gauge transformation \eqref{ExpresstdA} implies that
\beq \label{GaugeTdA11}[\td{A}_{\boldsymbol{\alpha}}(\lambda)]_{1,1}=[A_{\boldsymbol{\alpha}}(\lambda)]_{1,1}+\td{L}_{1,1}(\lambda)[A_{\boldsymbol{\alpha}}(\lambda)]_{1,2}\eeq
From \cite{MarchalOrantinAlameddine2022}, $[A_{\boldsymbol{\alpha}}(\lambda)]_{1,1}$ is given by \eqref{ExpreA11}. Since we know that $[\td{A}_{\boldsymbol{\alpha}}(\lambda)]_{1,1}$ is a rational function of $\lambda$ with poles in $\left\{X_1,\dots,X_n,\infty\right\}$, expression \eqref{ExpreA11} implies that $[A_{\boldsymbol{\alpha}}(\lambda)]_{1,1}$ only contributes with a factor $c_{\infty,0}$. Hence, \eqref{GaugeTdA11} reduces to
\beq \label{GaugeTdA11Reduced}[\td{A}_{\boldsymbol{\alpha}}(\lambda)]_{1,1}=c_{\infty,0}+\td{L}_{1,1}(\lambda)[A_{\boldsymbol{\alpha}}(\lambda)]_{1,2}\eeq
We now need to evaluate the r.h.s. at each pole. We have for any $s\in \llbracket 1,n\rrbracket$:
\footnotesize{\bea &&A_{1,2}(\lambda)\td{L}_{1,1}(\lambda)\overset{\lambda\to X_s}{=}\left(\sum_{i=0}^{r_s-1}\nu_{X_s,i}^{(\boldsymbol{\alpha})} (\lambda-X_s)^{i} +O\left((\lambda-X_s)^{r_s}\right)\right)\cr
&&\left(\sum_{k=1}^{r_s}\sum_{m=1}^{r_s+1-k}\frac{P_{X_s,m}Q_{X_s,k+m-1}}{(\lambda-X_s)^{k}}-\frac{(t_{\infty,r_\infty-1}X_s+g_0)}{\omega}\sum_{k=1}^{r_s}\frac{Q_{X_s,k}}{(\lambda-X_s)^k}-\frac{t_{\infty,r_\infty-1}}{\omega}\sum_{k=1}^{r_s}\frac{Q_{X_s,k}}{(\lambda-X_s)^{k-1}} +O(1)\right)\cr
&&=\sum_{i=0}^{r_s-1}\sum_{k=1}^{r_s}\sum_{m=1}^{r_s+1-k}\nu_{X_s,i}^{(\boldsymbol{\alpha})}P_{X_s,m}Q_{X_s,k+m-1}(\lambda-X_s)^{i-k}- \frac{(t_{\infty,r_\infty-1}X_s+g_0)}{\omega}\sum_{i=0}^{r_s-1}\sum_{k=1}^{r_s}\nu_{X_s,i}^{(\boldsymbol{\alpha})} Q_{X_s,k} (\lambda-X_s)^{i-k}\cr
&&- \frac{t_{\infty,r_\infty-1}}{\omega}\sum_{i=0}^{r_s-1}\sum_{k=1}^{r_s}\nu_{X_s,i}^{(\boldsymbol{\alpha})}Q_{X_s,k}(\lambda-X_s)^{i+1-k}+O(1)\cr
&&\overset{k=r+i}{=}\sum_{r=1}^{r_s}\sum_{i=0}^{r_s-r}\sum_{m=1}^{r_s+1-r-i}\nu_{X_s,i}^{(\boldsymbol{\alpha})}P_{X_s,m}Q_{X_s,r+i+m-1}(\lambda-X_s)^{-r}\cr
&&- \frac{(t_{\infty,r_\infty-1}X_s+g_0)}{\omega}\sum_{r=1}^{r_s}\sum_{i=0}^{r_s-r}\nu_{X_s,i}^{(\boldsymbol{\alpha})} Q_{X_s,r+i} (\lambda-X_s)^{-r}\cr
&&- \frac{t_{\infty,r_\infty-1}}{\omega}\sum_{r=1}^{r_s-1}\sum_{i=0}^{r_s-1-r}\nu_{X_s,i}^{(\boldsymbol{\alpha})}Q_{X_s,r+i+1}(\lambda-X_s)^{-r}+O(1)\cr
&=&\left(P_{X_s,1}- \frac{t_{\infty,r_\infty-1}X_s+g_0}{\omega}\right)Q_{X_s,r_s}\nu_{X_s,0}^{(\boldsymbol{\alpha})}(\lambda-X_s)^{-r_s}\cr
&&+ \sum_{r=1}^{r_s-1}\Big[\sum_{i=0}^{r_s-r}\sum_{m=1}^{r_s+1-r-i}\nu_{X_s,i}^{(\boldsymbol{\alpha})}P_{X_s,m}Q_{X_s,r+i+m-1}-\frac{(t_{\infty,r_\infty-1}X_s+g_0)}{\omega}\sum_{i=0}^{r_s-r}\nu_{X_s,i}^{(\boldsymbol{\alpha})} Q_{X_s,r+i}\cr
&&- \frac{t_{\infty,r_\infty-1}}{\omega}\sum_{i=0}^{r_s-1-r}\nu_{X_s,i}^{(\boldsymbol{\alpha})}Q_{X_s,r+i+1}\Big](\lambda-X_s)^{-r}+O(1)
\eea}
\normalsize{At} infinity, for $r_\infty\geq 3$, we have:
\footnotesize{\bea &&A_{1,2}(\lambda)\td{L}_{1,1}(\lambda)\overset{\lambda\to \infty}{=}\left(\sum_{i=1}^{r_\infty-2}\nu_{\infty,i}^{(\boldsymbol{\alpha})} \lambda^{-i} +O\left(\lambda^{-r_\infty+1}\right)\right)\cr
&&\Big[-\omega\sum_{k=0}^{r_\infty-4} P_{\infty,r_\infty-4-k}\lambda^k-\sum_{k=0}^{r_\infty-5}\sum_{m=0}^{r_\infty-5-k}P_{\infty,m}Q_{\infty,k+1+m}\lambda^k-t_{\infty,r_\infty-1} \lambda^{r_\infty-2}\cr
&&-\frac{1}{\omega} \sum_{k=0}^{r_\infty-4}(t_{\infty,r_\infty-1}Q_{\infty,k-1}\delta_{k\geq 1}+g_0Q_{\infty,k})\lambda^{k}+ O(\lambda^{-1})\Big]
\cr
&&=-t_{\infty,r_\infty-1}\sum_{i=1}^{r_\infty-2}\nu_{\infty,i}^{(\boldsymbol{\alpha})} \lambda^{r_\infty-2-i} +\sum_{i=1}^{r_\infty-4}\nu_{\infty,i}^{(\boldsymbol{\alpha})}\left(-\omega P_{\infty,0}-\frac{t_{\infty,r_\infty-1}Q_{\infty,r_\infty-5}\delta_{r_\infty\geq 5}+Q_{\infty,r_\infty-4}g_0}{\omega}\right) \lambda^{r_\infty-4-i} \cr
&&+ \sum_{i=1}^{r_\infty-5}\nu_{\infty,i}^{(\boldsymbol{\alpha})}\sum_{k=0}^{r_\infty-5}\left(-\omega P_{\infty,r_\infty-4-k}-\sum_{m=0}^{r_\infty-5-k}P_{\infty,m}Q_{\infty,k+1+m}-\frac{t_{\infty,r_\infty-1}Q_{\infty,k-1}\delta_{k\geq 1}+g_0Q_{\infty,k}}{\omega}\right)\lambda^{k-i}\cr
&&+ O(\lambda^{-1})\cr
&&=-t_{\infty,r_\infty-1}\sum_{j=0}^{r_\infty-3}\nu_{\infty,r_\infty-2-j}^{(\boldsymbol{\alpha})} \lambda^{j}+\sum_{j=0}^{r_\infty-5}\nu_{\infty,r_\infty-4-j}^{(\boldsymbol{\alpha})}\left(-\omega P_{\infty,0}-\frac{t_{\infty,r_\infty-1}Q_{\infty,r_\infty-5}+Q_{\infty,r_\infty-4}g_0}{\omega}\right) \lambda^{j} \cr
&&+ \sum_{j=0}^{r_\infty-6}\sum_{i=1}^{r_\infty-5-j}\nu_{\infty,i}^{(\boldsymbol{\alpha})}\left(- \omega P_{\infty,r_\infty-4-i-j}-\sum_{m=0}^{r_\infty-5-i-j}P_{\infty,m}Q_{\infty,j+i+1+m}-\frac{t_{\infty,r_\infty-1}Q_{\infty,j+i-1}+g_0Q_{\infty,j+i}}{\omega}\right)\lambda^{j}\cr
&&+ O(\lambda^{-1})\cr
&&=-t_{\infty,r_\infty-1}\nu_{\infty,1}^{(\boldsymbol{\alpha})}\lambda^{r_\infty-3}\delta_{r_\infty\geq 3}-t_{\infty,r_\infty-1}\nu_{\infty,2}^{(\boldsymbol{\alpha})}\lambda^{r_\infty-4}\delta_{r_\infty\geq 4}\cr
&&-\left(t_{\infty,r_\infty-1}\nu_{\infty,3}^{(\boldsymbol{\alpha})}+\left(\omega P_{\infty,0}+\frac{t_{\infty,r_\infty-1}Q_{\infty,r_\infty-5}+Q_{\infty,r_\infty-4}g_0}{\omega}\right)\nu_{\infty,1}^{(\boldsymbol{\alpha})}\right)\lambda^{r_\infty-5}\delta_{r_\infty\geq 5}\cr
&&-\sum_{j=0}^{r_\infty-6}\Big[t_{\infty,r_\infty-1}\nu_{\infty,r_\infty-2-j}^{(\boldsymbol{\alpha})}+\nu_{\infty,r_\infty-4-j}^{(\boldsymbol{\alpha})}\left(\omega P_{\infty,0}+\frac{t_{\infty,r_\infty-1}Q_{\infty,r_\infty-5}+Q_{\infty,r_\infty-4}g_0}{\omega}\right)\cr
&&+\sum_{i=1}^{r_\infty-5-j}\nu_{\infty,i}^{(\boldsymbol{\alpha})}\left(\omega P_{\infty,r_\infty-4-i-j}+\sum_{m=0}^{r_\infty-5-i-j}P_{\infty,m}Q_{\infty,j+i+1+m}+\frac{t_{\infty,r_\infty-1}Q_{\infty,j+i-1}+g_0Q_{\infty,j+i}}{\omega}\right)\Big]\lambda^j\cr
&&+ O(\lambda^{-1})\cr
&&
\eea}

\normalsize{Note} that for $r_\infty\geq 3$, the term in $\lambda^0$ implies $\nu_{\infty,r_\infty-2}^{(\boldsymbol{\alpha})}$. In order to determine it, we simply look at order $\lambda^{-1}$ at $\lambda\to \infty$ of $[\td{A}_{\boldsymbol{\alpha}}(\lambda)]_{1,2}$ in Theorem \ref{GeoLaxMatrices}. We find:
\beq \text{For } r_\infty\geq 3\,:\, \sum_{s=1}^n\sum_{k=1}^{r_s}\nu_{X_s,k-1}^{(\boldsymbol{\alpha})}Q_{X_s,k}=\omega\,\nu_{\infty,r_\infty-2}^{(\boldsymbol{\alpha})}+\sum_{j=1}^{r_\infty-3}\nu_{\infty,j}^{(\boldsymbol{\alpha})}Q_{\infty,j-1}\eeq
Since, we shall need it later, one may also determine $\nu_{\infty,r_\infty-1}^{(\boldsymbol{\alpha})}$ by looking at order $\lambda^{-2}$ at $\lambda\to \infty$ of $[\td{A}_{\boldsymbol{\alpha}}(\lambda)]_{1,2}$ in Theorem \ref{GeoLaxMatrices} and we find
\beq  \omega\,\nu_{\infty,r_\infty-1}^{(\boldsymbol{\alpha})}=\sum_{s=1}^n\sum_{k=2}^{r_s}\nu_{X_s,k-2}^{(\boldsymbol{\alpha})}Q_{X_s,k}+\sum_{s=1}^n\sum_{k=1}^{r_s}\nu_{X_s,k-1}^{(\boldsymbol{\alpha})}X_sQ_{X_s,k}-\sum_{j=2}^{r_\infty-2}\nu_{\infty,j}^{(\boldsymbol{\alpha})}Q_{\infty,j-2}-\nu_{\infty,1}^{(\boldsymbol{\alpha})}\left(\sum_{s=1}^n Q_{X_s,1}\right)\eeq

For $r_\infty=2$, we have $[A_{\boldsymbol{\alpha}}(\lambda)]_{1,2}=\nu_{\infty,0}^{(\boldsymbol{\alpha})} +O(\lambda^{-1})$ and $\td{L}_{1,1}(\lambda)=-t_{\infty,1} -t_{\infty,0}\lambda^{-1}+O(\lambda^{-2})$, thus we end up with  $\td{L}_{1,1}(\lambda)[A_{\boldsymbol{\alpha}}(\lambda)]_{1,2}=-t_{\infty,1}\nu_{\infty,0}^{(\boldsymbol{\alpha})} +O(\lambda^{-1})$.
For $r_\infty=1$ we have $[A_{\boldsymbol{\alpha}}(\lambda)]_{1,2}=\nu_{\infty,-1}^{(\boldsymbol{\alpha})}\lambda+ \nu_{\infty,0}^{(\boldsymbol{\alpha})} +O(\lambda^{-1})$ and $\td{L}_{1,1}(\lambda)=-t_{\infty,0}\lambda^{-1}+O(\lambda^{-2})$ so that 
\beq \label{tdA11rinftyequal1} \text{For }r_\infty=1\,\,:\,\, [\td{A}_{\boldsymbol{\alpha}}(\lambda)]_{1,1}=c_{\infty,0}-t_{\infty,0}\nu_{\infty,-1}^{(\boldsymbol{\alpha})}+O(\lambda^{-1})\eeq

Thus we end up with the formula given by Theorem \ref{GeoLaxMatrices}.

\subsection{Entry $(2,1)$}
\subsubsection{Computation of $\td{L}_{2,1}$}\label{SectionAppL21}
The gauge transformation \eqref{GaugeTransfo} implies that
\beq \label{GaugeTdL} G(\lambda)\td{L}(\lambda)=L(\lambda)G(\lambda)-\partial_\lambda G(\lambda)\eeq
Entry $(2,1$) of the previous equality provides
\beq \td{L}_{1,1}(\lambda)^2+\td{L}_{1,2}(\lambda)\td{L}_{2,1}(\lambda)=L_{2,1}(\lambda)+\td{L}_{1,1}(\lambda)L_{2,2}(\lambda)-\partial_\lambda\td{L}_{1,1}(\lambda)\eeq
Entry $(2,2)$ of \eqref{GaugeTdL} provides $L_{2,2}(\lambda)=\frac{\partial_\lambda \td{L}_{1,2}(\lambda) }{\td{L}_{1,2}(\lambda)}$. Thus, we end up with
\beq \label{GaugeTdL21}\td{L}_{2,1}(\lambda)=\frac{L_{2,1}(\lambda)}{\td{L}_{1,2}(\lambda)}-\frac{\td{L}_{1,1}(\lambda)^2}{\td{L}_{1,2}(\lambda)}-\partial_\lambda\left(\frac{\td{L}_{1,1}(\lambda)}{\td{L}_{1,2}(\lambda)}\right)\eeq

We now use the fact that $\td{L}_{2,1}$ is a rational function of $\lambda$ with poles only in $\{\infty,X_1,\dots,X_n\}$. Thus, we may simply study the behavior of each term at these poles. Let us first discuss the term $\frac{L_{2,1}(\lambda)}{\td{L}_{1,2}(\lambda)}$. We recall that $\td{L}_{1,2}(\lambda)=\omega \,\delta_{r_\infty\geq 3}\lambda^{r_\infty-3}+\underset{k=0}{\overset{r_\infty-4}{\sum}} Q_{\infty,k}\lambda^k+\underset{s=1}{\overset{n}{\sum}}\underset{k=1}{\overset{r_s}{\sum}} \frac{Q_{X_s,k}}{(\lambda-X_s)^k}$ and that from \eqref{L21}
\beq L_{2,1}(\lambda)= -\td{P}_2(\lambda) +\sum_{j=0}^{r_\infty-4}H_{\infty,j}\lambda^j+\sum_{s=1}^n\sum_{j=1}^{r_s}H_{X_s,j}(\lambda-X_s)^{-j}-
 \lambda^{r_\infty-3}\delta_{r_\infty\geq 3}-\sum_{j=1}^{g} \frac{p_j}{\lambda-q_j}
\eeq 
Thus we have  
\bea\label{L21overtdL12} \frac{L_{2,1}(\lambda)}{\td{L}_{1,2}(\lambda)}&=&\frac{L_{2,1}(\lambda)}{\omega \,\delta_{r_\infty\geq 3}\lambda^{r_\infty-3}+\underset{k=0}{\overset{r_\infty-4}{\sum}} Q_{\infty,k}\lambda^k+\underset{s=1}{\overset{n}{\sum}}\underset{k=1}{\overset{r_s}{\sum}} \frac{Q_{X_s,k}}{(\lambda-X_s)^k}}\cr
&\overset{\lambda\to X_s}{=}&\frac{\left( -\td{P}_2(\lambda) +O\left((\lambda-X_s)^{-r_s}\right)\right)}{\omega \,\delta_{r_\infty\geq 3}\lambda^{r_\infty-3}+\underset{k=0}{\overset{r_\infty-4}{\sum}} Q_{\infty,k}\lambda^k+\underset{s=1}{\overset{n}{\sum}}\underset{k=1}{\overset{r_s}{\sum}} \frac{Q_{X_s,k}}{(\lambda-X_s)^k}} \cr
&\overset{\lambda\to X_s}{=}&\frac{-\td{P}_2(\lambda)}{\omega \,\delta_{r_\infty\geq 3}\lambda^{r_\infty-3}+\underset{k=0}{\overset{r_\infty-4}{\sum}} Q_{\infty,k}\lambda^k+\underset{s=1}{\overset{n}{\sum}}\underset{k=1}{\overset{r_s}{\sum}} \frac{Q_{X_s,k}}{(\lambda-X_s)^k}}+  O(1)\cr
&\overset{\lambda\to X_s}{=}&\left[\frac{\underset{j=r_s+1}{\overset{2r_s}{\sum}}\left(\underset{m=0}{\overset{2r_s-j}{\sum}} t_{X_s,r_s-1-m}t_{X_s,j+m-r_s-1}\right) (\lambda-X_s)^{-j}}{\omega \,\delta_{r_\infty\geq 3}\lambda^{r_\infty-3}+\underset{k=0}{\overset{r_\infty-4}{\sum}} Q_{\infty,k}\lambda^k+\underset{s=1}{\overset{n}{\sum}}\underset{k=1}{\overset{r_s}{\sum}} \frac{Q_{X_s,k}}{(\lambda-X_s)^k}}\right]_{X_s,-} +O(1)\cr
&&
\eea
At infinity, we have
\begin{itemize}\item For $r_\infty\geq 3$:
\bea \label{L21overtdL21rinftygeq3}\frac{L_{2,1}(\lambda)}{\td{L}_{1,2}(\lambda)}&\overset{\lambda\to \infty}{=}&  \frac{-\td{P}_2(\lambda)-t_{\infty,r_\infty-1}\lambda^{r_\infty-3} + O(\lambda^{r_\infty-4})}{\omega \,\lambda^{r_\infty-3}+\underset{k=0}{\overset{r_\infty-4}{\sum}} Q_{\infty,k}\lambda^k+\underset{s=1}{\overset{n}{\sum}}\underset{k=1}{\overset{r_s}{\sum}} \frac{Q_{X_s,k}}{(\lambda-X_s)^k}}\cr
&\overset{\lambda\to \infty}{=}&-\frac{t_{\infty,r_\infty-1}}{\omega}+ \left[\frac{\underset{j=r_\infty-3}{\overset{2r_\infty-4}{\sum}}\left(\underset{m=0}{\overset{2r_\infty-4-j}{\sum}} t_{\infty,r_\infty-1-m}t_{\infty,j+m-r_\infty+3}\right) \lambda^{j}}{\omega \,\lambda^{r_\infty-3}+\underset{k=0}{\overset{r_\infty-4}{\sum}} Q_{\infty,k}\lambda^k+\underset{s=1}{\overset{n}{\sum}}\underset{k=1}{\overset{r_s}{\sum}} \frac{Q_{X_s,k}}{(\lambda-X_s)^k}}\right]_{\infty,+} \cr
&&+O(\lambda^{-1})
\eea
\item For $r_\infty=2$: From \eqref{ConditionsAddrinftyequal2} we have:
\beq L_{2,1}(\lambda)\overset{\lambda\to \infty}{=}(t_{\infty,1})^2+(2t_{\infty,1}t_{\infty,0}-t_{\infty,1})\lambda^{-1}+O(\lambda^{-2})\eeq
while
\beq \td{L}_{1,2}(\lambda) \overset{\lambda\to \infty}{=}\omega \lambda^{-1}+\left(\sum_{s=1}^n(Q_{X_s,2}+X_sQ_{X_s,1})\right)\lambda^{-2} +O(\lambda^{-3})\eeq
Thus, we obtain
\beq \label{L21overtdL21rinftyequal2}\frac{L_{2,1}(\lambda)}{\td{L}_{1,2}(\lambda)}\overset{\lambda\to \infty}{=}\frac{(t_{\infty,1})^2}{\omega}\lambda-\frac{(t_{\infty,1})^2}{\omega^2}\left(\sum_{s=1}^n(Q_{X_s,2}+X_sQ_{X_s,1})\right)+\frac{2t_{\infty,1}t_{\infty,0}-t_{\infty,1}}{\omega}+O(\lambda^{-1})\eeq
Moreover, we have:
\bea-\frac{(\td{L}_{1,1}(\lambda))^2}{\td{L}_{1,2}(\lambda)}&\overset{\lambda\to\infty}{=}&-\frac{\left(t_{\infty,1}+t_{\infty,0}\lambda^{-1}+O(\lambda^{-2})\right)^2}{\omega \lambda^{-1}+\left(\underset{s=1}{\overset{n}{\sum}} X_sQ_{X_s,1}+Q_{X_s,2}\right)\lambda^{-2}+O(\lambda^{-3})}\cr
&=&-\frac{(t_{\infty,1})^2}{\omega}\lambda-\frac{2t_{\infty,0}t_{\infty,1}}{\omega}+\frac{(t_{\infty,1})^2}{\omega^2}\left(\sum_{s=1}^n X_sQ_{X_s,1}+Q_{X_s,2}\right)+O(\lambda^{-1})\cr
&&
\eea

\item For $r_\infty=1$: We have from \eqref{ConditionsAddrinftyequal1}
\beq L_{2,1}(\lambda)\overset{\lambda\to \infty}{=}\left(t_{\infty,0}(t_{\infty,0}-1)
\right)\lambda^{-2}+O(\lambda^{-3})\eeq
while
\beq \td{L}_{1,2}(\lambda)\overset{\lambda\to \infty}{=}\omega \lambda^{-2}+O(\lambda^{-3})\eeq
Thus, we obtain
\beq\label{L21overtdL21rinftyequal1}\frac{L_{2,1}(\lambda)}{\td{L}_{1,2}(\lambda)}\overset{\lambda\to \infty}{=}\frac{t_{\infty,0}(t_{\infty,0}-1)}{\omega}
+O(\lambda^{-1})\eeq
Moreover, we have:
\beq-\frac{(\td{L}_{1,1}(\lambda))^2}{\td{L}_{1,2}(\lambda)}\overset{\lambda\to\infty}{=}-\frac{\left(t_{\infty,0}\lambda^{-1}+O(\lambda^{-2})\right)^2}{\omega \lambda^{-2}+O(\lambda^{-3})}=-\frac{(t_{\infty,0})^2}{\omega}+O(\lambda^{-1})
\eeq
\end{itemize}

Let us now discuss the term $-\partial_\lambda\left(\frac{\td{L}_{1,1}(\lambda)}{\td{L}_{1,2}(\lambda)}\right)$. Since both $\td{L}_{1,1}(\lambda)$ and $\td{L}_{1,2}(\lambda)$ are singular at $\lambda\to X_s$ of order $(\lambda-X_s)^{r_s}$, it is obvious that
\beq-\partial_\lambda\left(\frac{\td{L}_{1,1}(\lambda)}{\td{L}_{1,2}(\lambda)}\right)\overset{\lambda\to X_s}{=}O(1) \,\,,\,\, \forall\, s\in\llbracket 1,n\rrbracket\eeq
Because of the normalization at infinity, we have:
\beq-\partial_\lambda\left(\frac{\td{L}_{1,1}(\lambda)}{\td{L}_{1,2}(\lambda)}\right)\overset{\lambda\to \infty}{=}-\partial_{\lambda}\left(\frac{-t_{\infty,r_\infty-1}\lambda^{r_\infty-2}+O(\lambda^{r_\infty-3})}{\omega\lambda^{r_\infty-3}+ O(\lambda^{r_\infty-4})} \right)\overset{\lambda\to \infty}{=}\frac{t_{\infty,r_\infty-1}}{\omega}+O(\lambda^{-2})\eeq
so that the only contribution of $-\partial_\lambda\left(\frac{\td{L}_{1,1}(\lambda)}{\td{L}_{1,2}(\lambda)}\right)$ in \eqref{GaugeTdL21} is $\frac{t_{\infty,r_\infty-1}}{\omega}$. Regrouping the contributions of each term at each pole, we end up with formulas of Theorem \ref{GeoLaxMatrices}.

\subsubsection{The special case of $r_\infty=1$}
Finally, let us discuss the special case of $r_\infty=1$. The previous contributions indicate that we may rewrite $\td{L}_{2,1}(\lambda)$ as
\bea \td{L}_{2,1}(\lambda)&=&
\sum_{s=1}^n\left[\frac{\underset{j=r_s+1}{\overset{2r_s}{\sum}}\left(\underset{m=0}{\overset{2r_s-j}{\sum}} t_{X_s,r_s-1-m}t_{X_s,j+m-r_s-1}\right) (\lambda-X_s)^{-j} -\mathring{L}_{1,1}(\lambda)^2}{\td{L}_{1,2}(\lambda)}\right]_{X_s,-}\cr
&&+\sum_{s=1}^n\left[\frac{2(t_{\infty,0}\lambda+g_0)}{\omega}\mathring{L}_{1,1}(\lambda)-\frac{(t_{\infty,0}\lambda+g_0)^2}{\omega^2}\td{L}_{1,2}(\lambda)\right]_{X_s,-}\cr
&&
\eea
Let us now observe that $\frac{2(t_{\infty,0}\lambda+g_0)}{\omega}\mathring{L}_{1,1}(\lambda)-\frac{(t_{\infty,0}\lambda+g_0)^2}{\omega^2}\td{L}_{1,2}(\lambda)$ is a rational function of $\lambda$ with poles in $\mathcal{R}$. At infinity we have
\beq \frac{2(t_{\infty,0}\lambda+g_0)}{\omega}\mathring{L}_{1,1}(\lambda)-\frac{(t_{\infty,0}\lambda+g_0)^2}{\omega^2}\td{L}_{1,2}(\lambda)\overset{\lambda\to \infty}{=}-\frac{(t_{\infty,0})^2}{\omega}+O(\lambda^{-1})\eeq
so that we get
\bea&&\frac{2(t_{\infty,0}\lambda+g_0)}{\omega}\mathring{L}_{1,1}(\lambda)-\frac{(t_{\infty,0}\lambda+g_0)^2}{\omega^2}\td{L}_{1,2}(\lambda)=\cr
&&\sum_{s=1}^n\left[\frac{(t_{\infty,0}\lambda+g_0)}{\omega}\mathring{L}_{1,1}(\lambda)-\frac{(t_{\infty,0}\lambda+g_0)^2}{\omega^2}\td{L}_{1,2}(\lambda)\right]_{X_s,-}-\frac{(t_{\infty,0})^2}{\omega}\eea
Thus, we end up with
\bea \label{SuperFormulatdL21}\td{L}_{2,1}(\lambda)&=&\frac{(t_{\infty,0})^2}{\omega}+\sum_{s=1}^n\left[\frac{\underset{j=r_s+1}{\overset{2r_s}{\sum}}\left(\underset{m=0}{\overset{2r_s-j}{\sum}} t_{X_s,r_s-1-m}t_{X_s,j+m-r_s-1}\right) (\lambda-X_s)^{-j} -\mathring{L}_{1,1}(\lambda)^2}{\td{L}_{1,2}(\lambda)}\right]_{X_s,-}\cr
&&+\frac{2(t_{\infty,0}\lambda+g_0)}{\omega}\mathring{L}_{1,1}(\lambda)-\frac{(t_{\infty,0}\lambda+g_0)^2}{\omega^2}\td{L}_{1,2}(\lambda)
\eea
Let us now obtain the expression for $g_0$. We first observe that $\mathring{L}_{1,1}(\lambda)$ is independent of $g_0$ by definition. Moreover, the gauge transformation and the normalization at infinity imply that 
\beq g_0=\frac{1}{\omega}\left(\frac{1}{2}-t_{\infty,0}\right)\left(\sum_{s=1}^n X_s^2Q_{X_s,1}+2X_sQ_{X_s,2}+Q_{X_s,3}\right)-\frac{1}{2t_{\infty,0}}\,\underset{\lambda\to \infty}{\Res} \lambda^2 L_{2,1}(\lambda)\eeq
and 
\bea \label{L21bis} L_{2,1}(\lambda)&=&(\td{L}_{1,1})^2+\td{L}_{1,2}(\lambda)\td{L}_{2,1}(\lambda)+\partial_\lambda \td{L}_{1,1}(\lambda)-\td{L}_{1,1}(\lambda)\frac{\partial_\lambda\td{L}_{1,2}(\lambda)}{\td{L}_{1,2}(\lambda)}\cr
&=&(\mathring{L}_{1,1}(\lambda))^2+\partial_\lambda \mathring{L}_{1,1}(\lambda)-\frac{t_{\infty,0}}{\omega}\td{L}_{1,2}(\lambda)-\mathring{L}_{1,1}(\lambda)\frac{\partial_\lambda\td{L}_{1,2}(\lambda)}{\td{L}_{1,2}(\lambda)}+\td{L}_{1,2}(\lambda)\td{L}_{2,1}(\lambda)\cr
&&-\frac{(t_{\infty,0}\lambda+g_0)}{\omega}\td{L}_{1,2}(\lambda)\mathring{L}_{1,1}(\lambda)+\frac{(t_{\infty,0}\lambda+g_0)^2}{\omega^2}\td{L}_{1,2}(\lambda)^2
\eea
Using the specific formula for $\td{L}_{2,1}(\lambda)$ given by \eqref{SuperFormulatdL21}, we end up with
\bea \label{SuperIdL21} &&L_{2,1}(\lambda)=(\mathring{L}_{1,1}(\lambda))^2+\partial_\lambda \mathring{L}_{1,1}(\lambda)-\frac{t_{\infty,0}}{\omega}\td{L}_{1,2}(\lambda)-\mathring{L}_{1,1}(\lambda)\frac{\partial_\lambda\td{L}_{1,2}(\lambda)}{\td{L}_{1,2}(\lambda)}+\frac{(t_{\infty,0})^2}{\omega}\td{L}_{1,2}(\lambda)\cr
&&+\td{L}_{1,2}(\lambda)\sum_{s=1}^n\left[\frac{\underset{j=r_s+1}{\overset{2r_s}{\sum}}\left(\underset{m=0}{\overset{2r_s-j}{\sum}} t_{X_s,r_s-1-m}t_{X_s,j+m-r_s-1}\right) (\lambda-X_s)^{-j} -\mathring{L}_{1,1}(\lambda)^2}{\td{L}_{1,2}(\lambda)}\right]_{X_s,-}\cr
&&
\eea
so that
\bea &&g_0=\frac{1}{\omega}\left(\frac{1}{2}-t_{\infty,0}\right)\left(\sum_{s=1}^n X_s^2Q_{X_s,1}+2X_sQ_{X_s,2}+Q_{X_s,3}\right)\cr
&&-\frac{1}{2t_{\infty,0}}\underset{\lambda\to \infty}{\Res} \lambda^2\Big[
(\mathring{L}_{1,1}(\lambda))^2+\partial_\lambda \mathring{L}_{1,1}(\lambda)-\frac{t_{\infty,0}}{\omega}\td{L}_{1,2}(\lambda)-\mathring{L}_{1,1}(\lambda)\frac{\partial_\lambda\td{L}_{1,2}(\lambda)}{\td{L}_{1,2}(\lambda)}+\frac{(t_{\infty,0})^2}{\omega}\td{L}_{1,2}(\lambda)\cr
&&+\td{L}_{1,2}(\lambda)\sum_{s=1}^n\left[\frac{\underset{j=r_s+1}{\overset{2r_s}{\sum}}\left(\underset{m=0}{\overset{2r_s-j}{\sum}} t_{X_s,r_s-1-m}t_{X_s,j+m-r_s-1}\right) (\lambda-X_s)^{-j} -\mathring{L}_{1,1}(\lambda)^2}{\td{L}_{1,2}(\lambda)}\right]_{X_s,-}
\Big] \cr
&&
\eea

\subsubsection{Computation of $[\td{A}_{\boldsymbol{\alpha}}(\lambda)]_{2,1}$}
Entries $(2,1)$ and $(2,2)$ of the gauge transformation \eqref{ExpresstdA} imply that
\bea \label{GaugeTdA21A22}
\td{L}_{1,1}(\lambda)[\td{A}_{\boldsymbol{\alpha}}(\lambda)]_{1,1}+\td{L}_{1,2}(\lambda)[\td{A}_{\boldsymbol{\alpha}}(\lambda)]_{2,1}&=&[A_{\boldsymbol{\alpha}}(\lambda)]_{2,1}+\td{L}_{1,1}(\lambda)[A_{\boldsymbol{\alpha}}(\lambda)]_{2,2}-\mathcal{L}_{\boldsymbol{\alpha}}[\td{L}_{1,1}(\lambda)]\cr
\td{L}_{1,1}(\lambda)[\td{A}_{\boldsymbol{\alpha}}(\lambda)]_{1,2}-\td{L}_{1,2}(\lambda)[\td{A}_{\boldsymbol{\alpha}}(\lambda)]_{1,1}&=&[A_{\boldsymbol{\alpha}}(\lambda)]_{2,2}\td{L}_{1,2}(\lambda)-\mathcal{L}_{\boldsymbol{\alpha}}[\td{L}_{1,2}(\lambda)]
\eea

Using \eqref{A21} we get
\bea \label{tdA21} [\td{A}_{\boldsymbol{\alpha}}(\lambda)]_{2,1}&=&-\mathcal{L}_{\boldsymbol{\alpha}}\left[\frac{\td{L}_{1,1}(\lambda)}{\td{L}_{1,2}(\lambda)}\right]+\frac{\partial_\lambda[A_{\boldsymbol{\alpha}}(\lambda)]_{1,1}}{\td{L}_{1,2}(\lambda)} +[A_{\boldsymbol{\alpha}}(\lambda)]_{1,2}\frac{L_{2,1}(\lambda)}{\td{L}_{1,2}(\lambda)}\cr
&&+\frac{\td{L}_{1,1}(\lambda)}{\td{L}_{1,2}(\lambda)}\left(\frac{\td{L}_{1,1}(\lambda)}{\td{L}_{1,2}(\lambda)}[\td{A}_{\boldsymbol{\alpha}}(\lambda)]_{1,2}-2[\td{A}_{\boldsymbol{\alpha}}(\lambda)]_{1,1}\right)
\eea
Since $[\td{A}_{\boldsymbol{\alpha}}(\lambda)]_{2,1}$ is a rational function of $\lambda$ with poles in $\{\infty,X_1,\dots,X_n\}$, we only need to evaluate the contributions at each pole.

\subsubsection{Term $-\mathcal{L}_{\boldsymbol{\alpha}}\left[\frac{\td{L}_{1,1}(\lambda)}{\td{L}_{1,2}(\lambda)}\right]$}\label{Sectech}

Let us discuss the term $-\mathcal{L}_{\boldsymbol{\alpha}}\left[\frac{\td{L}_{1,1}(\lambda)}{\td{L}_{1,2}(\lambda)}\right]$. We recall that
\beq \frac{\td{L}_{1,1}(\lambda)}{\td{L}_{1,2}(\lambda)}\overset{\lambda \to X_s}{=} O(1)\,\,\, \forall\, s\in \llbracket 1,n\rrbracket\eeq
At infinity, we have:
\beq \label{RatiotdL11tdL12}\frac{\td{L}_{1,1}(\lambda)}{\td{L}_{1,2}(\lambda)}\overset{\lambda \to \infty}{=}\frac{-t_{\infty,r_\infty-1}\lambda- g_0}{\omega}+O(\lambda^{-1})\eeq 
In all cases, the contribution of $-t_{\infty,r_\infty-1}\lambda$ to $-\mathcal{L}_{\boldsymbol{\alpha}}\left[\frac{\td{L}_{1,1}(\lambda)}{\td{L}_{1,2}(\lambda)}\right]$ is null because $t_{\infty,r_\infty-1}$ is either a monodromy parameter (for $r_\infty=1$) or fixed to $1$ (for $r_\infty\geq 2$) because of Proposition \ref{PropTrivialTimes}. Thus, we end up with the fact that $-\mathcal{L}_{\boldsymbol{\alpha}}\left[\frac{\td{L}_{1,1}(\lambda)}{\td{L}_{1,2}(\lambda)}\right]$ only contributes with either
\begin{itemize}\item $-\frac{t_{\infty,r_\infty-1}}{\omega^2}\mathcal{L}_{\boldsymbol{\alpha}}[\omega] \lambda-\mathcal{L}_{\boldsymbol{\alpha}}\left[\frac{Q_{\infty,r_\infty-4}}{\omega^2}\right]$ for $r_\infty\geq 4$
\item $-\frac{t_{\infty,2}}{\omega^2}\mathcal{L}_{\boldsymbol{\alpha}}[\omega] \lambda -\underset{s=1}{\overset{n}{\sum}}\mathcal{L}_{\boldsymbol{\alpha}}\left[\frac{Q_{X_s,1}}{\omega^2}\right]$ for $r_\infty=3$
\item $-\frac{t_{\infty,1}}{\omega^2}\mathcal{L}_{\boldsymbol{\alpha}}[\omega] \lambda -\underset{s=1}{\overset{n}{\sum}} \mathcal{L}_{\boldsymbol{\alpha}}\left[\frac{Q_{X_s,2}+X_sQ_{X_s,1}}{\omega^2}\right]$ for $r_\infty=2$
\item $-\frac{t_{\infty,0}}{\omega^2}\mathcal{L}_{\boldsymbol{\alpha}}[\omega] \lambda+\mathcal{L}_{\boldsymbol{\alpha}}\left[\frac{g_0}{\omega}\right]$ for $r_\infty=1$.
\end{itemize} 
Note in particular that in all four cases, the contribution is linear in $\lambda$. We may now use \eqref{LQinftyrinftygeq4}, \eqref{LQinftyrinftyequal3}, \eqref{LQinftyrinftyequal2} and \eqref{LQinftyrinftyequal1} to compute the contribution for each case.
In order to use these formulas, we need to express $a_1$ which is the coefficient of order $\lambda^{-1}$ of $[A_{\boldsymbol{\alpha}}]_{1,1}(\lambda)$ at $\lambda\to \infty$. From \eqref{RefA12} and \eqref{GaugeTdA11}, we have
\beq A^{(\boldsymbol{\alpha})}_{1,1}(\lambda)=\td{A}^{(\boldsymbol{\alpha})}_{1,1}(\lambda)-\frac{\td{L}^{(\boldsymbol{\alpha})}_{1,1}(\lambda)}{\td{L}^{(\boldsymbol{\alpha})}_{1,2}(\lambda)} \td{A}^{(\boldsymbol{\alpha})}_{1,2}(\lambda)
\eeq
so that 
\beq a_1=-\Res_{\lambda\to \infty}\left(\td{A}^{(\boldsymbol{\alpha})}_{1,1}(\lambda)-\frac{\td{L}^{(\boldsymbol{\alpha})}_{1,1}(\lambda)}{\td{L}^{(\boldsymbol{\alpha})}_{1,2}(\lambda)} \td{A}^{(\boldsymbol{\alpha})}_{1,2}(\lambda)\right)\eeq
Thus we get that the contribution of $-\mathcal{L}_{\boldsymbol{\alpha}}\left[\frac{\td{L}_{1,1}(\lambda)}{\td{L}_{1,2}(\lambda)}\right]$ is
\begin{itemize}
\item $-\frac{t_{\infty,r_\infty-1}}{\omega^2}\mathcal{L}_{\boldsymbol{\alpha}}[\omega] \lambda+\frac{2}{\omega}\,\underset{\lambda\to \infty}{\Res}\left(\td{A}^{(\boldsymbol{\alpha})}_{1,1}(\lambda)-\frac{\td{L}^{(\boldsymbol{\alpha})}_{1,1}(\lambda)}{\td{L}^{(\boldsymbol{\alpha})}_{1,2}(\lambda)} \td{A}^{(\boldsymbol{\alpha})}_{1,2}(\lambda)\right)+\frac{1}{\omega^3}Q_{\infty,r_\infty-4}\mathcal{L}_{\boldsymbol{\alpha}}[\omega] $ for $r_\infty\geq 4$.
\item $-\frac{t_{\infty,2}}{\omega^2}\mathcal{L}_{\boldsymbol{\alpha}}[\omega] \lambda+\frac{2}{\omega}\,\underset{\lambda\to \infty}{\Res}\left(\td{A}^{(\boldsymbol{\alpha})}_{1,1}(\lambda)-\frac{\td{L}^{(\boldsymbol{\alpha})}_{1,1}(\lambda)}{\td{L}^{(\boldsymbol{\alpha})}_{1,2}(\lambda)} \td{A}^{(\boldsymbol{\alpha})}_{1,2}(\lambda)\right)+\frac{1}{\omega^3}\left(\underset{s=1}{\overset{n}{\sum}}Q_{X_s,1}\right)\mathcal{L}_{\boldsymbol{\alpha}}[\omega] $ for $r_\infty=3$.
\item \sloppy{$-\frac{t_{\infty,1}}{\omega^2}\mathcal{L}_{\boldsymbol{\alpha}}[\omega] \lambda+\frac{1}{\omega}\nu^{(\boldsymbol{\alpha})}_{\infty,0}+\frac{2}{\omega}\,\underset{\lambda\to \infty}{\Res}\left(\td{A}^{(\boldsymbol{\alpha})}_{1,1}(\lambda)-\frac{\td{L}^{(\boldsymbol{\alpha})}_{1,1}(\lambda)}{\td{L}^{(\boldsymbol{\alpha})}_{1,2}(\lambda)} \td{A}^{(\boldsymbol{\alpha})}_{1,2}(\lambda)\right)+\frac{1}{\omega^3}\left(\underset{s=1}{\overset{n}{\sum}}X_sQ_{X_s,1}+Q_{X_s,2}\right)\mathcal{L}_{\boldsymbol{\alpha}}[\omega] $ for $r_\infty=2$.}
\item $-\frac{t_{\infty,0}}{\omega^2}\mathcal{L}_{\boldsymbol{\alpha}}[\omega] \lambda+\mathcal{L}_{\boldsymbol{\alpha}}\left[\frac{g_0}{\omega}\right]$ for $r_\infty=1$.
\end{itemize}

\subsubsection{Term $\frac{\partial_\lambda[A_{\boldsymbol{\alpha}}(\lambda)]_{1,1}}{\td{L}_{1,2}(\lambda)}$}

Let us now discuss the term $\frac{\partial_\lambda[A_{\boldsymbol{\alpha}}(\lambda)]_{1,1}}{\td{L}_{1,2}(\lambda)}$ in \eqref{tdA21}. This quantity is obviously regular at $\lambda\to X_s$ for all $s\in \llbracket 1,n\rrbracket$. At infinity, we have $[A_{\boldsymbol{\alpha}}(\lambda)]_{1,1}=O(1)$ so that $\partial_{\lambda}[A_{\boldsymbol{\alpha}}(\lambda)]_{1,1}=O(\lambda^{-2})$. Since $\td{L}_{1,2}(\lambda)=\omega\lambda^{r_\infty-3}+O(\lambda^{r_\infty-4})$, we only get a contribution for $r_\infty=1$ and it is constant in $\lambda$.

\subsubsection{Term $[A_{\boldsymbol{\alpha}}(\lambda)]_{1,2}\frac{L_{2,1}(\lambda)}{\td{L}_{1,2}(\lambda)}$}

Let us now discuss the term $[A_{\boldsymbol{\alpha}}(\lambda)]_{1,2}\frac{L_{2,1}(\lambda)}{\td{L}_{1,2}(\lambda)}$ in \eqref{tdA21}. For $s\in \llbracket 1,n\rrbracket$, the behavior is given by \eqref{L21overtdL12} and \eqref{A12nu}:
\bea [A_{\boldsymbol{\alpha}}(\lambda)]_{1,2}\frac{L_{2,1}(\lambda)}{\td{L}_{1,2}(\lambda)}&\overset{\lambda\to X_s}{=}&\left(\frac{\underset{j=r_s+1}{\overset{2r_s}{\sum}}\left(\underset{m=0}{\overset{2r_s-j}{\sum}} t_{X_s,r_s-1-m}t_{X_s,j+m-r_s-1}\right) (\lambda-X_s)^{-j}}{\omega\, \delta_{r_\infty\geq 3}\lambda^{r_\infty-3}+\underset{k=0}{\overset{r_\infty-4}{\sum}} Q_{\infty,k}\lambda^k+\underset{s=1}{\overset{n}{\sum}}\underset{k=1}{\overset{r_s}{\sum}} \frac{Q_{X_s,k}}{(\lambda-X_s)^k}}+ O(1)\right)\cr
&&\left(\sum_{i=0}^{r_s-1} \nu^{(\boldsymbol{\alpha})}_{X_s,i}(\lambda-X_s)^i +O\left((\lambda-X_s)^{r_s}\right)\right)\cr
&\overset{\lambda\to X_s}{=}&\frac{\underset{i=0}{\overset{r_s-1}{\sum}} \underset{j=r_s+1}{\overset{2r_s}{\sum}}\left(\underset{m=0}{\overset{2r_s-j}{\sum}} t_{X_s,r_s-1-m}t_{X_s,j+m-r_s-1} \nu^{(\boldsymbol{\alpha})}_{X_s,i}\right) (\lambda-X_s)^{i-j}}{\omega\,\delta_{r_\infty\geq 3}\lambda^{r_\infty-3}+\underset{k=0}{\overset{r_\infty-4}{\sum}} Q_{\infty,k}\lambda^k+\underset{s=1}{\overset{n}{\sum}}\underset{k=1}{\overset{r_s}{\sum}} \frac{Q_{X_s,k}}{(\lambda-X_s)^k}} +O(1)\cr
&\overset{k=j-i}{=}&\frac{\underset{k=r_s+1}{\overset{2r_s}{\sum}}\left( \underset{j=k}{\overset{2r_s}{\sum}}\underset{m=0}{\overset{2r_s-j}{\sum}} t_{X_s,r_s-1-m}t_{X_s,j+m-r_s-1} \nu^{(\boldsymbol{\alpha})}_{X_s,j-k}\right) (\lambda-X_s)^{-k}}{\omega\,\delta_{r_\infty\geq 3}\lambda^{r_\infty-3}+\underset{k=0}{\overset{r_\infty-4}{\sum}} Q_{\infty,k}\lambda^k+\underset{s=1}{\overset{n}{\sum}}\underset{k=1}{\overset{r_s}{\sum}} \frac{Q_{X_s,k}}{(\lambda-X_s)^k}} +O(1)\cr
&&
\eea
The situation at infinity depends on the value of $r_\infty$.
\begin{itemize}\item If $r_\infty\geq 3$, we have from \eqref{L21overtdL21rinftygeq3} and \eqref{A12nu}:
\small{\bea
[A_{\boldsymbol{\alpha}}(\lambda)]_{1,2}\frac{L_{2,1}(\lambda)}{\td{L}_{1,2}(\lambda)}&\overset{\lambda\to \infty}{=}&\left(\frac{\underset{j=r_\infty-2}{\overset{2r_\infty-4}{\sum}}\left(\underset{m=0}{\overset{2r_\infty-4-j}{\sum}} t_{\infty,r_\infty-1-m}t_{\infty,j+m-r_\infty+3}\right) \lambda^{j}}{\omega\,\lambda^{r_\infty-3}+\underset{k=0}{\overset{r_\infty-4}{\sum}} Q_{\infty,k}\lambda^k+\underset{s=1}{\overset{n}{\sum}}\underset{k=1}{\overset{r_s}{\sum}} \frac{Q_{X_s,k}}{(\lambda-X_s)^k}} +O(1)\right)\cr
&&\left(\sum_{i=1}^{r_\infty-1} \frac{\nu^{(\boldsymbol{\alpha})}_{\infty,i}}{\lambda^i} +O\left(\lambda^{-r_\infty}\right) \right)\cr
&\overset{\lambda\to \infty}{=}&\frac{\underset{i=1}{\overset{r_\infty-1}{\sum}}\underset{j=r_\infty-2}{\overset{2r_\infty-4}{\sum}}\left(\underset{m=0}{\overset{2r_\infty-4-j}{\sum}} t_{\infty,r_\infty-1-m}t_{\infty,j+m-r_\infty+3}\nu^{(\boldsymbol{\alpha})}_{\infty,i}\right) \lambda^{j-i}}{\omega\,\lambda^{r_\infty-3}+\underset{k=0}{\overset{r_\infty-4}{\sum}} Q_{\infty,k}\lambda^k+\underset{s=1}{\overset{n}{\sum}}\underset{k=1}{\overset{r_s}{\sum}} \frac{Q_{X_s,k}}{(\lambda-X_s)^k}} +O(\lambda^{-1})\cr
&\overset{k=j-i}{=}&\frac{\underset{k=r_\infty-3}{\overset{2r_\infty-5}{\sum}}\left(\underset{j=k+1}{\overset{2r_\infty-4}{\sum}}\underset{m=0}{\overset{2r_\infty-4-j}{\sum}} t_{\infty,r_\infty-1-m}t_{\infty,j+m-r_\infty+3}\nu^{(\boldsymbol{\alpha})}_{\infty,j-k}\right) \lambda^{k}}{\omega\,\lambda^{r_\infty-3}+\underset{k=0}{\overset{r_\infty-4}{\sum}} Q_{\infty,k}\lambda^k+\underset{s=1}{\overset{n}{\sum}}\underset{k=1}{\overset{r_s}{\sum}} \frac{Q_{X_s,k}}{(\lambda-X_s)^k}} +O(\lambda^{-1})\cr
&&
\eea}\normalsize{}
\item If $r_\infty=2$, we have from \eqref{L21overtdL21rinftyequal2} and \eqref{A12nu}:
\footnotesize{\bea [A_{\boldsymbol{\alpha}}(\lambda)]_{1,2}\frac{L_{2,1}(\lambda)}{\td{L}_{1,2}(\lambda)}&\overset{\lambda\to \infty}{=}&
\left( \frac{(t_{\infty,1})^2}{\omega}\lambda-\frac{(t_{\infty,1})^2}{\omega^2}\left(\sum_{s=1}^n(Q_{X_s,2}+X_sQ_{X_s,1})\right)+\frac{2t_{\infty,1}t_{\infty,0}-t_{\infty,1}}{\omega}+O(\lambda^{-1})\right)\cr
&&\left(\nu^{(\boldsymbol{\alpha})}_{\infty,0}+\nu^{(\boldsymbol{\alpha})}_{\infty,1}\lambda^{-1}+ O(\lambda^{-2}) \right)\cr
&\overset{\lambda\to \infty}{=}&\frac{(t_{\infty,1})^2}{\omega}\nu^{(\boldsymbol{\alpha})}_{\infty,0}\lambda+\frac{(t_{\infty,1})^2}{\omega}\nu^{(\boldsymbol{\alpha})}_{\infty,1}\cr
&&+\nu^{(\boldsymbol{\alpha})}_{\infty,0}\left(-\frac{(t_{\infty,1})^2}{\omega^2}\left(\sum_{s=1}^n(Q_{X_s,2}+X_sQ_{X_s,1})\right)+\frac{2t_{\infty,1}t_{\infty,0}-t_{\infty,1}}{\omega}\right)+O(\lambda^{-1})\cr
&&
\eea}
\normalsize{}
\item If $r_\infty=1$, we have from the definition of $L_{2,1}(\lambda)$ :
\beq L_{2,1}(\lambda)=t_{\infty,0}(t_{\infty,0}-1)\lambda^{-2} +O(\lambda^{-3})
\eeq
and
\beq \td{L}_{1,2}(\lambda)=\omega \lambda^{-2}+\left(\sum_{s=1}^n X_s^2Q_{X_s,1}+2X_sQ_{X_s,2}+Q_{X_s,3}\right)\lambda^{-3}+O(\lambda^{-4})\eeq
so that
\beq \frac{L_{2,1}(\lambda)}{\td{L}_{1,2}(\lambda)}=\frac{t_{\infty,0}(t_{\infty,0}-1)}{\omega}+O(\lambda^{-1})
\eeq
Thus we get:
\beq\label{ParT1} [A_{\boldsymbol{\alpha}}(\lambda)]_{1,2}\frac{L_{2,1}(\lambda)}{\td{L}_{1,2}(\lambda)}\overset{\lambda\to \infty}{=}
\overset{\lambda\to \infty}{=}\frac{t_{\infty,0}(t_{\infty,0}-1)}{\omega}\nu^{(\boldsymbol{\alpha})}_{\infty,-1}\lambda +O(1)
\eeq
\end{itemize}

\subsubsection{Term $\frac{\td{L}_{1,1}(\lambda)}{\td{L}_{1,2}(\lambda)}\left(\frac{\td{L}_{1,1}(\lambda)}{\td{L}_{1,2}(\lambda)}[\td{A}_{\boldsymbol{\alpha}}(\lambda)]_{1,2}-2[\td{A}_{\boldsymbol{\alpha}}(\lambda)]_{1,1}\right)$}
Let us now study the term $\frac{\td{L}_{1,1}(\lambda)}{\td{L}_{1,2}(\lambda)}\left(\frac{\td{L}_{1,1}(\lambda)}{\td{L}_{1,2}(\lambda)}[\td{A}_{\boldsymbol{\alpha}}(\lambda)]_{1,2}-2[\td{A}_{\boldsymbol{\alpha}}(\lambda)]_{1,1}\right)$ in \eqref{tdA21} for $r_\infty=1$. As mentioned in \eqref{RatiotdL11tdL12} we have:
\beq \frac{\td{L}_{1,1}(\lambda)}{\td{L}_{1,2}(\lambda)}=\frac{-t_{\infty,0}\lambda-g_0}{\omega} +O(\lambda^{-1})\eeq
Moreover, from \eqref{A12Infinityrinftyequal1}
\beq [\td{A}_{\boldsymbol{\alpha}}(\lambda)]_{1,2}= \omega \,\nu_{\infty,-1}^{(\boldsymbol{\alpha})} \lambda^{-1}+O(\lambda^{-2})
\eeq
and we have
\beq [\td{A}_{\boldsymbol{\alpha}}(\lambda)]_{1,1}=c_{\infty,0}-t_{\infty,0}\nu_{\infty,-1}^{(\boldsymbol{\alpha})}+O(\lambda^{-1})\eeq
so that
\bea\label{ParT2} \frac{\td{L}_{1,1}(\lambda)}{\td{L}_{1,2}(\lambda)}\left(\frac{\td{L}_{1,1}(\lambda)}{\td{L}_{1,2}(\lambda)}[\td{A}_{\boldsymbol{\alpha}}(\lambda)]_{1,2}-2[\td{A}_{\boldsymbol{\alpha}}(\lambda)]_{1,1}\right)&=&
\left(\frac{2t_{\infty,0}c_{\infty,0}}{\omega}-\frac{(t_{\infty,0})^2}{\omega}\nu_{\infty,-1}^{(\boldsymbol{\alpha})}\right)\lambda+O(1)\cr
&=& \left(\frac{t_{\infty,0} \mathcal{L}_{\boldsymbol{\alpha}}[\omega]}{\omega^2} +  \frac{t_{\infty,0}(1-t_{\infty,0})}{\omega}\nu_{\infty,-1}^{(\boldsymbol{\alpha})}\right)\lambda+O(1)\cr
&& 
\eea

\subsubsection{Regrouping results}
Let us regroup all the contributions in each case.
\begin{itemize}
\item For $r_\infty\geq 4$ we have:
\bea [\td{A}_{\boldsymbol{\alpha}}(\lambda)]_{2,1}&=&-\frac{t_{\infty,r_\infty-1}}{\omega^2}\mathcal{L}_{\boldsymbol{\alpha}}[\omega] \lambda-\frac{2}{\omega}\,\underset{\lambda\to \infty}{\Res}\left(\td{A}^{(\boldsymbol{\alpha})}_{1,1}(\lambda)-\frac{\td{L}^{(\boldsymbol{\alpha})}_{1,1}(\lambda)}{\td{L}^{(\boldsymbol{\alpha})}_{1,2}(\lambda)} \td{A}^{(\boldsymbol{\alpha})}_{1,2}(\lambda)\right)+\frac{1}{\omega^3}Q_{\infty,r_\infty-4}\mathcal{L}_{\boldsymbol{\alpha}}[\omega]\cr
&&+ \left[\frac{\td{L}_{1,1}(\lambda)}{\td{L}_{1,2}(\lambda)}\left(\frac{\td{L}_{1,1}(\lambda)}{\td{L}_{1,2}(\lambda)}[\td{A}_{\boldsymbol{\alpha}}(\lambda)]_{1,2}-2[\td{A}_{\boldsymbol{\alpha}}(\lambda)]_{1,1}\right)\right]_{\infty,+}\cr
&&+\sum_{s=1}^n\left[\frac{\underset{k=r_s+1}{\overset{2r_s}{\sum}}\left( \underset{j=k}{\overset{2r_s}{\sum}}\underset{m=0}{\overset{2r_s-j}{\sum}} t_{X_s,r_s-1-m}t_{X_s,j+m-r_s-1} \nu^{(\boldsymbol{\alpha})}_{X_s,j-k}\right) (\lambda-X_s)^{-k}}{\omega \lambda^{r_\infty-3}+\underset{k=0}{\overset{r_\infty-4}{\sum}} Q_{\infty,k}\lambda^k+\underset{s=1}{\overset{n}{\sum}}\underset{k=1}{\overset{r_s}{\sum}} \frac{Q_{X_s,k}}{(\lambda-X_s)^k}}\right]_{X_s,-}\cr
&&+\left[\frac{\underset{k=r_\infty-3}{\overset{2r_\infty-5}{\sum}}\left(\underset{j=k+1}{\overset{2r_\infty-4}{\sum}}\underset{m=0}{\overset{2r_\infty-4-j}{\sum}} t_{\infty,r_\infty-1-m}t_{\infty,j+m-r_\infty+3}\nu^{(\boldsymbol{\alpha})}_{\infty,j-k}\right) \lambda^{k}}{\omega\lambda^{r_\infty-3}+\underset{k=0}{\overset{r_\infty-4}{\sum}} Q_{\infty,k}\lambda^k+\underset{s=1}{\overset{n}{\sum}}\underset{k=1}{\overset{r_s}{\sum}} \frac{Q_{X_s,k}}{(\lambda-X_s)^k}}\right]_{\infty,+}\cr
&&+\sum_{s=1}^n \left[\frac{\td{L}_{1,1}(\lambda)}{\td{L}_{1,2}(\lambda)}\left(\frac{\td{L}_{1,1}(\lambda)}{\td{L}_{1,2}(\lambda)}[\td{A}_{\boldsymbol{\alpha}}(\lambda)]_{1,2}-2[\td{A}_{\boldsymbol{\alpha}}(\lambda)]_{1,1}\right)\right]_{X_s,-}
\eea
\item For $r_\infty=3$ we have:
\bea [\td{A}_{\boldsymbol{\alpha}}(\lambda)]_{2,1}&=&-\frac{t_{\infty,2}}{\omega^2}\mathcal{L}_{\boldsymbol{\alpha}}[\omega] \lambda-\frac{2}{\omega}\,\underset{\lambda\to \infty}{\Res}\left(\td{A}^{(\boldsymbol{\alpha})}_{1,1}(\lambda)-\frac{\td{L}^{(\boldsymbol{\alpha})}_{1,1}(\lambda)}{\td{L}^{(\boldsymbol{\alpha})}_{1,2}(\lambda)} \td{A}^{(\boldsymbol{\alpha})}_{1,2}(\lambda)\right)+\frac{1}{\omega^3}\left(\underset{s=1}{\overset{n}{\sum}}Q_{X_s,1}\right)\mathcal{L}_{\boldsymbol{\alpha}}[\omega]\cr
&&+ \left[\frac{\td{L}_{1,1}(\lambda)}{\td{L}_{1,2}(\lambda)}\left(\frac{\td{L}_{1,1}(\lambda)}{\td{L}_{1,2}(\lambda)}[\td{A}_{\boldsymbol{\alpha}}(\lambda)]_{1,2}-2[\td{A}_{\boldsymbol{\alpha}}(\lambda)]_{1,1}\right)\right]_{\infty,+}\cr
&&+\sum_{s=1}^n\left[\frac{\underset{k=r_s+1}{\overset{2r_s}{\sum}}\left( \underset{j=k}{\overset{2r_s}{\sum}}\underset{m=0}{\overset{2r_s-j}{\sum}} t_{X_s,r_s-1-m}t_{X_s,j+m-r_s-1} \nu^{(\boldsymbol{\alpha})}_{X_s,j-k}\right) (\lambda-X_s)^{-k}}{\omega \lambda^{r_\infty-3}+\underset{k=0}{\overset{r_\infty-4}{\sum}} Q_{\infty,k}\lambda^k+\underset{s=1}{\overset{n}{\sum}}\underset{k=1}{\overset{r_s}{\sum}} \frac{Q_{X_s,k}}{(\lambda-X_s)^k}}\right]_{X_s,-}\cr
&&+\left[\frac{\underset{k=r_\infty-3}{\overset{2r_\infty-5}{\sum}}\left(\underset{j=k+1}{\overset{2r_\infty-4}{\sum}}\underset{m=0}{\overset{2r_\infty-4-j}{\sum}} t_{\infty,r_\infty-1-m}t_{\infty,j+m-r_\infty+3}\nu^{(\boldsymbol{\alpha})}_{\infty,j-k}\right) \lambda^{k}}{\omega\lambda^{r_\infty-3}+\underset{k=0}{\overset{r_\infty-4}{\sum}} Q_{\infty,k}\lambda^k+\underset{s=1}{\overset{n}{\sum}}\underset{k=1}{\overset{r_s}{\sum}} \frac{Q_{X_s,k}}{(\lambda-X_s)^k}}\right]_{\infty,+}\cr
&&+\sum_{s=1}^n \left[\frac{\td{L}_{1,1}(\lambda)}{\td{L}_{1,2}(\lambda)}\left(\frac{\td{L}_{1,1}(\lambda)}{\td{L}_{1,2}(\lambda)}[\td{A}_{\boldsymbol{\alpha}}(\lambda)]_{1,2}-2[\td{A}_{\boldsymbol{\alpha}}(\lambda)]_{1,1}\right)\right]_{X_s,-}
\eea
\item For $r_\infty=2$ we have:
\bea[\td{A}_{\boldsymbol{\alpha}}(\lambda)]_{2,1}&=& -\frac{t_{\infty,1}}{\omega^2}\mathcal{L}_{\boldsymbol{\alpha}}[\omega] \lambda+\frac{1}{\omega}\nu^{(\boldsymbol{\alpha})}_{\infty,0}-\frac{2}{\omega}\,\underset{\lambda\to \infty}{\Res}\left(\td{A}^{(\boldsymbol{\alpha})}_{1,1}(\lambda)-\frac{\td{L}^{(\boldsymbol{\alpha})}_{1,1}(\lambda)}{\td{L}^{(\boldsymbol{\alpha})}_{1,2}(\lambda)} \td{A}^{(\boldsymbol{\alpha})}_{1,2}(\lambda)\right)\cr
&&+\frac{1}{\omega^3}\left(\underset{s=1}{\overset{n}{\sum}}X_sQ_{X_s,1}+Q_{X_s,2}\right)\mathcal{L}_{\boldsymbol{\alpha}}[\omega] \cr 
&&+\frac{(t_{\infty,1})^2}{\omega}\nu^{(\boldsymbol{\alpha})}_{\infty,0}\lambda+\frac{(t_{\infty,1})^2}{\omega}\nu^{(\boldsymbol{\alpha})}_{\infty,1}\cr
&&+\nu^{(\boldsymbol{\alpha})}_{\infty,0}\left(-\frac{(t_{\infty,1})^2}{\omega^2}\left(\sum_{s=1}^n(Q_{X_s,2}+X_sQ_{X_s,1})\right)+\frac{2t_{\infty,1}t_{\infty,0}-t_{\infty,1}}{\omega}\right)\cr
&&+ \left[\frac{\td{L}_{1,1}(\lambda)}{\td{L}_{1,2}(\lambda)}\left(\frac{\td{L}_{1,1}(\lambda)}{\td{L}_{1,2}(\lambda)}[\td{A}_{\boldsymbol{\alpha}}(\lambda)]_{1,2}-2[\td{A}_{\boldsymbol{\alpha}}(\lambda)]_{1,1}\right)\right]_{\infty,+}\cr
&&+\sum_{s=1}^n \left[\frac{\underset{k=r_s+1}{\overset{2r_s}{\sum}}\left( \underset{j=k}{\overset{2r_s}{\sum}}\underset{m=0}{\overset{2r_s-j}{\sum}} t_{X_s,r_s-1-m}t_{X_s,j+m-r_s-1} \nu^{(\boldsymbol{\alpha})}_{X_s,j-k}\right) (\lambda-X_s)^{-k}}{\underset{s=1}{\overset{n}{\sum}}\underset{k=1}{\overset{r_s}{\sum}} \frac{Q_{X_s,k}}{(\lambda-X_s)^k}}\right]_{X_s,-}\cr
&&+\sum_{s=1}^n \left[\frac{\td{L}_{1,1}(\lambda)}{\td{L}_{1,2}(\lambda)}\left(\frac{\td{L}_{1,1}(\lambda)}{\td{L}_{1,2}(\lambda)}[\td{A}_{\boldsymbol{\alpha}}(\lambda)]_{1,2}-2[\td{A}_{\boldsymbol{\alpha}}(\lambda)]_{1,1}\right)\right]_{X_s,-}
\eea
Using previous results, we may evaluate the contribution at infinity and observe that it is null.

\item For $r_\infty=1$, we get from Section \ref{Sectech}, \eqref{ParT1} and \eqref{ParT2} that $[\td{A}_{\boldsymbol{\alpha}}(\lambda)]_{2,1}$ behaves at infinity like
\small{\bea [\td{A}_{\boldsymbol{\alpha}}(\lambda)]_{2,1}&=&\left[-\frac{t_{\infty,0}}{\omega^2}\mathcal{L}_{\boldsymbol{\alpha}}[\omega]  + \frac{(t_{\infty,0})^2-t_{\infty,0}}{\omega}\nu^{(\boldsymbol{\alpha})}_{\infty,-1}+\frac{t_{\infty,0} \mathcal{L}_{\boldsymbol{\alpha}}[\omega]}{\omega^2} +  \frac{t_{\infty,0}(1-t_{\infty,0})}{\omega}\nu_{\infty,-1}^{(\boldsymbol{\alpha})} \right]\lambda\cr
&&+O(1)\cr
&=&
O(1)
\eea}
\normalsize{We} still need to compute the constant term. In order to obtain it, we write entry $(2,1)$ of the isomonodromic compatibility equation:
\beq\label{Compat21} \mathcal{L}_{\boldsymbol{\alpha}}[\td{L}_{2,1}(\lambda)]=\partial_\lambda [\td{A}_{\boldsymbol{\alpha}}(\lambda)]_{2,1}-2\left(\td{L}_{2,1}(\lambda)[\td{A}_{\boldsymbol{\alpha}}(\lambda)]_{1,1}-\td{L}_{1,1}(\lambda)[\td{A}_{\boldsymbol{\alpha}}(\lambda)]_{2,1}\right)\eeq
The normalization of $\td{L}(\lambda)$ at infinity implies that $\td{L}_{2,1}(\lambda)=O(\lambda^{-2})$. Since it is a derivative, the term $\partial_\lambda [\td{A}_{\boldsymbol{\alpha}}(\lambda)]_{2,1}$ does not provide any term or order $\lambda^{-1}$ at infinity. In the r.h.s. $\td{L}_{2,1}(\lambda)=O(\lambda^{-2})$ while $[\td{A}_{\boldsymbol{\alpha}}(\lambda)]_{1,1}=O(1)$ so that $\td{L}_{2,1}(\lambda)[\td{A}_{\boldsymbol{\alpha}}(\lambda)]_{1,1}=O(\lambda^{-2})$. Let us denote $A_0$ the constant term at infinity of $[\td{A}_{\boldsymbol{\alpha}}(\lambda)]_{2,1}$, i.e.
\beq [\td{A}_{\boldsymbol{\alpha}}(\lambda)]_{2,1}=\frac{2t_{\infty,0}}{\omega^2}\mathcal{L}_{\boldsymbol{\alpha}}[\omega]\lambda+A_0+O(\lambda^{-1})\eeq
We have $\td{L}_{1,1}(\lambda)=-t_{\infty,0}\lambda^{-1}+\beta_{-1}\lambda^{-2}+O(\lambda^{-3})$ so that order $\lambda^{-1}$ at infinity of \eqref{Compat21} provides
\beq 0=\frac{2t_{\infty,0}}{\omega^2}\mathcal{L}_{\boldsymbol{\alpha}}[\omega]\beta_{-1}-t_{\infty,0} A_0\,\,\Rightarrow\,\, A_0=\frac{2}{\omega^2}\mathcal{L}_{\boldsymbol{\alpha}}[\omega]\beta_{-1}\eeq
Finally, following Lemma \ref{LemmaQExpression} we have
\beq \mathring{L}_{1,1}(\lambda)=\td{L}_{1,1}(\lambda)+\frac{(t_{\infty,0}\lambda+g_0)}{\omega}\td{L}_{1,2}(\lambda)= O(\lambda^{-3})\eeq
It implies that
\beq \label{betaMinus1}\beta_{-1}=-g_0-\frac{t_{\infty,0}}{\omega}\left(\sum_{s=1}^n X_s^2Q_{X_s,1}+2X_sQ_{X_s,2}+Q_{X_s,3}\right)\eeq
so that we finally get
\beq A_0=-\frac{2}{\omega^2}\mathcal{L}_{\boldsymbol{\alpha}}[\omega]\left(g_0+\frac{t_{\infty,0}}{\omega}\left(\sum_{s=1}^n X_s^2Q_{X_s,1}+2X_sQ_{X_s,2}+Q_{X_s,3}\right)\right)\eeq
Thus we end up with:
\bea[\td{A}_{\boldsymbol{\alpha}}(\lambda)]_{2,1}&=&\frac{2t_{\infty,0}}{\omega^2}\mathcal{L}_{\boldsymbol{\alpha}}[\omega]\lambda-\frac{2}{\omega^2}\mathcal{L}_{\boldsymbol{\alpha}}[\omega]\left(g_0+\frac{t_{\infty,0}}{\omega}\left(\sum_{s=1}^n X_s^2Q_{X_s,1}+2X_sQ_{X_s,2}+Q_{X_s,3}\right)\right)\cr
&&+\left[\frac{\underset{k=r_s+1}{\overset{2r_s}{\sum}}\left( \underset{j=k}{\overset{2r_s}{\sum}}\underset{m=0}{\overset{2r_s-j}{\sum}} t_{X_s,r_s-1-m}t_{X_s,j+m-r_s-1} \nu^{(\boldsymbol{\alpha})}_{X_s,j-k}\right) (\lambda-X_s)^{-k}}{\underset{s=1}{\overset{n}{\sum}}\underset{k=1}{\overset{r_s}{\sum}} \frac{Q_{X_s,k}}{(\lambda-X_s)^k}}\right]_{X_s,-}\cr
&&+\sum_{s=1}^n \left[\frac{\td{L}_{1,1}(\lambda)}{\td{L}_{1,2}(\lambda)}\left(\frac{\td{L}_{1,1}(\lambda)}{\td{L}_{1,2}(\lambda)}[\td{A}_{\boldsymbol{\alpha}}(\lambda)]_{1,2}-2[\td{A}_{\boldsymbol{\alpha}}(\lambda)]_{1,1}\right)\right]_{X_s,-}
\eea
\end{itemize}

\section{Proof of Lemma \ref{LemmatdL11}}\label{AppendixLemmatdL11}
From Theorem \ref{GeoLaxMatrices}, we introduce
\bea R_{X_s,k}&=&\sum_{m=1}^{r_s+1-k} P_{X_s,m}Q_{X_s,k+m-1}-\frac{(g_0+t_{\infty,r_\infty-1}X_s)}{\omega}Q_{X_s,k}\cr
&&-\frac{t_{\infty,r_\infty-1}}{\omega}Q_{X_s,k+1}\delta_{k\leq r_s-1}\,\,,\,\, \forall \, (s,k)\in \llbracket 1,n\rrbracket\times\llbracket 1,r_s\rrbracket\cr
R_{\infty,r_\infty-4}&=&-\omega\, P_{\infty,0}-\frac{g_0}{\omega}-\frac{t_{\infty,r_\infty-1}}{\omega}Q_{\infty,r_\infty-4} \,\, \text{ if } \, \, r_\infty\geq 4\cr
R_{\infty,k}&=&-\omega\, P_{\infty,r_\infty-4-k}-\sum_{m=0}^{r_\infty-5-k}P_{\infty,m}Q_{\infty,k+1+m}-\frac{t_{\infty,r_\infty-1}}{\omega}Q_{\infty,k-1}\cr
&&-\frac{g_0}{\omega}Q_{\infty,k} \,\,,\,\, \forall \, k\in \llbracket 1, r_\infty-5\rrbracket\cr
R_{\infty,0}&=&-\omega\,P_{\infty,r_\infty-4}-\sum_{m=0}^{r_\infty-5}P_{\infty,m}Q_{\infty,m+1}-\frac{g_0}{\omega}Q_{\infty,0}-\frac{t_{\infty,r_\infty-1}}{\omega}\sum_{s=1}^nQ_{X_s,1}
\eea
so that
\beq \td{L}_{1,1}(\lambda)\overset{\lambda \to X_s}{=}\sum_{k=1}^{r_s}\frac{R_{X_s,k}}{(\lambda-X_s)^k}+ O(1)\,\,,\,\, \forall \, s\in \llbracket 1,n\rrbracket\eeq
and for $r_\infty\geq 3$:
\beq \td{L}_{1,1}(\lambda)\overset{\lambda \to \infty}{=}-t_{\infty,r_\infty-1}\lambda^{r_\infty-2}-t_{\infty,r_\infty-2}\lambda^{r_\infty-3}+  \sum_{k=0}^{r_\infty-4}R_{\infty,k}\lambda^k +O(\lambda^{-1})
\eeq
Since $\td{L}_{1,1}(\lambda)$ is a rational function of $\lambda$ with poles only in $\{\infty,X_1,\dots,X_n\}$, we obtain  Lemma \ref{LemmatdL11}. The additional constraints for $r_\infty\leq 2$ follow from the first two leading orders at infinity of $\td{L}_{1,1}(\lambda)$.

\section{Proof of Theorem \ref{GeoLaxMatricesQR}}\label{AppendixRtdA}
We have from \eqref{GaugeTdA11Reduced}
\beq [\td{A}_{\boldsymbol{\alpha}}(\lambda)]_{1,1}=c_{\infty,0}+\td{L}_{1,1}(\lambda)[A_{\boldsymbol{\alpha}}(\lambda)]_{1,2}\eeq
Contributions at infinity for $r_\infty\leq 2$ are similar to \eqref{tdA11rinftyequal1} using expression of $c_{\infty,0}$ of Section \ref{Appendixcinfty0}. For any $s\in \llbracket 1,n\rrbracket$, we have
\bea \td{L}_{1,1}(\lambda)[A_{\boldsymbol{\alpha}}(\lambda)]_{1,2}&=&\left(\sum_{i=0}^{r_s-1}\nu_{X_s,i}^{(\boldsymbol{\alpha})} (\lambda-X_s)^{i} +O\left((\lambda-X_s)^{r_s}\right)\right)\left(\sum_{k=1}^{r_s}\frac{R_{X_s,k}}{(\lambda-X_s)^k}+O(1)\right)\cr
&=&\sum_{i=0}^{r_s-1}\sum_{k=1}^{r_s}\nu_{X_s,i}^{(\boldsymbol{\alpha})}R_{X_s,k}(\lambda-X_s)^{i-k} +O(1)\cr
&\overset{j=k-i}{=}&\sum_{j=1}^{r_s}\left(\sum_{i=0}^{r_s-j}  \nu_{X_s,i}^{(\boldsymbol{\alpha})}R_{X_s,i+j} \right) (\lambda-X_s)^{-j}+O(1)
\eea
At infinity, for $r_\infty\geq 3$, we have
\bea \td{L}_{1,1}(\lambda)[A_{\boldsymbol{\alpha}}(\lambda)]_{1,2}&=&\left(\sum_{i=1}^{r_\infty-2}\nu_{\infty,i}^{(\boldsymbol{\alpha})} \lambda^{-i} +O\left(\lambda^{-r_\infty+1}\right)\right)\cr
&&\left(-t_{\infty,r_\infty-1}\lambda^{r_\infty-2}-t_{\infty,r_\infty-2}\lambda^{r_\infty-3} + \sum_{k=0}^{r_\infty-4}R_{\infty,k}\lambda^k +O(\lambda^{-1})\right)\cr
&=&-t_{\infty,r_\infty-1}\sum_{i=1}^{r_\infty-2}\nu_{\infty,i}^{(\boldsymbol{\alpha})}\lambda^{r_\infty-2-i} -t_{\infty,r_\infty-2}\sum_{i=1}^{r_\infty-3}\nu_{\infty,i}^{(\boldsymbol{\alpha})}\lambda^{r_\infty-3-i}\cr
&&+\sum_{i=1}^{r_\infty-2}\sum_{k=0}^{r_\infty-4}\nu_{\infty,i}^{(\boldsymbol{\alpha})}R_{\infty,k}\lambda^{k-i} +O(\lambda^{-1})\cr
&\overset{j=k-i}{=}&-t_{\infty,r_\infty-1}\sum_{j=0}^{r_\infty-3}\nu_{\infty,r_\infty-2-j}^{(\boldsymbol{\alpha})}\lambda^{j}-t_{\infty,r_\infty-2}\sum_{j=0}^{r_\infty-4}\nu_{\infty,r_\infty-3-j}^{(\boldsymbol{\alpha})}\lambda^{j}\cr
&&+\sum_{j=0}^{r_\infty-5}\left(\sum_{i=1}^{r_\infty-4-j}\nu_{\infty,i}^{(\boldsymbol{\alpha})}R_{\infty,j+i}\right)\lambda^j+O(\lambda^{-1})
\eea
Thus, we end up with formulas of Theorem \ref{GeoLaxMatricesQR}.

\section{Proof of Theorem \ref{ExplicitDependenceQ}}\label{AppendixExplicitDependenceQ}
\subsection{Study at a finite pole}\label{AppendixSubA}
From Proposition \ref{PropQCondition}, the condition to satisfy is
\beq \label{CondXsQ}\delta^{(\boldsymbol{\alpha})}_{\mathbf{t}}[Q_{X_s,k}]=- (k-1)\sum_{i=k}^{r_s} \nu_{X_s,i+1-k}^{(\boldsymbol{\alpha})}  Q_{X_s,i}\,\,,\,\,\forall (s,k)\in \llbracket 1,n\rrbracket\times\llbracket 2,r_s\rrbracket\eeq
Conditions \eqref{CondXsQ} are equivalent to
\beq-\begin{pmatrix} \frac{\delta^{(\boldsymbol{\alpha})}_{\mathbf{t}}[Q_{X_s,r_s}]}{r_s-1}\\ \vdots \\ \frac{\delta^{(\boldsymbol{\alpha})}_{\mathbf{t}}[Q_{X_s,2}]}{1}\end{pmatrix}= \begin{pmatrix}Q_{X_s,r_s}&0&\dots& &\dots &0\\
Q_{X_s,r_s-1}&Q_{X_s,r_s}& 0& & &\vdots\\
\vdots & \ddots&\ddots &\ddots  & &\vdots\\
\vdots &\ddots&\ddots&\ddots&0&\vdots\\
Q_{X_s,3}&\ddots &\ddots&\ddots& Q_{X_s,r_s}&0\\
Q_{X_s,2}&Q_{X_s,3}& \dots & & Q_{X_s,r_s-1}& Q_{X_s,r_s}
 \end{pmatrix}\begin{pmatrix} \nu^{(\boldsymbol{\alpha})}_{{X_s},1}\\  \vdots \\ \nu^{(\boldsymbol{\alpha})}_{{X_s},r_s-1}\end{pmatrix}
\eeq
i.e. using Proposition \ref{nus}:
\bea\label{EQQ}\begin{pmatrix} \frac{\delta^{(\boldsymbol{\alpha})}_{\mathbf{t}}[Q_{X_s,r_s}]}{r_s-1}\\ \vdots \\ \frac{\delta^{(\boldsymbol{\alpha})}_{\mathbf{t}}[Q_{X_s,2}]}{1}\end{pmatrix}&=& \begin{pmatrix}Q_{X_s,r_s}&0&\dots& &\dots &0\\
Q_{X_s,r_s-1}&Q_{X_s,r_s}& 0& & &\vdots\\
\vdots & \ddots&\ddots &\ddots  & &\vdots\\
\vdots &\ddots&\ddots&\ddots&0&\vdots\\
Q_{X_s,3}&\ddots &\ddots&\ddots& Q_{X_s,r_s}&0\\
Q_{X_s,2}&Q_{X_s,3}& \dots & & Q_{X_s,r_s-1}& Q_{X_s,r_s}
 \end{pmatrix}( M_{X_s})^{-1}\begin{pmatrix} \frac{\alpha_{X_s,r_s-1}}{r_s-1}\\ \vdots \\ \frac{\alpha_{X_s,1}}{1}\end{pmatrix}\cr
&&
\eea
We now observe that the \textbf{lower triangular Toeplitz matrices commute} so that conditions \eqref{CondXsQ} are equivalent to
\beq M_{X_s} \begin{pmatrix} \frac{\delta^{(\boldsymbol{\alpha})}_{\mathbf{t}}[Q_{X_s,r_s}]}{r_s-1}\\ \vdots \\ \frac{\delta^{(\boldsymbol{\alpha})}_{\mathbf{t}}[Q_{X_s,2}]}{1}\end{pmatrix}=\begin{pmatrix}Q_{X_s,r_s}&0&\dots& &\dots &0\\
Q_{X_s,r_s-1}&Q_{X_s,r_s}& 0& & &\vdots\\
\vdots & \ddots&\ddots &\ddots  & &\vdots\\
\vdots &\ddots&\ddots&\ddots&0&\vdots\\
Q_{X_s,3}&\ddots &\ddots&\ddots& Q_{X_s,r_s}&0\\
Q_{X_s,2}&Q_{X_s,3}& \dots & & Q_{X_s,r_s-1}& Q_{X_s,r_s}
 \end{pmatrix}\begin{pmatrix} \frac{\alpha_{X_s,r_s-1}}{r_s-1}\\ \vdots \\ \frac{\alpha_{X_s,1}}{1}\end{pmatrix} 
\eeq
We may now specify the last equation for any deformation $t_{X_s,i}$ with $i\in \llbracket 1, r_s-1\rrbracket$. It gives:
\footnotesize{\bea\label{SystemInftyQ} &&\begin{pmatrix}t_{X_s,r_s-1}&0&\dots&0\\
t_{X_s,r_s-2}&t_{X_s,r_s-1}& &0\\
\vdots& \ddots &\ddots &\vdots \\
t_{X_s,1}&\dots&t_{X_s,r_s-2}&t_{X_s,r_s-1}
\end{pmatrix}
\begin{pmatrix}\frac{1}{r_s-1}&0&\dots&0\\
0&\ddots& &0\\
\vdots& & \ddots&0\\
0&\dots&0 &\frac{1}{1} \end{pmatrix} \begin{pmatrix} \delta_{t_{X_s,r_s-1}}[Q_{X_s,r_s}]&\dots& \delta_{t_{X_s,1}}[Q_{X_s,r_s}]\\
\vdots&\vdots&\vdots\\
\delta_{t_{X_s,r_s-1}}[Q_{X_s,2}]&\dots& \delta_{t_{X_s,1}}[Q_{X_s,2}]
\end{pmatrix}\cr
&&=\begin{pmatrix}Q_{X_s,r_s}&0&\dots& &\dots &0\\
Q_{X_s,r_s-1}&Q_{X_s,r_s}& 0& & &\vdots\\
\vdots & \ddots&\ddots &\ddots  & &\vdots\\
\vdots &\ddots&\ddots&\ddots&0&\vdots\\
Q_{X_s,3}&\ddots &\ddots&\ddots& Q_{X_s,r_s}&0\\
Q_{X_s,2}&Q_{X_s,3}& \dots & & Q_{X_s,r_s-1}& Q_{X_s,r_s}
 \end{pmatrix}\begin{pmatrix}\frac{1}{r_s-1}&0&\dots&0\\
0&\ddots& &0\\
\vdots& & \ddots&0\\
0&\dots&0 &\frac{1}{1} \end{pmatrix}
\cr
&&
\eea}
\normalsize{We first observe} that a solution is trivially given by $Q_{X_s,k}=t_{X_s,k-1}u_{X_s,r_s}$ for all $k\in \llbracket 2,r_s\rrbracket$ where $u_{X_s,r_s-1}$ is independent of the irregular times. It is also obvious from the lower triangular form that for any $k\in \llbracket 2,r_s\rrbracket$, $Q_{X_s,k}$ may only depend on $t_{X_s,r_s-1},\dots,t_{X_s,k-1}$. The first line of \eqref{SystemInftyQ} provides
\beq t_{X_s,r_s-1}\delta_{t_{X_s,k}}\left[\frac{Q_{X_s,r_s}}{r_s-1}\right]=\delta_{k,r_s-1}\frac{Q_{X_s,r_s}}{r_s-1}\eeq
i.e.
\beq Q_{X_s,r_s}=t_{X_s,r_s-1}u_{X_s,r_s}\eeq
with $u_{X_s,r_s}$ independent of the irregular times. The second line of \eqref{SystemInftyQ} provides
\bea t_{X_s,r_s-2}\delta_{t_{X_s,r_s-1}}\left[\frac{Q_{X_s,r_s}}{r_s-1}\right]+t_{X_s,r_s-1}\delta_{t_{X_s,r_s-1}}\left[\frac{Q_{X_s,r_s-1}}{r_s-2}\right]&=&\frac{Q_{X_s,r_s-1}}{r_s-1}\cr
t_{X_s,r_s-1}\delta_{t_{X_s,r_s-2}}\left[\frac{Q_{X_s,r_s-1}}{r_s-2}\right]&=&\frac{Q_{X_s,r_s}}{r_s-2}\cr
\delta_{t_{X_s,r_s-k}}\left[Q_{X_s,r_s-1}\right]&=&0\,\,,\,\forall\, k\in \llbracket 3,r_s-1\rrbracket
\eea
i.e.
\beq Q_{X_s,r_s-1}=t_{X_s,r_s-2}u_{X_s,r_s}+ (t_{X_s,r_s-1})^{\frac{r_s-2}{r_s-1}}u_{X_s,r_s-1}\eeq
with $u_{X_s,r_s-1}$ independent of the irregular times. The third line of \eqref{SystemInftyQ} provides
\small{\bea\frac{Q_{X_s,r_s-2}}{r_s-1}&=& t_{X_s,r_s-3}\delta_{t_{X_s,r_s-1}}\left[\frac{Q_{X_s,r_s}}{r_s-1}\right]+ t_{X_s,r_s-2}\delta_{t_{X_s,r_s-1}}\left[\frac{Q_{X_s,r_s-1}}{r_s-2}\right]+t_{X_s,r_s-1}\delta_{t_{X_s,r_s-1}}\left[\frac{Q_{X_s,r_s-2}}{r_s-3}\right]\cr
 \frac{Q_{X_s,r_s-1}}{r_s-2}&=&t_{X_s,r_s-2}\delta_{t_{X_s,r_s-2}}\left[\frac{Q_{X_s,r_s-1}}{r_s-2}\right]+t_{X_s,r_s-1}\delta_{t_{X_s,r_s-2}}\left[\frac{Q_{X_s,r_s-2}}{r_s-3}\right]\cr
\frac{Q_{X_s,r_s}}{r_s-3}&=&t_{X_s,r_s-1}\delta_{t_{X_s,r_s-3}}\left[\frac{Q_{X_s,r_s-2}}{r_s-3}\right]\cr
0&=&\delta_{t_{X_s,r_s-k}}\left[Q_{X_s,r_s-1}\right]\,\,,\,\,\forall\, k\in \llbracket 4,r_s-1\rrbracket
\eea}
\normalsize{i.e.}
\beq Q_{X_s,r_s-2}=t_{\infty,r_\infty-3}u_{X_s,r_s}+\frac{r_s-3}{r_s-2}(t_{X_s,r_s-1})^{-\frac{1}{r_s-1}}t_{X_s,r_s-2}u_{X_s,r_s-1}+(t_{X_s,r_s-1})^{\frac{r_s-3}{r_s-1}}u_{X_s,r_s-2}\eeq
with $u_{X_s,r_s-2}$ independent of the irregular times. One can then proceed by induction to obtain the next lines. In particular if we take $j_0\in \llbracket 1, r_s-1\rrbracket$ and look for solutions of the form $\mathbf{f}_{j_0}(,t_{X_s,j_0},\dots,t_{X_s,r_s-1})=(\mathbf{0}_{j-1},[\mathbf{f}_{j_0}]_{j_0+1}(t_{X_s,r_s-1}),\dots,[\mathbf{f}_{j_0}]_{j_0+1}(t_{X_s,j_0},\dots,t_{X_s,r_s-1}))^t$. The differential system \eqref{SystemInftyQ} reduces to the lower-left corner of size $(r_s-j_0)\times (r_s-j_0)$ because all other entries are trivial. Entry $(j_0,1)$ provides a simple ODE for $[\mathbf{f}_{X_s,j_0}]_{j_0}(t_{X_s,r_s-1})$ whose solution is provided in Theorem \ref{ExplicitDependenceQ}. Then looking at the line $(j_0+1)$ whose only non-trivial entries are $(j_0+1,1)$ and $(j_0+1,2)$, we obtain the explicit dependence of $[\mathbf{f}_{X_s,j_0}]_{j_0+1}(t_{X_s,r_s-1},t_{X_s,r_s-2})$ relatively to $t_{X_s,r_s-1}$ ($(j_0+1,2)$ entry) and $t_{X_s,r_s-2}$ ($(j_0+1,1)$ entry). One then proceed with line $(j_0+2)$ (whose only first three columns are non-trivial) to determine $[\mathbf{f}_{X_s,j_0}]_{j_0+2}(t_{X_s,r_s-1},t_{X_s,r_s-2},t_{X_s,r_s-3})$ and a simple induction gives all coefficients of the vector $\mathbf{f}_{X_s,j_0}$. Thus, we end up with
\footnotesize{\beq \begin{pmatrix}Q_{X_s,r_s}\\ Q_{X_s,r_s-1}\\ \vdots\\\vdots\\ Q_{X_s,2}\end{pmatrix}=
\begin{pmatrix} f^{(X_s)}_{1,1}(t_{X_s,r_s-1})&0&\dots&\dots& 0\\
 f^{(X_s)}_{2,1}(t_{X_s,r_s-2},t_{X_s,r_s-1})& f^{(X_s)}_{2,2}(t_{X_s,r_s-1})&0&&\vdots\\
\vdots&& \ddots& \ddots &\vdots\\
f^{(X_s)}_{r_s-2,1}(t_{X_s,2},\dots,t_{X_s,r_s-1})&\dots&\dots& f^{(X_s)}_{r_s-2,r_s-1}(t_{X_s,r_s-1})&0\\
f^{(X_s)}_{r_s-1,1}(t_{X_s,1},\dots,t_{X_s,r_s-1})&\dots&&\dots& f^{(X_s)}_{r_s-1,r_s-1}(t_{X_s,r_s-1})
\end{pmatrix}
\begin{pmatrix}u_{X_s,r_s}\\ u_{X_s,r_s-1}\\\vdots\\ \vdots\\ u_{X_s,2}
\end{pmatrix} 
\eeq}
\normalsize{with} all $\left(u_{X_s,k}\right)_{1\leq k\leq r_s}$ independent of the irregular times. Moreover, we have:
\bea f^{(X_s)}_{j,j}(t_{X_s,r_s-1})&=&(t_{X_s,r_s-1})^{\frac{r_s-j}{r_s-1}}\,\,,\,\, \forall\, j\in \llbracket 1,r_s-1\rrbracket\cr
f^{(X_s)}_{j+1,j}(t_{X_s,r_s-1})&=&\frac{r_s-j-1}{r_s-2}(t_{X_s,r_s-1})^{\frac{1-j}{r_s-1}}t_{X_s,r_s-2}\,\,,\,\, \forall\, j\in \llbracket 1,r_s-2\rrbracket\cr
f^{(X_s)}_{j,1}(t_{X_s,r_s-1})&=&t_{X_s,r_s-j}\,\,,\,\, \forall\, j\in \llbracket 1,r_s-1\rrbracket
\eea

\medskip
Let us now discuss the special cases of $r_\infty=2$ and $r_\infty=1$.
\begin{itemize}\item For $r_\infty=2$, we need to impose the extra condition $\underset{s=1}{\overset{n}{\sum}} Q_{X_s,1}=\omega$ from Definition \ref{DefGeometricCoordinates}. This is equivalent to $\underset{s=1}{\overset{n}{\sum}} u_{X_s,1}=\omega$ which is coherent since $\delta^{(\boldsymbol{\alpha})}_{\mathbf{t}}[\omega]=0$ so that $\omega$ does not depend on the trivial times.
\item For $r_\infty=1$, we need to impose the extra conditions $\underset{s=1}{\overset{n}{\sum}} Q_{X_s,1}=0$ and $\underset{s=1}{\overset{n}{\sum}} X_s Q_{X_s,1}+Q_{X_s,2}=\omega$ from Definition \ref{DefGeometricCoordinates}. The first condition is equivalent to $\underset{s=1}{\overset{n}{\sum}} u_{X_s,1}=0$ which is obviously consistent. On the contrary, the second one does not look consistent a priori because $Q_{X_s,2}$ may depend on the irregular times in a very non-trivial way. However, for $r_\infty=1$, the additional constraints \eqref{ExtraConditionsrinftyequal1} provide
\beq\omega\, \nu_{\infty,-1}^{(\boldsymbol{\alpha})}=\sum_{s=1}^n\left(\sum_{k=1}^{r_s} \nu_{X_s,k-1}^{(\boldsymbol{\alpha})}  Q_{X_s,k}\right)\eeq
with $\delta^{(\boldsymbol{\alpha})}_{\mathbf{t}}[\omega]=-\omega\,  \nu_{\infty,-1}^{(\boldsymbol{\alpha})}$.
Inserting this relation into the isospectral condition \eqref{CondXsQ} implies that 
\beq \delta^{(\boldsymbol{\alpha})}_{\mathbf{t}}\left[\sum_{s=1}^nQ_{X_s,2}\right]=-\sum_{s=1}^n\left(\sum_{k=2}^{r_s} \nu_{X_s,k-1}^{(\boldsymbol{\alpha})}  Q_{X_s,k}\right)\overset{\eqref{ExtraConditionsrinftyequal1}}{=}-\omega\, \nu_{\infty,-1}^{(\boldsymbol{\alpha})}+\sum_{s=1}^n\nu_{X_s,0}^{(\boldsymbol{\alpha})}  Q_{X_s,1}
\eeq
so that since $\nu_{X_s,0}^{(\boldsymbol{\alpha})}=-\alpha_{X_s}$ and $\delta^{(\boldsymbol{\alpha})}_{\mathbf{t}}Q_{X_s,1}=0$:
\beq \delta^{(\boldsymbol{\alpha})}_{\mathbf{t}}\left[\sum_{s=1}^nQ_{X_s,2}+X_s Q_{X_s,1}\right]= -\omega\, \nu_{\infty,-1}^{(\boldsymbol{\alpha})} \eeq
Extra conditions \eqref{AddConstrains2} in Definition \ref{DefGeometricCoordinates} imply that the l.h.s. is $\delta^{(\boldsymbol{\alpha})}_{\mathbf{t}}[\omega]$ which is coherent with the fact that $\delta^{(\boldsymbol{\alpha})}_{\mathbf{t}}[\omega]=-\omega\,  \nu_{\infty,-1}^{(\boldsymbol{\alpha})}$.
\end{itemize}

\subsection{Study at infinity}\label{SectionInfinity}
The situation at infinity is a little more complicated due to the normalization of the Lax matrices. Let us take $r_\infty\geq 5$ which is the only non-trivial case. From Proposition \ref{PropQCondition}, the condition to satisfy implies $Q_{\infty,r_\infty-4}=\omega\,u_{\infty,r_\infty-4}$ and $\omega$ independent of the irregular times. Moreover, we have:
\bea\delta^{(\boldsymbol{\alpha})}_{\mathbf{t}}[Q_{\infty,r_\infty-5}]&=&(r_\infty-4)\omega \,\nu_{\infty,1}^{(\boldsymbol{\alpha})}\cr
\delta^{(\boldsymbol{\alpha})}_{\mathbf{t}}[Q_{\infty,m}]&=&(m+1)\left(\omega\, \nu_{\infty,r_\infty-4-m}^{(\boldsymbol{\alpha})}+\sum_{k=m+2}^{r_\infty-4}\nu_{\infty,k-m-1}^{(\boldsymbol{\alpha})}Q_{\infty,k}\right) \,\,,\,\, \forall\, m\in \llbracket 0, r_\infty-6\rrbracket\cr
&&
\eea
This can be rewritten into a $(r_\infty-4)\times(r_\infty-4)$ linear system:
\beq \begin{pmatrix} \delta^{(\boldsymbol{\alpha})}_{\mathbf{t}}\left[\frac{Q_{\infty,r_\infty-5}}{r_\infty-4}\right]\\
\vdots\\
\delta^{(\boldsymbol{\alpha})}_{\mathbf{t}}\left[\frac{Q_{\infty,0}}{1}\right]\end{pmatrix}=\begin{pmatrix}\omega&0&\dots&0\\
Q_{\infty,r_\infty-4}&\omega& &0\\
\vdots& \ddots &\ddots &\vdots \\
Q_{\infty,2}&\dots&Q_{\infty,r_\infty-4}&\omega
\end{pmatrix}\begin{pmatrix}\nu_{\infty,1}^{(\boldsymbol{\alpha})}\\ \vdots\\ \nu_{\infty,r_\infty-4}^{(\boldsymbol{\alpha})}\end{pmatrix}
\eeq
Using the first $(r_\infty-4)$ lines of Proposition \ref{nus}, we get
\footnotesize{\beq \begin{pmatrix} \delta^{(\boldsymbol{\alpha})}_{\mathbf{t}}\left[\frac{Q_{\infty,r_\infty-5}}{r_\infty-4}\right]\\
\vdots\\
\delta^{(\boldsymbol{\alpha})}_{\mathbf{t}}\left[\frac{Q_{\infty,0}}{1}\right]\end{pmatrix}=
\begin{pmatrix}\omega&0&\dots&0\\
Q_{\infty,r_\infty-4}&\omega& &0\\
\vdots& \ddots &\ddots &\vdots \\
Q_{\infty,2}&\dots&Q_{\infty,r_\infty-4}&\omega
\end{pmatrix}\begin{pmatrix}t_{\infty,r_\infty-1}&0&\dots&0\\
t_{\infty,r_\infty-2}&t_{\infty,r_\infty-1}& &0\\
\vdots& \ddots &\ddots &\vdots \\
t_{\infty,4}&\dots&t_{\infty,r_\infty-2}&t_{\infty,r_\infty-1}
\end{pmatrix}^{-1}
 \begin{pmatrix}\frac{\alpha_{\infty,r_\infty-3}}{r_\infty-3} \\ \vdots\\ \frac{\alpha_{\infty,2}}{2}\end{pmatrix}
\eeq}
\normalsize{Note} that the r.h.s. does not involve $\alpha_{\infty,1}$ nor $t_{\infty,1}$ so that all $(Q_{\infty,k})_{0\leq k\leq r_\infty-4}$ are independent of $t_{\infty,1}$. Specializing $\delta^{(\boldsymbol{\alpha})}_{\mathbf{t}}$ to each $\delta_{t_{\infty,i}}$ with $i\in \llbracket 2, r_\infty-3\rrbracket$ and using the commutation of lower-triangular Toeplitz matrices and the fact that $t_{\infty,r_\infty-1}=1$ and $t_{\infty,r_\infty-2}=0$ (because $r_\infty\geq 3$), we end up with :
\footnotesize{\bea\label{SysInfty}&&\begin{pmatrix}1&0&\dots&&0\\
0&1& &&0\\
t_{\infty,r_\infty-3}&0&\ddots &\ddots &\vdots\\
\vdots& \ddots &\ddots &\ddots&\vdots \\
t_{\infty,2}&\dots&t_{\infty,r_\infty-3}&0&1
\end{pmatrix}\begin{pmatrix}\frac{1}{r_\infty-4}&0&\dots&0\\
0&\frac{1}{r_\infty-5}& &0\\
\vdots& \ddots &\ddots &\vdots \\
0&\dots&\dots &\frac{1}{1}
\end{pmatrix} \begin{pmatrix} \delta_{t_{\infty,r_\infty-3}}[Q_{\infty,r_\infty-5}]& \dots& \delta_{t_{\infty,2}}[Q_{\infty,r_\infty-5}]\\
\vdots\\
\delta_{t_{\infty,r_\infty-3}}[Q_{\infty,0}]& \dots &\delta_{t_{\infty,4}}[Q_{\infty,0}]\end{pmatrix}\cr 
&&=\begin{pmatrix}\omega&0&\dots&0\\
Q_{\infty,r_\infty-4}&\omega& &0\\
\vdots& \ddots &\ddots &\vdots \\
Q_{\infty,2}&\dots&Q_{\infty,r_\infty-4}&\omega
\end{pmatrix}\begin{pmatrix}\frac{1}{r_\infty-3}&0&\dots&0\\
0&\frac{1}{r_\infty-4}& &0\\
\vdots& \ddots &\ddots &\vdots \\
0&\dots&\dots &\frac{1}{2}
\end{pmatrix}\cr
&&\eea
}
\normalsize{The} first line of \eqref{SysInfty} provides
\beq \frac{1}{r_\infty-4}\delta_{t_{\infty,k}}[Q_{\infty,r_\infty-5}]=\delta_{k,r_\infty-3}\frac{\omega}{r_\infty-3}\,\,,\,\, \forall\, k\in \llbracket 2, r_\infty-3\rrbracket\eeq
i.e.
\beq Q_{\infty,r_\infty-5}=\frac{r_\infty-4}{r_\infty-3}\omega\,t_{\infty,r_\infty-3}+\omega\, u_{\infty,r_\infty-5}\eeq
with $u_{\infty,r_\infty-5}$ independent of the irregular times. The second line of \eqref{SysInfty} provides
\beq\frac{1}{r_\infty-5}\delta_{t_{\infty,k}}[Q_{\infty,r_\infty-6}]=\delta_{k,r_\infty-3}\frac{Q_{\infty,r_\infty-4}}{r_\infty-3}+\delta_{k,r_\infty-4}\frac{\omega}{r_\infty-4} \,\,,\,\, \forall\, k\in \llbracket 2, r_\infty-3\rrbracket
\eeq
i.e.
\beq Q_{\infty,r_\infty-6}=\frac{r_\infty-5}{r_\infty-3}\omega\,u_{\infty,r_\infty-4}t_{\infty,r_\infty-3}+\frac{r_\infty-5}{r_\infty-4}\omega\,t_{\infty,r_\infty-4}+\omega\,u_{\infty,r_\infty-6}\eeq
with $u_{\infty,r_\infty-6}$ independent of the irregular times. The third line of \eqref{SysInfty} provides
\bea t_{\infty,r_\infty-3}\delta_{t_{\infty,r_\infty-3}}\left[\frac{Q_{\infty,r_\infty-5}}{r_\infty-4}\right]+\delta_{t_{\infty,r_\infty-3}}\left[\frac{Q_{\infty,r_\infty-7}}{r_\infty-6}\right]&=&\frac{Q_{\infty,r_\infty-5}}{r_\infty-3}\cr
\delta_{t_{\infty,r_\infty-4}}\left[\frac{Q_{\infty,r_\infty-7}}{r_\infty-6}\right]&=&\frac{Q_{\infty,r_\infty-4}}{r_\infty-4}\cr
\delta_{t_{\infty,r_\infty-5}}\left[\frac{Q_{\infty,r_\infty-7}}{r_\infty-6}\right]&=&\frac{\omega}{r_\infty-5}
\eea
i.e.
\bea Q_{\infty,r_\infty-7}&=&\frac{(r_\infty-6)}{(r_\infty-5)}\omega\,t_{\infty,r_\infty-5}+\frac{(r_\infty-6)}{(r_\infty-4)}\omega\,t_{\infty,r_\infty-4}u_{\infty,r_\infty-4}+\frac{(r_\infty-4)(r_\infty-6)\omega}{2(r_\infty-5)(r_\infty-3)}(t_{\infty,r_\infty-3})^2\cr
&&+\frac{(r_\infty-6)}{(r_\infty-3)}t_{\infty,r_\infty-3}\omega\,u_{\infty,r_\infty-5}-\frac{(r_\infty-6)(r_\infty-4)\omega}{2(r_\infty-5)}(t_{\infty,r_\infty-3})^2+\omega\,u_{\infty,r_\infty-7}\cr
&&
\eea
with $u_{\infty,r_\infty-6}$ independent of the irregular times. It is then obvious that each line will determine the next $Q_{\infty,r_\infty-k}$ in terms of the previous one so that the system \eqref{SysInfty} has a unique solution and it is of the form
\small{\beq \begin{pmatrix}\omega\\ Q_{\infty,r_\infty-4}\\ Q_{\infty,r_\infty-5}\\\vdots\\ \vdots\\ Q_{\infty,0}\end{pmatrix}=\omega\,
\begin{pmatrix}1&0&0&\dots&\dots&0\\
0&1&0&\ddots&&\vdots\\
f^{(\infty)}_{3,1}(t_{\infty,r_\infty-3})&0&1&0&&0\\
f^{(\infty)}_{4,1}(t_{\infty,r_\infty-4},t_{\infty,r_\infty-3})&f^{(\infty)}_{4,2}(t_{\infty,r_\infty-3})&0&1&\ddots&\vdots\\
\vdots& \ddots&&&\ddots&0\\
f^{(\infty)}_{r_\infty-2,1}(t_{\infty,2},\dots ,t_{\infty,r_\infty-3})&\dots&&f^{(\infty)}_{r_\infty-2,r_\infty-4}(t_{\infty,r_\infty-3})&0&1
\end{pmatrix}
\begin{pmatrix}1\\ u_{\infty,r_\infty-4}\\ u_{\infty,r_\infty-5}\\ \vdots\\ \vdots\\ u_{\infty,0}\end{pmatrix}
\eeq}
\normalsize{with} all $\left(u_{\infty,k}\right)_{0\leq k\leq r_\infty-4}$ independent of the irregular times.
The first extra line appears because of our choice of normalization at infinity ($\td{L}_{1,2}=\omega\,\lambda^{r_\infty-3}+O(\lambda^{r_\infty-4})$) that freezes a possible coordinate $Q_{\infty,r_\infty-3}$ to $\omega$. It is also easy to observe from the differential system \eqref{SysInfty} that
\beq f^{(\infty)}_{k,k-2}(t_{\infty,r_\infty-3})=\frac{r_\infty-1-k}{r_\infty-3}t_{\infty,r_\infty-3} \,,\,\, \forall \, k\in \llbracket 3, r_\infty-2\rrbracket\eeq

\section{Proof of Theorem \ref{TheoExplicitDependenceP}}\label{AppendixExplicitDependenceR}
We recall that from \eqref{GaugeTdA11Reduced}, we have
\beq [\td{A}_{\boldsymbol{\alpha}}(\lambda)]_{1,1}=c_{\infty,0}+ \td{L}_{1,1}(\lambda)[A_{\boldsymbol{\alpha}}(\lambda)]_{1,2} \eeq
Thus, using Lemma \ref{LemmatdL11}, the isospectral condition $\delta^{(\boldsymbol{\alpha})}_{\mathbf{t}}[\td{L}_{1,1}(\lambda)]=\partial_\lambda [\td{A}_{\boldsymbol{\alpha}}(\lambda)]_{1,1}$ is identical at finite poles to the case $\delta^{(\boldsymbol{\alpha})}_{\mathbf{t}}[\td{L}_{1,2}(\lambda)]=\partial_\lambda [\td{A}_{\boldsymbol{\alpha}}(\lambda)]_{1,2}$ given by \eqref{RefA12}. We end up with $\delta^{(\boldsymbol{\alpha})}_{\mathbf{t}}R_{X_s,1}=0$ for all $s\in \llbracket 1,n\rrbracket$ and
\footnotesize{\bea\label{SystemXsR} &&\begin{pmatrix}t_{X_s,r_s-1}&0&\dots&0\\
t_{X_s,r_s-2}&t_{X_s,r_s-1}& &\vdots\\
\vdots& \ddots &\ddots &0 \\
t_{X_s,1}&\dots&t_{X_s,r_s-2}&t_{X_s,r_s-1}
\end{pmatrix}
\begin{pmatrix}\frac{1}{r_s-1}&0&\dots&0\\
0&\ddots& &0\\
\vdots& & \ddots&0\\
0&\dots&0 &\frac{1}{1} \end{pmatrix} \begin{pmatrix} \delta_{t_{X_s,r_s-1}}[R_{X_s,r_s}]&\dots& \delta_{t_{X_s,1}}[R_{X_s,r_s}]\\
\vdots&\vdots&\vdots\\
\delta_{t_{X_s,r_s-1}}[R_{X_s,2}]&\dots& \delta_{t_{X_s,1}}[R_{X_s,2}]
\end{pmatrix}\cr
&&=\begin{pmatrix}R_{X_s,r_s}&0&\dots& &\dots &0\\
R_{X_s,r_s-1}&R_{X_s,r_s}& 0& & &\vdots\\
\vdots & \ddots&\ddots &\ddots  & &\vdots\\
\vdots &\ddots&\ddots&\ddots&0&\vdots\\
R_{X_s,3}&\ddots &\ddots&\ddots& R_{X_s,r_s}&0\\
R_{X_s,2}&R_{X_s,3}& \dots & & R_{X_s,r_s-1}& R_{X_s,r_s}
 \end{pmatrix}\begin{pmatrix}\frac{1}{r_s-1}&0&\dots&0\\
0&\ddots& &0\\
\vdots& & \ddots&0\\
0&\dots&0 &\frac{1}{1} \end{pmatrix}
\cr
&&
\eea}
\normalsize{For} $r_\infty=2$ we may choose the variables $\left(v_{X_s,1}\right)_{1\leq s\leq n}$ so that $\underset{s=1}{\overset{n}{\sum}} v_{X_s,1}=-t_{\infty,0}$ since $-t_{\infty,0}$ does not depend on the irregular times. For $r_\infty=1$, the same remark applies too. For $r_\infty=1$, the second additional constraint \eqref{AddRRelations} is also compatible and from \eqref{betaMinus1} it provides the explicit dependence of $\beta_{-1}$ in terms of the irregular times in order to satisfy the isospectral condition.

\medskip

The situation is a little more complicated at infinity. Let us first observe that the study at infinity is only needed for $r_\infty\geq 4$ and in this case, we have $c_{\infty,0}= \frac{1}{2\omega}\mathcal{L}_{\boldsymbol{\alpha}}[\omega]=0$ from Section \ref{Appendixcinfty0} and Proposition \ref{PropQCondition}. Thus, the isospectral condition reduces for $r_\infty\geq 4$ to
\beq  [\td{A}_{\boldsymbol{\alpha}}(\lambda)]_{1,1}=\td{L}_{1,1}(\lambda)[A_{\boldsymbol{\alpha}}(\lambda)]_{1,2} \eeq
This relation is analogous to $[\td{A}_{\boldsymbol{\alpha}}(\lambda)]_{1,2}=\td{L}_{1,2}(\lambda)[A_{\boldsymbol{\alpha}}(\lambda)]_{1,2}$ that has been used to derive Proposition \ref{PropQCondition}. The only difference is that $\td{L}_{1,1}(\lambda)=-\lambda^{r_\infty-2}+O(\lambda^{r_\infty-4})$ while $\td{L}_{1,2}(\lambda)=\omega \lambda^{r_\infty-3}+O(\lambda^{r_\infty-4})$. However, it is obvious that the computations made for \eqref{EqEQ} can easily be adapted to this case. Using Lemma \ref{LemmatdL11}, equation \eqref{EqEQ} is adapted into 
\bea \label{eqEQ2}\sum_{k=0}^{r_\infty-4} \delta^{(\boldsymbol{\alpha})}_{\mathbf{t}}[R_{\infty,k}]\lambda^k&=&-(r_\infty-3)\nu_{\infty,1}^{(\boldsymbol{\alpha})}\lambda^{r_\infty-4}\delta_{r_\infty\geq 4}-(r_\infty-4)\nu_{\infty,2}^{(\boldsymbol{\alpha})}\lambda^{r_\infty-5}\delta_{r_\infty\geq 5}\cr
&&+\sum_{j=0}^{r_\infty-6}(m+1)\left(-\nu_{\infty,r_\infty-3-j}^{(\boldsymbol{\alpha})}+\sum_{k=j+2}^{r_\infty-4}\nu_{\infty,k-1-j}^{(\boldsymbol{\alpha})}R_{\infty,k}\right)\lambda^{j}+O(\lambda^{-1})\cr
&&\eea
Indeed, for $r_\infty\geq 3$, we have $t_{\infty,r_\infty-1}=1$ and $t_{\infty,r_\infty-2}=0$ and there is no deformation relatively to $t_{\infty,r_\infty-1}$ nor $t_{\infty,r_\infty-2}$ (i.e. $\alpha_{\infty,r_\infty-1}=\alpha_{\infty,r_\infty-2}=0$). Using Proposition \ref{nus}, this system may be rewritten as a $(r_\infty-3)\times(r_\infty-3)$ system:
\footnotesize{\bea\label{SystemInftyR} &&\begin{pmatrix}1&0&\dots&\dots&0\\
0&1& &&\vdots\\
t_{\infty,r_\infty-3}&0&1& &\vdots\\
\vdots& \ddots &\ddots &\ddots&0 \\
t_{\infty,3}&\dots&t_{\infty,r_\infty-3}&0&1 
\end{pmatrix}
\begin{pmatrix}\frac{1}{r_\infty-3}&0&\dots&0\\
0&\ddots& &0\\
\vdots& & \ddots&0\\
0&\dots&0 &\frac{1}{1} \end{pmatrix} \begin{pmatrix} \delta_{t_{\infty,r_\infty-3}}[R_{\infty,r_\infty-4}]&\dots& \delta_{t_{\infty,2}}[R_{\infty,r_\infty-4}]\\
\vdots&\vdots&\vdots\\
\delta_{t_{\infty,r_\infty-3}}[R_{\infty,0}]&\dots& \delta_{t_{\infty,2}}[R_{\infty,0}]
\end{pmatrix}\cr
&&=\begin{pmatrix}-1&0&\dots&\dots& \dots&\dots &0\\
0&-1& 0&& & &\vdots\\
R_{\infty,r_\infty-4}&0& -1& && &\vdots\\
\vdots & \ddots&\ddots &\ddots&\ddots  & &\vdots\\
\vdots &\ddots&\ddots&\ddots&-1&0&\vdots\\
R_{\infty,r_\infty,3}&\ddots &\ddots&\ddots&\ddots& -1&0\\
R_{\infty,r_\infty,2}&R_{\infty,r_\infty,3}& \dots &\dots&R_{\infty,r_\infty-4}& 0& -1
 \end{pmatrix}\begin{pmatrix}\frac{1}{r_\infty-3}&0&\dots&0\\
0&\ddots& &0\\
\vdots& & \ddots&0\\
0&\dots&0 &\frac{1}{1} \end{pmatrix}
\cr
&&
\eea}
\normalsize{It is obvious} that $R_{\infty,k}=-t_{\infty,k+1}$ for all $k\in \llbracket 0,r_\infty-4\rrbracket$ provides a solution of the system.  Straightforward computations similar to Section \ref{SectionInfinity} provide the expressions presented in Theorem \ref{TheoExplicitDependenceP}.

\section{Example of the second element of the Painlev\'{e} $2$ hierarchy}\label{AppendixP2}
In this appendix, we propose to apply our results to the case of the second element of the Painlev\'{e} $2$ hierarchy corresponding to $n=0$ (i.e. only one pole at infinity) and $r_\infty=5$ giving rise to a spectral curve of genus $g=2$. In \cite{MarchalOrantinAlameddine2022}, the derivation of the Lax pairs and Hamiltonian system was done and we shall only reproduce the results there. In \cite{MarchalOrantinAlameddine2022}, it is proved that using M\"{o}bius transformations, one may  fix the values of the leading times at infinity without changing the symplectic structure. Thus, we shall take
\beq t_{\infty,4}=1=\,,\,t_{\infty,3}=0
\eeq
Since we are in the case where $r_\infty=5>1$, we shall also take $\omega=1$ for simplicity following the convention of \cite{MarchalOrantinAlameddine2022} because it does not play any role from Proposition \ref{PropQCondition}. There are two apparent singularities $\mathbf{q}=(q_1,q_2)$ as well as two dual coordinates $\mathbf{p}=(p_1,p_2)$ on the spectral curve. There are also two spectral invariants $(H_{\infty,0},H_{\infty,1})$ that can be expressed using these Darboux coordinates as
\bea H_{\infty,0}&=&\frac{q_1p_2^2-q_2p_1^2}{q_1-q_2}-
\frac{p_1-p_2}{q_1-q_2} +q_1q_2\Big(q_1^4+q_2^4+q_1^3q_2+q_2^3q_1+q_1^2q_2^2\cr
&&+2(q_1^2+q_1q_2+q_2^2)t_{\infty,2}+(q_1+q_2)\tau_{\infty,1}+t_{\infty,2}+2t_{\infty^{(1)},0}-1
\Big),\cr
H_{\infty,1}&=&\frac{p_1^2-p_2^2}{q_1-q_2}-(q_1+q_2)\left((q_1^2+q_2^2+t_{\infty,2})^2-q_1^2q_2^2+2t_{\infty,0}-1
\right)-2(q_1^2+q_1q_2+q_2^2)t_{\infty,1}\cr
&&
\eea
Solving the compatibility equations provide two Hamiltonians in the two remaining directions $\partial_{t_{\infty,1}}$ and $\partial_{t_{\infty,2}}$ in the tangent space (note that in \cite{MarchalOrantinAlameddine2022}, the Hamiltonians are given relatively to $\frac{1}{2}\tau_{\infty,2}$ and $\frac{1}{2}\tau_{\infty,1}$ so there is an extra factor $\frac{1}{2}$ in the Hamiltonians).
\beq\text{Ham}^{(\boldsymbol{\alpha}_{{\infty,1}})}(q_1,q_2,p_1,p_2)=H_{\infty,1} \,\,,\,\, \text{Ham}^{(\boldsymbol{\alpha}_{{\infty,2}})}(q_1,q_2,p_1,p_2)=\frac{1}{2}H_{\infty,0}
\eeq
The Lax matrix is given by 
\small{\bea \td{L}_{1,1}&=&-\lambda^3+ \left(\frac{p_1-p_2}{q_1-q_2}+q_1^2+q_2^2+q_1q_2\right)\lambda+ \frac{p_2q_1-p_1q_2}{q_1-q_2}-q_1q_2(q_1+q_2),\cr
\td{L}_{1,2}&=&  (\lambda-q_1)(\lambda-q_2),\cr
\td{L}_{2,2}&=&-\td{L}_{1,1},\cr
\td{L}_{2,1}&=&2\left(\frac{p_1-p_2}{q_1-q_2}+q_1^2+q_1q_2+q_2^2+t_{\infty,2}\right)\lambda^2+2\left(\frac{p_1q_1-p_2q_2}{q_1-q_2}+(q_1+q_2)(q_1^2+q_2^2+t_{\infty,2})+t_{\infty,1}\right)\lambda\cr
&&-\frac{(p_1-p_2)^2}{(q_1-q_2)^2}-\frac{2(p_1q_2-p_2q_1)(q_1+q_2)}{q_1-q_2}+q_1^4+q_2^4-q_1^2q_2^2+ (q_1^2+q_1q_2+q_2^2)t_{\infty,2}\cr
&&+ 2(q_1+q_2)t_{\infty,1}+t_{\infty,2}^2+2t_{\infty,0}
\eea}
The auxiliary matrices $(\td{A}_{t_{\infty,1}},\td{A}_{t_{\infty,2}})$ can also be found in Appendix $8.7$ (up to the factor $2$ mentioned above) of \cite{MarchalOrantinAlameddine2022} and are given by
\bea \td{A}_{\boldsymbol{\alpha}_{\infty,1}}&=&\begin{pmatrix}-\lambda-q_1-q_2& 1\\
\frac{2(p_1+p_2)}{q_1-q_2}+2(q_1^2+q_2^2+q_1q_2+t_{\infty,2})&\lambda+q_1+q_2\end{pmatrix}\cr
\left[\td{A}_{\boldsymbol{\alpha}_{\infty,2}}\right]_{1,1}&=& -\frac{1}{2}\lambda^2+ \frac{p_1+p_2}{2(q_1-q_2)}+\frac{1}{2}(q_1+q_2)^2\cr
\left[\td{A}_{\boldsymbol{\alpha}_{\infty,2}}\right]_{1,2}&=&\frac{\lambda-q_1-q_2}{2}\cr
\left[\td{A}_{\boldsymbol{\alpha}_{\infty,2}}\right]_{2,1}&=&\left( \frac{p_1+p_2}{q_1-q_2}+q_1^2+q_1q_2+q_2^2+t_{\infty,2}\right)\lambda+\frac{p_1q_1-p_2q_2}{q_1-q_2}+q_1^3+q_2^3+q_1^2q_2+q_1q_2^2\cr &&+(q_1+q_2)t_{\infty,2}+t_{\infty,1}\cr
\left[\td{A}_{\boldsymbol{\alpha}_{\infty,2}}\right]_{2,2}&=&-\left[\td{A}_{\boldsymbol{\alpha}_{\infty,2}}\right]_{1,1}
\eea

\medskip

Let us now perform the various change of Darboux coordinates presented in this article to obtain the isospectral coordinates. Let us first note that $M_\infty(\mathbf{t})=I_2$ due to the fact that $t_{\infty,4}=1$ and $t_{\infty,3}=0$. In this case geometric Darboux coordinates $(Q_{\infty,0},Q_{\infty,1},P_{\infty,0},P_{\infty_1})$ of Definition \ref{DefGeometricCoordinates} are given by
\beq Q_{\infty,1}=-(q_1+q_2) \,\,,Q_{\infty,0}= q_1q_2\,,\,\, 
P_{\infty,0}=-\frac{p_1-p_2}{(q_1-q_2)} \,\,,\,\, P_{\infty,1}=-\frac{p_1q_1-p_2q_2}{(q_1-q_2)}\eeq
i.e.
\bea
p_1&=& P_{\infty,0} q_2- P_{\infty,1}\,,\, p_2= P_{\infty,0} q_1- P_{\infty,1}\cr
q_1&=&\frac{-Q_{\infty,1}-\sqrt{Q_{\infty,1}^2-4 Q_{\infty,0}}}{2}\,,\,q_2=\frac{-Q_{\infty,1}+\sqrt{Q_{\infty,1}^2-4 Q_{\infty,0}}}{2} 
\eea
The Hamiltonians are given by
\bea &&\text{Ham}^{(\boldsymbol{\alpha}_{{\infty,1}})}(Q_{\infty,0},Q_{\infty,1},P_{\infty,0},P_{\infty,1})=2P_{\infty,0}P_{\infty,1}^2+P_{\infty,0}^2 Q_{\infty,1} + Q_{\infty,1}^5 
-4Q_{\infty,0}Q_{\infty,1}^3\cr&&+2t_{\infty,2}Q_{\infty,1}^3 -2t_{\infty,1}Q_{\infty,1}^2+3Q_{\infty,0}^2 Q_{\infty,1}
-4Q_{\infty,0}Q_{\infty,1}t_{\infty,2} 
+Q_{\infty,1}\left(t_{\infty,2}^2+2t_{\infty,0}-1\right)+2t_{\infty,1}Q_{\infty,0}\cr
&&\text{Ham}^{(\boldsymbol{\alpha}_{{\infty,2}})}(Q_{\infty,0},Q_{\infty,1},P_{\infty,0},P_{\infty,1})=\frac{1}{2}P_{\infty,1}^2+Q_{\infty,1}P_{\infty,0}P_{\infty,1}
-\frac{1}{2}Q_{\infty,0}P_{\infty,0}^2+\frac{1}{2}Q_{\infty,1}^2P_{\infty,0}^2\cr&&+ \frac{1}{2}P_{\infty,0}
+\frac{1}{2}Q_{\infty,0}^3-\frac{3}{2}Q_{\infty,1}^2Q_{\infty,0}^2-t_{\infty,2}Q_{\infty,0}^2+\frac{1}{2}Q_{\infty,0}Q_{\infty,1}^4
+t_{\infty,2}Q_{\infty,0}Q_{\infty,1}^2\cr&& -t_{\infty,1}Q_{\infty,0}Q_{\infty,1}+\frac{1}{2}Q_{\infty,0}\left(t_{\infty,2}^2+2t_{\infty,0}-1\right)
\eea
The corresponding Lax matrix is
\bea\label{LaxExample} \td{L}_{1,1}(\lambda)&=& -\lambda^3+\left(Q_{\infty,1}^2-Q_{\infty,0}-P_{\infty,0}\right)\lambda-P_{\infty,0}Q_{\infty,1}-Q_{\infty,0}Q_{\infty,1}-P_{\infty,1}\cr
\td{L}_{1,2}(\lambda)&=&\lambda^2+Q_{\infty,1}\lambda+Q_{\infty,0}\cr
\td{L}_{2,1}(\lambda)&=&2\left(Q_{\infty,1}^2-P_{\infty,0}-Q_{\infty,0}+t_{\infty,2}\right)\lambda^2+2\left(-P_{\infty,1} -Q_{\infty,1}^3 +2Q_{\infty,0}Q_{\infty,1}-t_{\infty,2}Q_{\infty,1}+t_{\infty,1}\right)\lambda \cr&&
+Q_{\infty,1}^4+2P_{\infty,0}Q_{\infty,1}^2-4Q_{\infty,0}Q_{\infty,1}^2+2t_{\infty,2}Q_{\infty,1}^2-2t_{\infty,1}Q_{\infty,1} +2P_{\infty,1}Q_{\infty,1}+t_{\infty,2}^2\cr&&
-2t_{\infty,2}Q_{\infty,0} -P_{\infty,0}^2+Q_{\infty,0}^2+2t_{\infty,0}\cr
\td{L}_{2,2}(\lambda)&=&-\td{L}_{1,1}(\lambda)
\eea
Note that we have $\text{Ham}^{(\boldsymbol{\alpha}_{{\infty,1}})}(\mathbf{Q},\mathbf{P})=H_{\infty,1}$ and $\text{Ham}^{(\boldsymbol{\alpha}_{{\infty,2}})}(\mathbf{Q},\mathbf{P})=\frac{1}{2}H_{\infty,0}$ which is coherent with Theorem \ref{TheoHamNew}.

\medskip
Let us now perform the change of Darboux coordinates $(\mathbf{Q},\mathbf{P})\to (\mathbf{Q},\mathbf{R})$ using Definition \ref{DefR} and the fact that $g_0=-Q_{\infty,1}$ from Theorem \ref{GeoLaxMatrices}.
\bea R_{\infty,1}&=& -P_{\infty,0}-Q_{\infty,0} +Q_{\infty,1}^2\cr
R_{\infty,0}&=& -P_{\infty,1}-P_{\infty,0}Q_{\infty,1}+ Q_{\infty,0}Q_{\infty,1}
\eea
i.e.
\bea P_{\infty,0}&=&Q_{\infty,1}^2-Q_{\infty,0}-R_{\infty,1}\cr
P_{\infty,1}&=&-Q_{\infty,1}^3+2Q_{\infty,0}Q_{\infty1}+Q_{\infty,1}R_{\infty,1}-R_{\infty,0}
\eea
The Lax matrix in these coordinates reads:
\bea \td{L}_{1,1}(\lambda)&=& -\lambda^3+R_{\infty,1}\lambda+R_{\infty,0}\cr
\td{L}_{1,2}(\lambda)&=&\lambda^2+Q_{\infty,1}\lambda+Q_{\infty,0}\cr
\td{L}_{2,1}(\lambda)&=&2(R_{\infty,1}+t_{\infty,2})\lambda^2+2\left(-R_{\infty,1}Q_{\infty,1}-t_{\infty,2}Q_{\infty,1}+R_{\infty_0}+t_{\infty,1}\right)\lambda\cr&&
-R_{\infty,1}^2+2Q_{\infty,1}^2 R_{\infty,1}-2Q_{\infty,0}R_{\infty,1}-2R_{\infty,0}Q_{\infty,1}+2t_{\infty,2}Q_{\infty,1}^2-2\tau_{\infty,1}Q_{\infty,1}+(t_{\infty,2})^2\cr&&
-2t_{\infty,2}Q_{\infty,0}+2t_{\infty,0}\cr
\td{L}_{2,2}(\lambda)&=&-\td{L}_{1,1}(\lambda)
\eea
in agreement with Theorem \ref{GeoLaxMatricesQR}.

We can compute the quantity $I_{\infty,1}$ and $I_{\infty,2}$ from the asymptotics expansion at $\lambda\to \infty$ of \eqref{LaxExample}. We find
\bea I_{\infty,1}&=&P_{\infty,0}P_{\infty,1}+t_{\infty,1}Q_{\infty,0}+\frac{1}{2}Q_{\infty,1}^5 +t_{\infty,2}Q_{\infty,1}^3 -2Q_{\infty,0}Q_{\infty,1}^3-t_{\infty,1}Q_{\infty,1}^2+\frac{1}{2}P_{\infty,0}^2Q_{\infty,1}\cr&&+\frac{3}{2}Q_{\infty,0}^2Q_{\infty,1}
-2t_{\infty,2}Q_{\infty,0}Q_{\infty,1}+\frac{1}{2}t_{\infty,2}^2Q_{\infty,1}+t_{\infty,0}Q_{\infty,1}-t_{\infty,1}t_{\infty,2}\cr
I_{\infty,2}&=&\frac{1}{4}P_{\infty,1}^2 +\frac{1}{4}P_{\infty,0}^2Q_{\infty,1}^2
+\frac{1}{2}P_{\infty,0}P_{\infty,1}Q_{\infty,1}- P_{\infty,0}^2Q_{\infty,0}
+\frac{1}{4}Q_{\infty,0}^3-
\frac{3}{4}Q_{\infty,1}^2Q_{\infty,0}^2\cr&&-\frac{1}{2}t_{\infty,2}Q_{\infty,0}^2+\frac{1}{4}Q_{\infty,1}^4Q_{\infty,0}+\frac{1}{2}Q_{\infty,1}^2Q_{\infty,0}-\frac{1}{2}t_{\infty,1}Q_{\infty,1}Q_{\infty,0}+\frac{1}{4}t_{\infty,2}^2Q_{\infty,0}\cr&&+\frac{1}{2}t_{\infty,0}Q_{\infty,0}-\frac{1}{4}t_{\infty,1}^2-\frac{1}{2}t_{\infty,2}t_{\infty,0}
\eea
from which one can easily verify using \eqref{LaxExample} Theorem \ref{TheoHamSpectral}.

\medskip

Let us finally obtain the isospectral coordinates $(\mathbf{u},\mathbf{v})$. Theorem \ref{ExplicitDependenceQ} is equivalent in this example to
\beq Q_{\infty,1}=u_{\infty,1} \,\,,\,\, Q_{\infty,0}=f_{3,1}(t_{\infty,2})+u_{\infty,0}= u_{\infty,0}+\frac{1}{2}t_{\infty,2} \eeq
Similarly, Theorem \ref{TheoExplicitDependenceP} is equivalent in this example to
\beq R_{\infty,1}=-t_{\infty,2}+ v_{\infty,1}\,\,,\,\, R_{\infty,0}=-t_{\infty,1}+v_{\infty,0}\eeq
The Lax matrix in these coordinates is given by
\bea
\td{L}_{1,1}(\lambda)&=&-\lambda^3 +(v_{\infty,1}-t_{\infty,2})\lambda +v_{\infty,0}-t_{\infty,1}\cr
\td{L}_{1,2}(\lambda)&=&\lambda^2+t_{\infty,1}\lambda+ u_{\infty,0}+\frac{1}{2}t_{\infty,2}\cr
\td{L}_{2,1}(\lambda)&=&2 v_{\infty,1}\lambda^2+2(v_{\infty,0}-u_{\infty,1}v_{\infty,1})\lambda -v_{\infty,1}^2-2u_{\infty,1}v_{\infty,0}+2u_{\infty,1}^2v_{\infty,1}\cr&&
-2u_{\infty,0}v_{\infty,1}+t_{\infty,2}v_{\infty,1}+2t_{\infty,0}\cr
\td{L}_{2,2}(\lambda)&=&-\td{L}_{1,1}(\lambda)
\eea
and
\bea 
\td{A}_{\boldsymbol{\alpha}_{\infty,1}}&=&\begin{pmatrix}
    -\lambda+u_{\infty,1}&1 \\
2v_{\infty,1}& \lambda- u_{\infty,1}
\end{pmatrix}\cr
\td{A}_{\boldsymbol{\alpha}_{\infty,2}}&=&\begin{pmatrix}
    -\frac{1}{2}\lambda^2+\frac{1}{2}(u_{\infty,0}+v_{\infty,1})-\frac{t_{\infty,2}}{4}& \frac{\lambda+ u_{\infty,1}}{2}\\
    (\lambda-u_{\infty,1})v_{\infty,1}+v_{\infty,0}& \frac{1}{2}\lambda^2-\frac{1}{2}(u_{\infty,0}+v_{\infty,1})+\frac{t_{\infty,2}}{4}
\end{pmatrix}
\eea
from which one can easily check that the isospectral condition $\delta_{t_i}\td{L}(\lambda)=\partial_\lambda \td{A}_{\boldsymbol{\alpha}_{\infty,i}}(\lambda)$ is verified for $i\in \{1,2\}$. Note that the coordinates $(\mathbf{u},\mathbf{v})$ are not canonical in this case. The time-independent change of coordinates to obtain a canonical form is given by
\beq x_2=u_{\infty,0}\,,\, x_3=u_{\infty,1},\,,\, y_2=v_{\infty,1}\,,\, y_3=v_{\infty,0}-u_{\infty,1}v_{\infty,1}\eeq
where the notations are inspired by Appendix $A$ of \cite{BertolaHarnadHurtubise2022} (the present case corresponds to $x_1=-\frac{1}{\sqrt{2}}$ in their setting). With these coordinates, the symplectic one-form is $\omega=y_2dx_2+y_3dx_3$ so that
\beq \partial_{t_{\infty,i}}x_j=\frac{\partial\text{Ham}^{(\boldsymbol{\alpha}_{\infty,i})}(x_2,x_3,y_2,y_3) }{\partial y_j}\,\,,\,\, \partial_{t_{\infty,i}}y_j=-\frac{\partial\text{Ham}^{(\boldsymbol{\alpha}_{\infty,i})}(x_2,x_3,y_2,y_3) }{\partial x_j} \,\,\,,\,\forall\, (i,j)\in \{1,2\}\times\{2,3\}.\eeq
The corresponding Lax matrices in these coordinates are
\small{\bea \td{L}(\lambda)&=&\begin{pmatrix}
    -\lambda^3+(y_2-t_{\infty,2})\lambda+x_2y_2+y_3-t_{\infty,1}&\lambda^2+x_3\lambda+x_2+\frac{t_{\infty,2}}{2}\\
2y_2\lambda^2+2y_3\lambda-y_2^2-(2x_2-t_{\infty,2}y_2-2x_3y_3+2t_{\infty,0}    & \lambda^3-(y_2-t_{\infty,2})\lambda-x_2y_2-y_3+t_{\infty,1}
\end{pmatrix}\cr
\td{A}_{\boldsymbol{\alpha}_{\infty,1}}(\lambda)&=&\begin{pmatrix}
    -\lambda+x_2&1\\ 2y_2&\lambda-x_3
\end{pmatrix}\cr
\td{A}_{\boldsymbol{\alpha}_{\infty,2}}(\lambda)&=&\begin{pmatrix}
-\frac{1}{2}\lambda^2+\frac{1}{2}y_2+\frac{1}{2}x_2-\frac{1}{4}t_{\infty,2}&
\frac{1}{2}(\lambda+x_3)\\
y_2\lambda+y_3& \frac{1}{2}\lambda^2-\frac{1}{2}y_2-\frac{1}{2}x_2+\frac{1}{4}t_{\infty,2}
\end{pmatrix}
\eea}
\normalsize{and} the Hamiltonians identify with the spectral invariants:
\bea \text{Ham}^{(\boldsymbol{\alpha}_{\infty,1})}(x_2,x_3,y_2,y_3)&=&-2I_{\infty,1}\cr
&=&2x_2^2y_3-y_2^2x_3+2x_2x_3y_2 -2y_2y_3-2x_2y_3+t_{\infty,2}x_3y_2+t_{\infty,2}y_3\cr&&
+2t_{\infty,1}y_2-2t_{\infty,0}x_3\cr
\text{Ham}^{(\boldsymbol{\alpha}_{\infty,2})}(x_2,x_3,y_2,y_3)&=&-2I_{\infty,2}\cr
&=&-\frac{1}{2}y_2^3-\frac{1}{2}x_3^2y_2^2+\frac{1}{4}t_{\infty,2}y_2^2+\frac{1}{2}x_2y_2^2
+t_{\infty,1}x_3y_2-\frac{1}{4}t_{\infty,2}^2y_2+x_2^2y_2\cr&&-x_3y_2y_3
+\frac{1}{2}t_{\infty,2}x_3y_3+x_2x_3y_3+t_{\infty,1}y_3-t_{\infty,0}x_2\cr&&
\eea
Part of the computations for this example have been done using Maple software. The corresponding Maple sheet is available at 
\url{http://math.univ-lyon1.fr/~marchal/AdditionalRessources/index.html}.

\addcontentsline{toc}{section}{References}
\bibliographystyle{plain}
\bibliography{Biblio}

\end{document}